\documentclass[12pt]{article}

\usepackage[margin=1in,
  includefoot,
  headsep=15pt 
]{geometry}
\usepackage{graphicx,psfrag,verbatim}
\usepackage{amsmath,amssymb,amsthm}
\usepackage{amsfonts,dsfont,bm}
\usepackage{algorithm,algorithmic}
\usepackage[small,bf]{caption}
\usepackage{subcaption} 
\usepackage[dvipsnames]{xcolor}
\usepackage{hyperref}
\hypersetup{
    colorlinks=true,
    linkcolor=Red,     
    urlcolor=Blue,
    citecolor={Green},
}
\usepackage{url}
\usepackage[toc,page]{appendix}
\usepackage{cases}
\usepackage[shortlabels]{enumitem}
\usepackage{fancyhdr}
\usepackage{multirow}
\usepackage{booktabs}  
\usepackage{adjustbox}
\usepackage{secdot}
\allowdisplaybreaks

\def\suda{{\bfseries \footnotesize SUDA}}
\newtheorem{theorem}{{Theorem}}
\newtheorem{lemma}{{Lemma}}
\newtheorem{corollary}{{Corollary}}
\newtheorem{assumption}{{ Assumption}}
\newtheorem{remark}{{Remark}}

\usepackage{cite}
\DeclareMathOperator*{\minimize}{minimize}

\newcommand{\grad}{{\nabla}}   
\newcommand{\zero}{\mathbf{0}}  
\newcommand{\one}{\mathbf{1}}   
\newcommand{\real}{\mathbb{R}}  

\newcommand{\diag}{\mathrm{diag}}  
\newcommand{\bdiag}{\mathrm{blkdiag}}  
\newcommand{\col}{\mathrm{col}}     

\newcommand{\Ex}{\mathds{E\hspace{0.05cm}}}  

\def\tran{^{\mathsf{T}}}  
\newcommand{\qd}{\hfill{$\square$}}
\newcommand{\define}{\;\stackrel{\Delta}{=}\;} 

\newcommand{\eg}{{\it e.g.}}
\newcommand{\ie}{{\it i.e.}}

\def\bxi        {{\boldsymbol \xi}}

\def\A{{\mathbf{A}}}
\def\B{{\mathbf{B}}}
\def\C{{\mathbf{C}}}
\def\D{{\mathbf{D}}}

\def\F{{\mathbf{F}}}
\def\G{{\mathbf{G}}}

\def\I{{\mathbf{I}}}

\def\P{{\mathbf{P}}}

\def\U{{\mathbf{U}}}
\def\V{{\mathbf{V}}}
\def\W{{\mathbf{W}}}

\def\a{{\mathbf{a}}}
\def\b{{\mathbf{b}}}

\def\e{{\mathbf{e}}}
\def\f{{\mathbf{f}}}
\def\g{{\mathbf{g}}}

\def\s{{\mathbf{s}}}

\def\u{{\mathbf{u}}}

\def\w{{\mathbf{w}}}
\def\x{{\mathbf{x}}}
\def\y{{\mathbf{y}}}
\def\z{{\mathbf{z}}}

\newcommand{\cF}{{\mathcal{F}}}

\newcommand{\cN}{{\mathcal{N}}}

\newcommand{\cU}{{\mathcal{U}}}

\pdfoutput=1

\title{\bfseries A Unified and Refined Convergence Analysis for Non-Convex Decentralized Learning}
\author{Sulaiman A. Alghunaim \\
\small Kuwait University \\
\texttt{\small sulaiman.alghunaim@ku.edu.kw}
\and Kun Yuan \\
\small DAMO Academy, Alibaba Group \\
\texttt{\small kun.yuan@alibaba-inc.com} \vspace{5mm}\\
}
\begin{document}

\maketitle
                                                        
\begin{abstract}
We study the  consensus decentralized optimization problem where the objective function is the average  of $n$ agents private non-convex cost functions; moreover, the agents can only communicate to their neighbors on a given network topology.  The  stochastic learning setting is considered in this paper where each agent can only access a noisy estimate of its gradient.  Many decentralized methods can solve such problem including EXTRA, Exact-Diffusion/D$^2$, and gradient-tracking. Unlike the famed \textsc{Dsgd} algorithm, these methods have been shown to be  robust to the  heterogeneity across the local cost functions.  However, the established convergence rates for these methods  indicate that their sensitivity to the network topology is worse than \textsc{Dsgd}. Such theoretical results imply that these methods can perform much worse than \textsc{Dsgd} over sparse networks, which, however, contradicts empirical experiments where \textsc{Dsgd} is observed to be more sensitive to the network topology.

In this work, we study a general \underline{s}tochastic \underline{u}nified \underline{d}ecentralized \underline{a}lgorithm (\suda) that includes the above  methods as special cases. We establish the convergence of \suda~under both non-convex and the Polyak-\L{}ojasiewicz condition settings. Our results provide improved network topology dependent bounds for these methods (such as Exact-Diffusion/D$^2$ and gradient-tracking) compared with existing literature.  Moreover, our results show that these methods are often less sensitive to the network topology compared to \textsc{Dsgd}, which agrees with numerical experiments.

\end{abstract}


\section{Introduction}
In a distributed multi-agent optimization problem, the inputs  (\eg, functions, variables, data) are spread over multiple computing agents (\eg, nodes, processors) that are connected over some network, and the agents are required to communicate with each other to solve this problem. Distributed  optimization have attracted a lot of attention due to the need of developing efficient methods to solve large-scale optimization problems \cite{boyd2011admm} such as in deep neural networks applications \cite{li2014scaling}. Decentralized optimization methods are algorithms where the agents seek to find a solution through local interactions  (dictated by the network connection) with their neighboring agents.     Decentralized methods have several advantages over centralized methods, which require all the agents to communicate with a central coordinator, such as their robustness to failure and privacy. Moreover, decentralized methods have been  shown to enjoy lower communication cost  compared to centralized methods under certain practical scenarios \cite{lian2017can,assran2019stochastic,chen2021accelerating}.

 In this work, we consider a network (graph) of  $n$ collaborative agents (nodes) that are interested in solving the following distributed stochastic optimization problem: 
\begin{equation} \label{min_learning_prob}
 \begin{aligned} 
 \minimize_{x \in \real^d} \quad  f(x)=\frac{1}{n} \sum_{i=1}^n f_i(x), \quad f_i(x)\define\Ex_{\xi_i} [F_i(x;\xi_i)]. 
\end{aligned}
\end{equation}
In the above formulation, $F_i:\real^d \rightarrow \real$ is a smooth non-convex function privately known by agent $i$. The notation $\Ex_{\xi_i}$ is the expected value of the random variable $\xi_i$ (\eg, random data samples) taken with respect to some local distribution. The above formulation is known as the consensus formulation  since the agents share a common variable, which they need to agree upon \cite{boyd2011admm}. We consider {\em decentralized methods} where the agents aim to find a solution of \eqref{min_learning_prob} through local interactions  (each agent can only send and receive information to its immediate neighbors).

Two important measures of the performance of distributed (or decentralized) methods are the {linear speedup} and { transient time}. A decentralized  method is said to achieve {\em linear speedup} if the gradient computational complexity needed to reach certain accuracy reduces linearly with the network size $n$. The {\em transient time} of a distributed method is the number of iterations needed  to achieve linear speedup.   A  common method  for solving problem \eqref{min_learning_prob} is the decentralized/distributed stochastic gradient descent (\textsc{Dsgd}) method \cite{ram2010distributed,cattivelli2010diffusion}, where each agent employs a local stochastic gradient descent update and a local gossip step (there are several variations  based on the order of the gossip step such as diffusion or consensus methods \cite{cattivelli2010diffusion
,chen2013distributed,nedic2009distributed}).  \textsc{Dsgd} is simple to implement; moreover, it has been shown to achieve linear speedup asymptotically \cite{lian2017can}. This implies that the convergence rate of \textsc{Dsgd} asymptotically achieves the same network independent rate as  the centralized (also known as parallel) stochastic gradient descent (\textsc{Psgd}) with a central coordinator.   While being attractive, \textsc{Dsgd} suffers from an error or bias term caused by the heterogeneity between the local cost functions minimizers (\eg, heterogeneous data distributions across the agents) \cite{chen2013distributed,yuan2016convergence}. The existence of such  bias term will slow down the convergence of \textsc{Dsgd}, and hence enlarge its transient time.

    Several bias-correction methods have been proposed to remove the bias of \textsc{Dsgd}  such as EXTRA \cite{shi2015extra}, Exact-Diffusion (ED) (a.k.a D$^2$ or NIDS)  \cite{yuan2019exactdiffI,li2017nids,yuan2020influence,tang2018d}, and gradient-tracking (GT) methods  \cite{xu2015augmented,di2016next,qu2017harnessing,nedic2017achieving}. Although these methods have been extensively studied, their convergence properties have not been fully understood as we now explain. Under convex stochastic settings,  ED/D$^2$ is theoretically shown to improve upon the transient time of \textsc{Dsgd}   \cite{yuan2021removing,huang2021improving}, especially under sparse topologies.  However,  existing non-convex results imply that  ED/D$^2$ has worse transient time compared to \textsc{Dsgd} for sparse networks  \cite{tang2018d}.  Moreover,  the transient time of GT-methods are theoretically worse than \textsc{Dsgd} under sparse networks even under convex settings  \cite{pu2021distributed}. These existing theoretical results imply that under non-convex settings, bias-correction methods can suffer from worse transient time compared to \textsc{Dsgd}. However, empirical results suggest that both ED/D$^2$ and GT methods outperform \textsc{Dsgd} under sparse topologies (without any acceleration) \cite{lu2021optimal,xin2021improved,tang2018d,xin2021fast}.  This phenomenon is yet to be explained.

In this work, we provide a novel unified convergence analysis  of several  decentralized  bias-correction methods  including both ED/D$^2$ and GT methods under {\em non-convex} settings. We establish refined and improved convergence rate bounds over existing results. Moreover, our results show that bias-correction methods such as  Exact-Diffusion/D$^2$ and GT methods have better network topology dependent bounds compared to \textsc{Dsgd}. We also study these  methods under the  Polyak-\L{}ojasiewicz (PL) condition \cite{Pol63,karimi2016linear} and provide refined bounds over existing literature.  Before we state our main contributions, we will go over the related works.

\subsection{Related Works}
There exists many works that study decentralized optimization methods under deterministic settings (full knowledge of gradients) -- see  \cite{nedic2009distributed,shi2015extra,
nedic2017achieving,scutari2019distributed,scaman2019optimal
,alghunaim2019linearly,arjevani2020ideal,alghunaim2019decentralized,xu2021distributed}  and references therein.   For example,  the works \cite{alghunaim2019decentralized,xu2021distributed,sundararajan2018canonical,sundararajan2020analysis,jakovetic2019unification}  propose unified frameworks that cover several state-of-the-art-methods and study their convergence, albeit under deterministic and convex settings. This work considers {\em nonconvex} costs and  focuses on the   {\em stochastic learning} setting where each agent has access to a random estimate of its gradient at each iteration. For this setting, \textsc{Dsgd}  is the most widely studied and understood method  \cite{bianchi2012convergence,chen2015onthepart1,chen2015onthepart2
,sayed2014nowbook,tatarenko2017non,swenson2020distributed
,pu2019sharp,vlaski2019distributedII,jiang2017collaborative
,lian2017can,koloskova2019decentralized
,assran2019stochastic,koloskova2020unified,wang2021cooperative}. Under non-convex settings, the transient time of \textsc{Dsgd}  is on the order of $O(n^3/(1-\lambda)^4)$  \cite{assran2019stochastic,lian2017can,koloskova2020unified}, where  $1 - \lambda \in (0,1)$ denotes the 
the network spectral gap 
that measures the connectivity of the network topology (\eg, it goes to zero for sparse networks).  As a result, improving the  convergence rate dependence  on the network topology quantity $1-\lambda$  is 
crucial to enhance the transient time of decentralized methods.

The severe dependence on the network topology in \textsc{Dsgd} is  caused by the data heterogeneity between different agents \cite{yuan2020influence,koloskova2020unified}. Consequently, the dependence on the network topology can be ameliorated by removing the bias caused by data heterogeneity. For example, the transient time of ED/D$^2$ has been shown to have enhanced dependence on the network topology compared to \textsc{Dsgd} \cite{yuan2021removing,huang2021improving} under convex settings. However, it is unclear whether bias-correction methods can achieve the same results for non-convex settings \cite{tang2018d,zhang2019decentralized,lu2019gnsd,
lu2020decentralized,xin2021improved,yi2020primal}. In fact, the established transient time of bias-correction methods such as ED/D$^2$ and GT in literature are even worse than that of \textsc{Dsgd}. For instance, the best known transient time for both ED/D$^2$ and GT is on the order of $O(n^3/(1-\lambda)^6)$ \cite{tang2018d,xin2021improved}, which is  worse than \textsc{Dsgd} with transient time $O(n^3/(1-\lambda)^4)$. These counter-intuitive results naturally motivates us to study whether ED/D$^2$ and GT can enjoy an enhanced dependence on the network topology in the non-convex setting. It is also worth noting that the  dependence on network topology established in existing GT references are worse than \textsc{Dsgd} even for convex scenarios \cite{pu2021distributed}.  This work provides  refined and enhanced  convergence rates for both ED/D$^2$ and GT (as well as other methods such as EXTRA) under the non-convex setting.


   In this work, we also study the convergence properties of decentralized methods  under the Polyak-\L{}ojasiewicz (PL) condition \cite{Pol63}.  The PL condition can hold for non-convex costs, yet it can be used to establish similar convergence  rates  to strong-convexity rates \cite{karimi2016linear}. For strongly-convex settings, the works \cite{huang2021improving,yuan2021removing} showed that the transient time of ED/D$^2$ is on the order of  $O(n/1-\lambda)$. These are the best available network bounds for decentralized methods so far for strongly-convex settings. It is still unclear, whether bias-correction methods can achieve similar bounds to \textsc{Dsgd} under the PL condition. For example, the work \cite{xin2021improved}  shows that under the PL condition, GT methods have transient time on the order $O(n/(1-\lambda)^3)$.

We remark that this work only considers non-accelerated decentralized methods with a {\em single} gossip round per iteration. It has been shown that combining GT methods with  {\em multiple} gossip rounds can further improve the dependence on network topology \cite{lu2021optimal,xin2021stochastic}, and this technique can also be incorporated into our studied algorithm and its analysis.
However, it is worth noting that the utilization of multiple gossip rounds in decentralized stochastic methods might suffer from several limitations. First, it requires the knowledge of the quantity $\lambda$ to decide the number of gossip rounds per iteration, which, however, might not be available in practice. Second, the multiple gossip rounds  update can take even more time than a global average operation. For example,  the experiments provided in \cite[Table 17]{chen2021accelerating} indicate that, under a certain practical scenario, one gossip step requires half or third the communication overhead of a centralized \textsc{Ring-Allreduce} operation \cite{patarasuk2009bandwidth}, which conducts global averaging. This implies that decentralized methods with as much as two or three gossip rounds per iteration can be more costly than global averaging.  Third, the theoretical improvements brought by multiple gossip rounds rely heavily on gradient accumulation. Such gradient accumulation can easily cause large batch-size which are empirically  and theoretically  found to be harmful for generalization performance on unseen dataset \cite{you2017large,gurbuzbalaban2021heavy}.

\begin{table}[t]  
\renewcommand{\arraystretch}{2}
\begin{center}
\caption{\small Comparison with existing {\em non-convex} convergence  rates highlighting the  network quantities. Here, $\varsigma_0^2 = \tfrac{1}{n} \sum_{i=1}^n \big\| \grad f_i(0)-\grad f(0) \big\|^2$ and $\varsigma^2$ satisfies $\tfrac{1}{n} \sum_{i=1}^n \big\| \grad f_i(x)-\grad f(x) \big\|^2 \leq \varsigma^2$ for all $x \in \real^{d}$ for \textsc{Dsgd}. The quantity $\lambda=\rho(W-\tfrac{1}{n} \one \one\tran)$ is the mixing rate of the network where $W$ is the network combination matrix.  Compared with GT methods our result assumes that $W$ is symmetric and positive-semidefinite.  }
\begin{adjustbox}{max width=\textwidth}
\begin{tabular}{cccc} \toprule
 {\sc method}  & {\sc Work} & {\sc Convergence rate}  & {\sc Transient time} \\ \midrule
 \textsc{Dsgd}  & \cite{koloskova2020unified}
        & $
	  	 O\left(\frac{1}{\sqrt{nK}}+\frac{\lambda^{2/3}}{(1-\lambda)^{1/3} K^{2/3}}+\frac{\lambda^{2/3}\varsigma^{2/3}}{(1-\lambda)^{2/3} K^{2/3}}\right)$
	  	 & $O\left(\frac{n^3}{(1-\lambda)^4}\right)$ \vspace{2mm}
	  	 \\ \hline 
  \multirow{2}{*}{ED/D$^2$} & \cite{tang2018d}  &  $O \left(  \frac{ 1 }{ \sqrt{n K}}	  
	 	   	 	  	 +  \frac{n  \lambda^2     }{ (1-\lambda)^3  K} 
	 	  	  +   \frac{     n  \varsigma_0^2}{  (1-\lambda)^2 K^2} \right)$ & $O\left(\frac{n^3}{(1-\lambda)^6}\right)$     \\ 
       &    \textbf{This work}       & $O \left(  \frac{ 1 }{ \sqrt{n K}}		 + \frac{n  \lambda^2     }{ (1-\lambda)  K} 
      	 	  	  +    \frac{     n \lambda^2  \varsigma_0^2}{  (1-\lambda)^2 K^2} \right)$
	 	  	  & $O\left( \frac{n^3}{(1-\lambda)^2} \right)$
	 	  	     \vspace{2mm}  \\ \hline
  \multirow{2}{*}{ATC-GT}                 &   \cite{xin2021improved}     &    
$	  	O \left(  \frac{1}{ \sqrt{n K}}	  
	 	    + \frac{n  \lambda^2     }{(1-\lambda)^3 K}
	 	  	  +    \frac{     \lambda^4  \sum_{i=1}^n \| \grad f_i(0) \|^2}{  (1-\lambda)^3 K^2} \right)$  
	 	  	  & 
	 	  	 $O\left(\frac{n^3}{(1-\lambda)^6}\right)$ \\
         & \textbf{This work}   & $O \left(  \frac{1}{ \sqrt{n K}}	  
	 	  	 + \frac{n  \lambda^4     }{ (1-\lambda)K} 
	 	  	  +  \frac{n  \lambda^4    }{ (1-\lambda)^4 K^2} 
	 	  	  
	 	  	  +    \frac{   n  \lambda^4  \varsigma_0^2}{  (1-\lambda)^3 K^2} \right)$
	 	  	  & $O\left( \max \left\{\frac{n^3}{(1-\lambda)^2} ,~\frac{n}{(1-\lambda)^{8/3}} \right\}\right)$
	 	  	\vspace{2mm}  \\ \bottomrule
\end{tabular}
\end{adjustbox}
\label{table_non_convex}
\end{center}
\end{table}


\subsection{Main Contributions}
 Our main contributions are formally listed below.
\begin{itemize}
\item  We unify the analysis of several  well-known  decentralized methods under {\em non-convex} and {\em stochastic} settings. In particular,   we study the convergence properties of a general primal-dual algorithmic framework, called stochastic unified decentralized algorithm (\suda), which includes several existing methods  such as EXTRA, ED/D$^2$,  and GT methods as special cases. 

\item We provide a novel analysis technique for these type of methods. In particular,  we employ several novel transformations to \suda~that are key to establish our refined convergence rates bounds (see Remark \ref{remark:gt_difference}).  Our analysis provides improved network dependent bounds for the special cases of \suda~such as  \textsc{ED/D$^2$} and \textsc{GT} methods compared to existing best known results. In addition, the established transient time of ED/D$^2$ and ATC-GT have improved  network topology dependence compared to \textsc{Dsgd} -- see Table \ref{table_non_convex}.

\item  We also study the convergence properties of  \suda~under the PL condition. When specifying \suda~to  ED/D$^2$, we achieve network dependent bound matching the best known bounds established under strongly-convex settings. When specifying \suda~to GT methods, we achieve an improved network dependent bound compared to current results even under strong-convexity. Table \ref{table_pl} compares the transient times network dependent bounds under the PL or strongly-convex setting. 

\end{itemize}

\begin{table}[t]  
\renewcommand{\arraystretch}{1.5}
\begin{center}
\caption{ \small  Comparison with existing network dependent transient times under both strongly-convex and PL condition settings. Here, the quantity $\lambda=\rho(W-\tfrac{1}{n} \one \one\tran)$ is the mixing rate of the network where $W$ is the network combination matrix.   Compared with GT methods our result assumes that $W$ is symmetric and positive-semidefinite.  }
\begin{tabular}{cccc} \toprule
 {\sc method}  & {\sc Work} & {\sc Assumption}  & {\sc Transient time} \\ \midrule
 \textsc{Dsgd}  & \cite{koloskova2020unified}
        & Strongly-convex 
        & $O\left(\frac{n}{(1-\lambda)^2}\right)$
	  \vspace{1mm}	 \\ \hline 
  \multirow{2}{*}{ED/D$^2$} & \cite{yuan2021removing,huang2021improving}   &  Strongly-convex  
  & $O\left(\frac{n}{1-\lambda}\right)$ \\ 
       &    \textbf{This work}       & PL condition
	 	  	  & $O\left(\frac{n}{1-\lambda}\right)$
	 	  	     \vspace{1mm}  \\ \hline
  \multirow{3}{*}{GT}         &   \cite{pu2021distributed}     &    
Strongly-convex 
	 	  	  & 
	 	  	 $O\left(\frac{n}{(1-\lambda)^3}\right)$\\        &   \cite{xin2021improved}     &    
PL condition 
	 	  	  & 
	 	  	 $O\left(\frac{n}{(1-\lambda)^3}\right)$ 
	 	  	  \\
         & \textbf{This work}   & PL condition
	 	  	  & $O\left( \max \left\{\frac{n}{1-\lambda} ,~\frac{1}{(1-\lambda)^{4/3}} \right\}\right)$
	 	   \\  \bottomrule
\end{tabular}
\label{table_pl}
\end{center}
\end{table}



{\bf Notation.} Vectors and scalars are denoted by lowercase letters. Matrices are denoted using uppercase letters. We use $\col\{a_1,\ldots,a_n\}$ (or $\col\{a_i\}_{i=1}^n$) to denote the vector that stacks the vectors (or scalars) $a_i$ on top of each other. We use $\diag\{d_1,\ldots,d_n\}$ (or $\diag\{d_i\}_{i=1}^n$) to denote a diagonal matrix with diagonal elements $d_i$. We also use $\bdiag\{D_1,\ldots,D_n\}$ (or $\bdiag\{D_i\}_{i=1}^n$) to denote a block diagonal matrix with diagonal blocks $D_i$. The vector of all ones with size $n$ is denoted by $\one_n$ (or $\one$ and size is known from context). The inner product of two vectors $a$ and $b$ is denoted by $\langle a,b \rangle$. The Kronecker product operation is denoted by $\otimes$. For a square matrix $A$, we let $\rho(A)$ denote the spectral radius of $A$, which is the largest absolute value of its eignevalues. Upright bold symbols (\eg, $\x,\f,\W$) are used to denote augmented network quantities.

\section{General Algorithm Description}
In this section, we describe the deterministic form of the studied algorithm and list several specific instances of interest to us.
\subsection{General Algorithm} \label{sec:uda}
 To describe the algorithm, we introduce the network quantities:
\begin{subequations}
\begin{align}
\x& \define \col\{x_1,\dots,x_n\} \in \real^{dn}, \\
\f(\x)& \define \sum_{i=1}^n f_i(x_i).
\end{align}
\end{subequations}
We also  introduce the matrix  $\B \in \real^{dn \times dn}$ that satisfies
\begin{align} \label{null_B}
\B \x&=\zero \iff x_1=x_2=\dots=x_n.
\end{align}
Using the previous definitions, the general algorithmic framework can be described as follows. Set an arbitrary initial estimate $\mathbf{x}^{0} \in \real^{dn}$ and set $\mathbf{y}^{0}=\zero$. Repeat for $k=0,1,\ldots$
\begin{subequations} \label{UDA_alg}
\begin{align} 
\mathbf{x}^{k+1} &=  \A \big(\C \mathbf{x}^{k}-\alpha \grad \mathbf{f}(\mathbf{x}^{k}) \big)  -  \mathbf{B} \mathbf{y}^{k}, \label{x_UDA}  \\
\mathbf{y}^{k+1} &= \mathbf{y}^{k}+ \mathbf{B}  \mathbf{x}^{k+1}. \label{dual_UDA} 
\end{align}
\end{subequations} 
Here, $\alpha>0$ is the step size (learning rate), and the matrices $\A \in \real^{dn \times dn}$ and $\C \in \real^{dn \times dn}$ are  doubly stochastic matrices that are chosen according to the network combination matrix introduced next. 

\subsection{Network Combination matrix}
 We let $W=[w_{ij}] \in \real^{n \times n}$ denote the network combination (weighting) matrix assumed to be symmetric. Here, the $(i,j)$th entry $w_{ij} \geq 0$ is used by agent $i$ to scale information received from agent $j$. We consider a  decentralized setup where $w_{ij}=0$ if $j \notin \cN_i$ where $\cN_i$ is the neighborhood of agent $i$. If we introduce the augmented combination matrix $
\W=W \otimes I_d \in \real^{dn \times dn}$, then
the matrices $\A,\B,\C$ can be chosen as a function of $\W$ to recover several existing decentralized methods. Note that if $\u=\col\{u_i\}_{i=1}^n$ where $u_i \in \real^d$ is local to agent $i$, then, the $i$th block of $\W\u=\col\{\sum_{j \in \cN_i} w_{ij} u_j\}_{i=1}^n$ can be computed by agent $i$ through local interactions with its neighbors.

\subsection{Relation to Existing Decentralized Methods} 
  Below, we list several important well-known decentralized  algorithms that are covered in our framework. Please see Appendix \ref{app:relation_to_other_methods}  for more details.

 \noindent \textbf{Exact-Diffusion/D$^2$ and EXTRA.} If  we select $\A=\W$, $\B=(\I-\W)^{1/2}$, and $\C=\I$. Then, algorithm \eqref{UDA_alg} becomes equivalent to   ED/D$^{2}$ \cite{yuan2019exactdiffI,tang2018d}:\begin{align} \label{exact_diff}
x_i^{k+2}=\sum_{j \in \cN_i} w_{ij} \Big(2 x_j^{k+1} - x_j^{k}-\alpha  \big( \grad f_j(x_j^{k+1})-\grad f_j(x_j^{k})\big) \Big),
\end{align}
with $x_i^1=\sum_{j \in \cN_i} w_{ij} \big( x_j^0-\alpha \grad f_j(x_j^0)\big)$.   If we instead  select  $\A=\I$, $\B=(\I-\W)^{1/2}$, and $\C=\W$, then algorithm \eqref{UDA_alg} is equivalent to  EXTRA \cite{shi2015extra}:
\begin{align} \label{EXTRA}
x_i^{k+2}=\sum_{j \in \cN_i} w_{ij} \big(2 x_j^{k+1} - x_j^{k} \big)-\alpha  \big( \grad f(x_i^{k+1})-\grad f(x_i^{k})\big), 
\end{align}
with $x_i^1=\sum_{j \in \cN_i} w_{ij}  x_j^0-\alpha \grad f_i(x_i^0)$. 

 \noindent \textbf{Gradient-Tracking (GT) methods.}  Consider  the adapt-then-combine gradient-tracking (ATC-GT) method \cite{xu2015augmented}:
\begin{subequations} \label{GT_atc_alg}
\begin{align} 
x_i^{k+1}&=\sum_{j \in \cN_i} w_{ij} (x_j^{k} - \alpha g_j^{k}) \\
g_i^{k+1} &= \sum_{j \in \cN_i} w_{ij} \big(g_j^{k} + \grad f_j(x_j^{k+1})-\grad f_j(x_j^{k}) \big).
\end{align}
\end{subequations}
    With proper initialization, the above is equivalent to  \eqref{UDA_alg} when $
  \A=\W^2$, $\B=\I-\W$, and $\C= \I$. We can also recover other gradient-tracking variants. For example, if we select $
\A=\I$, $\B=\I-\W$, and  $\C= \W^2$, then \eqref{UDA_single_update} becomes equivalent to the non-ATC-GT method \cite{qu2017harnessing}:
\begin{subequations} \label{GT_nonatc_alg}
\begin{align}
x_i^{k+1}&=\sum_{j \in \cN_i} w_{ij} x_j^{k} - \alpha g_i^{k} \\
g_i^{k+1} &= \sum_{j \in \cN_i} w_{ij} g_j^{k} + \grad f_i(x_i^{k+1})-\grad f_i(x_i^{k}).
\end{align}
\end{subequations}
Notice that in \eqref{GT_nonatc_alg} the communication (gossip) step only involves the terms $x_j^{k}$ and $g_j^{k}$ in the update of each vector. This is in contrast to the ATC structure \eqref{GT_atc_alg} where the  communication (gossip) step involves all terms. Similarly, we can also cover the semi-ATC-GT variations \cite{di2016next} where only the update of $x_i^{k}$ or $g_i^{k}$  uses the ATC structure. Please see Appendix \ref{app:relation_to_other_methods} for details.

	\begin{remark}[\sc  Relation with other frameworks] \rm 
		The unified decentralized algorithm (UDA) from \cite{alghunaim2019decentralized} is equivalent to \eqref{UDA_alg} if $\A$ and $\B^2$ commute (\ie, $\A \B^2=\B^2 \A$).  Therefore, all the methods covered in \cite{alghunaim2019decentralized} are also covered by our framework (such as DLM \cite{ling2015dlm}). Moreover, under certain conditions, the frameworks from \cite{xu2021distributed} and \cite{sundararajan2018canonical} can also be related with \eqref{UDA_alg} -- see Appendix \ref{app:relation_to_other_methods}.  The works \cite{alghunaim2019decentralized,xu2021distributed,sundararajan2018canonical} only studied convergence under deterministic and convex settings.  In contrast, we study  the {\em non-convex and stochastic} case, and more importantly, we establish tighter rates for the above special bias-correction methods, which is the main focus of this work.
	\end{remark}

\section{Stochastic UDA and Assumptions} \label{sec:suda}
In this section, we describe the stochastic version of algorithm \eqref{UDA_alg} and list the assumptions used to analyze it.

 As stated in problem \eqref{min_learning_prob}, we consider stochastic settings where each agent may only have access to a stochastic gradient $\grad F_i (x_i,\xi_i^k)$ at each iteration $k$ instead of the true gradient. This scenario arises in online learning settings, where the data are not known in advance; hence, we do not have access to the actual gradient. Moreover, even if all the data is available, the true gradient might be expensive to compute for large datasets and can be replaced by a gradient at one sample or a mini-batch.

Replacing the actual gradient by its stochastic approximation in \eqref{UDA_alg}, we get the \textbf{S}tochastic \textbf{U}nified \textbf{D}ecentralized \textbf{A}lgorithm (\suda):
\begin{subequations} \label{SUDA_alg}
\begin{align} 
\mathbf{x}^{k+1} &=  \A \big(\C \mathbf{x}^{k}-\alpha \grad \F (\x^k,\bxi^k) \big)  -  \mathbf{B} \mathbf{y}^{k}, \label{x_SUDA}  \\
\mathbf{y}^{k+1} &= \mathbf{y}^{k}+ \mathbf{B}  \mathbf{x}^{k+1}, \label{dual_SUDA} 
\end{align}
\end{subequations}
where
\begin{align*}
\grad \F (\x,\bxi^k) \define \col\{ \grad F_1 (x_1,\xi_1^k),\dots,\grad F_n (x_n,\xi_n^k) \}.
\end{align*}

We next  list the assumptions used in our analyses. Our first assumption is about the network combination matrix given next.  
\begin{assumption}[\bfseries \small Combination matrix] \label{assump:network} 
The combination matrix $W$ is assumed to be doubly stochastic, symmetric, and primitive. Moreover, we assume that the matrices $\A,\B^2,\C$ are chosen as a polynomial function of $\W$:
\begin{align} \label{polynomials_matrices}
\A=\sum_{l=0}^p a_l \W^l, \quad \B^2=\sum_{l=0}^p b_l \W^l, \quad
\C=\sum_{l=0}^p c_l \W^l, 
\end{align}
where $p \geq 1$.  The constants $\{a_l,b_l,c_l\}_{l=0}^p$ are chosen such that $\A$ and $\C$ are doubly stochastic and the matrix $\B$ satisfies equation \eqref{null_B}. \qd
\end{assumption}
\noindent   Under Assumption \ref{assump:network}, the combination matrix $W$ has a single eigenvalue at one, denoted by $\lambda_1=1$. Moreover, all other eigenvalues, denoted by $\{\lambda_i\}_{i=2}^n$, are strictly less than one in magnitude \cite{sayed2014nowbook},  and the mixing rate of the network is:
\begin{align} \label{graph_mixing_rate}
\lambda \define \rho \big(W-\tfrac{1}{n} \one \one\tran\big) =\max_{i \in \{2,\ldots,n\}} |\lambda_i| <1.
\end{align}
Note that the assumptions on $\A,\B^2,\C$ are mild and hold for all the special cases described before.

We now introduce the main assumption on the objective function. 
\begin{assumption}[\bfseries \small Objective function] \label{assump:smoothness} Each function $f_i: \real^d \rightarrow \real$ is $L$-smooth:
\begin{align} \label{smooth_f_eq}
\|\grad f_i(z) -\grad f_i(y)\| \leq L \|z -y\|, \quad \forall~z,y \in \real^d, 
\end{align}
for some $L>0$. We also assume that the aggregate function $f(x)=\frac{1}{n} \sum_{i=1}^n f_i(x)$ is bounded below, \ie, $f(x) \geq  f^\star > -\infty$ $\forall~ x \in \real^d$ where $f^\star$ denote the optimal value of $f$.  \qd
\end{assumption}
\noindent  The above assumption is standard to establish convergence under non-convex settings. Note that we do not impose the strong assumption of bounded gradient dissimilarity, which is required to establish convergence of \textsc{Dsgd} -- see  \cite{assran2019stochastic,lian2017can,koloskova2020unified}.

We next list our assumption on the stochastic gradient. To do that, we define the filtration generated by the random process \eqref{SUDA_alg}:
\begin{align}
\bm{\cF}^k \define \{\x^0,\x^2,\ldots,\x^k\}.
\end{align}
The filtration $\bm{\cF}^k$ can be interpreted as the collection of all information available on the past iterates up to time $k$. 
\begin{assumption}[\bfseries \small Gradient noise] \label{assump:noise} 
For all $\{i\}_{i=1}^n$ and $k=0,1,\ldots$, we assume that the following holds
\begin{subequations} \label{noise_bound_eq}
\begin{align}
\Ex \big[\grad F_i(x_i^k;\xi_i^k)-\grad f_i(x_i^k) ~|~ \bm{\cF}^{k}\big] &=0, \label{noise_bound_eq_mean} \\
\Ex \big[\|\grad F_i(x_i^k;\xi_i^k)-\grad f_i(x_i^k)\|^2 ~|~ \bm{\cF}^{k} \big] &\leq \sigma^2, \label{noise_bound_eq_variance}
\end{align}
\end{subequations}
 for some $\sigma^2 \geq 0$. We also assume that conditioned on $\bm{\cF}^{k}$, the random data  $\{\xi_i^t\}$ are independent of each other for all $\{i\}_{i=1}^n$ and $\{t \}_{t \leq k}$.   \qd
\end{assumption}

 The previous assumptions will be used to analyze \suda~\eqref{SUDA_alg} for general non-convex costs. In the sequel, we  will also study \suda~under the following additional assumption. 
\begin{assumption}[\bfseries \small PL condition] \label{assump:PL} The aggregate function $f(x)=\frac{1}{n} \sum_{i=1}^n f_i(x)$ satisfies the  PL inequality:
\begin{align} \label{PL_cond}
2\mu \big(f(x)-f^\star\big) \leq \|\grad f(x)\|^2, \quad \forall~ x \in \real^d,
\end{align}
for some $\mu >0$ where $f^\star$ denote the optimal value of $f$.  \qd
\end{assumption}
\noindent The above condition implies that every stationary point is a global minimizer, which is weaker than many assumptions used to establish linear convergence without strong-convexity \cite{karimi2016linear}. Note that the PL condition is also referred to as the gradient dominated condition \cite{tang2020distributed}.

\section{Fundamental Transformations}
The updates of \suda~\eqref{SUDA_alg},  while useful for implementation purposes, are not helpful for analysis purposes. In this section,  we will transform  \suda~\eqref{SUDA_alg} into  an equivalent recursion that is fundamental to arrive at our results.  

\subsection{Transformation I}
  Using the change of variable 
\begin{align}
\z^{k}\define \y^{k}-\B \x^{k}
\end{align} 
  in \eqref{x_SUDA}--\eqref{dual_SUDA}, we can describe \eqref{SUDA_alg} by the following equivalent non-incremental form: 
\begin{subequations} \label{alg_uda_noninc}
\begin{align}
\mathbf{x}^{k+1} &=  (\A\C-\B^2) \mathbf{x}^{k}-\alpha \A (\grad \mathbf{f}(\x^{k}) +\w^k)   -  \mathbf{B} \mathbf{z}^{k},  \\
\mathbf{z}^{k+1} &= \mathbf{z}^{k}+ \mathbf{B}  \mathbf{x}^{k},
\end{align}
\end{subequations}
where $\w^k$ is the gradient noise defined as:
\begin{align} \label{noise_gradient}
\w^k \define \grad \F (\x^k,\bxi^k) -\grad \mathbf{f}(\x^{k}).
\end{align}
 If we introduce the quantities
 \begin{subequations}
 \begin{align}
 \bar{x}^k & \define \tfrac{1}{n} (\one_n\tran \otimes I_d) \x^{k}=\frac{1}{n} \sum_{i=1}^n x_i^k, \label{bar_x_def}  \\
  \bar{\x}^k& \define \one_n \otimes  \bar{x}^k, \label{bar_x_augmented_def} \\
  \s^{k}&\define\B \z^{k} +\alpha \A \grad \mathbf{f}(\bar{\x}^{k}), \label{s_definition}
\end{align}  
\end{subequations}
then recursion \eqref{alg_uda_noninc} can be rewritten as:
\begin{subequations} \label{error0}
\begin{align} 
\x^{k+1} &=  (\A\C-\B^2) \x^{k} -  \s^{k} -\alpha \A \big( \grad \mathbf{f}(\x^{k}) - \grad \mathbf{f}(\bar{\x}^{k}) +\w^k\big)  ,  \label{x_error0} \\
\s^{k+1} &= \s^{k}+ \B^2  \x^{k}+ \alpha \A \big( \grad \mathbf{f}(\bar{\x}^{k+1})- \grad \mathbf{f}(\bar{\x}^{k})\big). \label{s_error0} 
\end{align}
\end{subequations}
	\begin{remark}[\sc Motivation behind  $\s^k$] \rm
		Suppose that $(\x,\s)$ is a fixed point  of the deterministic form of \eqref{error0} (gradient noise $\w^k=\zero$) where $\s \in \real^{dn}$ and  $\x=(x_1,\dots,x_n) \in \real^{dn}$. Then, from \eqref{s_error0} it holds that 
		$\zero =  \B^2  \x \iff x_1=\dots=x_n=x$ and from \eqref{x_error0}, we have $
		\s=\zero$. 
		Therefore, using the definition of $\s^k$ in  \eqref{s_definition} and $\bar{\x}=\one \otimes x$, it follows that there exists some $\z \in \real^{dn}$ such that  
		\begin{align*}
			&\s= \alpha \A  \grad \mathbf{f}(\bar{\x})   + \B \z=\zero  \\
			& \Rightarrow  \tfrac{1}{n}(\one\tran \otimes I_n) \left(\alpha \A  \grad \mathbf{f}(\bar{\x})   + \B \z \right) =\frac{\alpha}{n} \sum_{i=1}^n \grad f_i(x)=0.
		\end{align*}
		Hence, $x$ is a stationary point of problem \eqref{min_learning_prob}.   \qd
	\end{remark}
\subsection{Transformation II} 
 We next  exploit the structure of the matrices $\A$, $\B$, and $\C$ to further transform recursion \eqref{error0} into a more useful form.  Under Assumption \ref{assump:network}, the combination matrix $W$  can be decomposed as
\begin{align*}
W=U \Lambda U\tran = \begin{bmatrix}
\frac{1}{\sqrt{n}}  \one & \hat{U}
\end{bmatrix} \begin{bmatrix}
1 & 0 \\
0 & \hat{\Lambda}
\end{bmatrix} \begin{bmatrix}
\frac{1}{\sqrt{n}} \one\tran \vspace{0.6mm} \\  \hat{U}\tran
\end{bmatrix},
\end{align*}
where $\hat{\Lambda}=\diag\{\lambda_i\}_{i=2}^n$. The matrix $U$ is an orthogonal matrix ($UU\tran=U\tran U=I$) and $\hat{U}$ is an ${n \times (n-1)}$ matrix that satisfies $\hat{U}\hat{U}\tran=I_n-\tfrac{1}{n} \one \one\tran$  and  $\one\tran \hat{U}=0$. It follows that
\begin{align*}
\W&=\U\mathbf{\Lambda} \U\tran = \begin{bmatrix}
\frac{1}{\sqrt{n}}   \one \otimes I_d & \hat{\U}
\end{bmatrix} \begin{bmatrix}
I_d & 0 \\
0 & \hat{\mathbf{\Lambda}}
\end{bmatrix} \begin{bmatrix}
\frac{1}{\sqrt{n}}  \one\tran \otimes I_d \\ \hat{\U}\tran
\end{bmatrix},
\end{align*}
where $\hat{\mathbf{\Lambda}} \define \hat{\Lambda} \otimes I_d \in \real^{d(n-1)\times d(n-1)}$,  $\U \in \real^{dn \times dn}$ is an orthogonal matrix, and $\hat{\U}\define \hat{U} \otimes I_d \in \real^{dn \times d(n-1)}$ satisfies:
\begin{align} \label{UUtran}
\hat{\U}\tran \hat{\U}=\I, \quad \hat{\U}\hat{\U}\tran=\I-\tfrac{1}{n} \one \one\tran \otimes I_d, \quad (\one\tran \otimes I_d) \hat{\U}=0.
\end{align}
Now, since $\A,\B^2,\C$ are chosen as polynomial function of $\W$ as described in Assumption \ref{assump:network}, it holds that
\begin{subequations} \label{ABC_decompositon}
\begin{align}
\A&=\U\mathbf{\Lambda}_a \U\tran = \begin{bmatrix}
\frac{1}{\sqrt{n}}   \one \otimes I_d & \hat{\U}
\end{bmatrix} \begin{bmatrix}
I_d & 0 \\
0 & \hat{\mathbf{\Lambda}}_a
\end{bmatrix} \begin{bmatrix}
\frac{1}{\sqrt{n}} \one\tran \otimes I_d \\ \hat{\U}\tran
\end{bmatrix}, \\
\C&=\U\mathbf{\Lambda}_c \U\tran = \begin{bmatrix}
\frac{1}{\sqrt{n}}  \one \otimes I_d & \hat{\U}
\end{bmatrix} \begin{bmatrix}
I_d & 0 \\
0 & \hat{\mathbf{\Lambda}}_c
\end{bmatrix} \begin{bmatrix}
\frac{1}{\sqrt{n}}  \one\tran \otimes I_d \\ \hat{\U}\tran
\end{bmatrix}, \\
\B^2&=\U \mathbf{\Lambda}_b^2 \U\tran = \begin{bmatrix}
\frac{1}{\sqrt{n}}  \one \otimes I_d & \hat{\U}
\end{bmatrix} \begin{bmatrix}
0 & 0 \\
0 & \hat{\mathbf{\Lambda}}_b^2
\end{bmatrix} \begin{bmatrix}
\frac{1}{\sqrt{n}}  \one\tran \otimes I_d \\ \hat{\U}\tran
\end{bmatrix}, 
\end{align}
\end{subequations}
where 
\begin{align}
\hat{\mathbf{\Lambda}}_a=\diag\{\lambda_{a,i}\}_{i=2}^n \otimes I_d, \quad \hat{\mathbf{\Lambda}}_b^2=\diag\{\lambda_{b,i}^2\}_{i=2}^n \otimes I_d, \quad \hat{\mathbf{\Lambda}}_c=\diag\{\lambda_{i,c}\}_{i=2}^n \otimes I_d,\label{znznas82}
\end{align}
 with $\lambda_{a,i} \define \sum_{l=0}^p a_l \lambda_i^l$, $\lambda_{b,i}^2 \define \sum_{l=0}^p b_l \lambda_{i}^l$, and $
\lambda_{c,i} \define \sum_{l=0}^p c_l \lambda_{i}^l$. Moreover,  $\hat{\mathbf{\Lambda}}_b$ is positive definite due to the null space condition \eqref{null_B}.   Multiplying both sides of \eqref{error0} by $\U\tran$ and using the structure \eqref{ABC_decompositon}, we get
\begin{subequations} \label{error1}
\begin{align} 
\U\tran \x^{k+1} &=   (\mathbf{\Lambda}_a \mathbf{\Lambda}_c-\mathbf{\Lambda}_b^2 ) \U\tran \x^{k} 
- \U\tran \s^{k} 
-\alpha  \mathbf{\Lambda}_a \U\tran \big( \grad \mathbf{f}(\mathbf{x}^{k}) - \grad \mathbf{f}(\bar{\x}^{k})+\w^k \big)  ,  \label{x_error1} \\
\U\tran \s^{k+1} &= \U\tran \s^{k}+    \mathbf{\Lambda}_b^2 \U\tran  \x^{k}+ \alpha  \mathbf{\Lambda}_a \U\tran \big( \grad \mathbf{f}(\bar{\x}^{k+1})- \grad \mathbf{f}(\bar{\x}^{k}) \big). \label{s_error1} 
\end{align}
\end{subequations}
 Note that  
\begin{align}
( \one\tran \otimes I_d)\s^k \overset{\eqref{s_definition}}{=}(\one\tran \otimes I_d) \big(\B \z^{k} +\alpha \A \grad \mathbf{f}(\bar{\x}^{k}) \big)=  \alpha \sum_{i=1}^n \grad f_i(\bar{x}^k).
\end{align}
Moreover, utilizing the structure of $\U$, we have
\begin{subequations}
\begin{align}
\U\tran  \x^{k}&=\begin{bmatrix}
\sqrt{n}\, \bar{x}^k  \vspace{0.5mm} \\
\hat{\U}\tran \x^{k}
\end{bmatrix}, 
\quad 
\U\tran  \s^{k} =\begin{bmatrix}
\sqrt{n}\, \alpha \overline{\grad \f}(\bar{\x}^k)  \vspace{0.5mm}  \\
\hat{\U}\tran \s^{k}
\end{bmatrix},
\\ 
\U\tran  \grad \mathbf{f}(\mathbf{x})&=\begin{bmatrix}
 \sqrt{n}\, \overline{\grad \f}(\x^k)  \vspace{0.5mm} \\
\hat{\U}\tran \grad \mathbf{f}(\mathbf{x})
\end{bmatrix},  
\quad
\U\tran  \w^{k} =\begin{bmatrix}
\sqrt{n}\, \bar{\w}^k  \vspace{0.5mm}  \\
\hat{\U}\tran \w^{k}
\end{bmatrix},
\end{align} 
\end{subequations}
where
\begin{subequations}
\begin{align}
 \overline{\grad \f}(\x^k)& \define\tfrac{1}{n} (\one_n\tran \otimes I_d) \grad \f(\x^k)=\frac{1}{n} \sum_{i=1}^n \grad f_i(x_i^k), \\
 \bar{\w}^k& \define \tfrac{1}{n} (\one_n\tran \otimes I_d) \w^k=\frac{1}{n} \sum_{i=1}^n \big(\grad F_i(x_i^k,\xi^k_i)-\grad f_i(x_i^k) \big).
\end{align}
\end{subequations}
Hence, using the previous quantities and the structure of $\mathbf{\Lambda}_a,\mathbf{\Lambda}_b^2,\mathbf{\Lambda}_c$ given in \eqref{ABC_decompositon}, we find that
\begin{subequations} 
\begin{align}
\bar{x}^{k+1} &=\bar{x}^{k} - \alpha  \overline{\grad \f}(\x^k)-\alpha \bar{\w}^k,
\\
\hat{\U}\tran \x^{k+1} &=   (\hat{\mathbf{\Lambda}}_a \hat{\mathbf{\Lambda}}_c-\hat{\mathbf{\Lambda}}_b^2 ) \hat{\U}\tran \x^{k}
 - \hat{\U}\tran \s^{k}
 -\alpha  \hat{\mathbf{\Lambda}}_a \hat{\U}\tran \big( \grad \mathbf{f}(\mathbf{x}^{k}) - \grad \mathbf{f}(\bar{\x}^{k})+\w^k \big) , 
\\
\hat{\U}\tran \s^{k+1} &= \hat{\U}\tran \s^{k}+    \hat{\mathbf{\Lambda}}_b^2 \hat{\U}\tran  \x^{k}+ \alpha  \hat{\mathbf{\Lambda}}_a \hat{\U}\tran \big( \grad \mathbf{f}(\bar{\x}^{k+1})- \grad \mathbf{f}(\bar{\x}^{k})\big). 
\end{align}
\end{subequations}
 Multiplying the third equation by $\hat{\mathbf{\Lambda}}_b^{-1}$ and  rewriting the previous recursion in matrix notation, we obtain: 
\begin{subequations} \label{trans_uda_non_diag}
\begin{align}
\hspace{-0.75mm} \bar{x}^{k+1} &=\bar{x}^{k} - \alpha  \overline{\grad \f}(\x^k) -\alpha \bar{\w}^k, \label{trans_uda_non_diag_x} \\
 \begin{bmatrix}
\hat{\U}\tran\x^{k+1} \\
\hat{\mathbf{\Lambda}}_b^{-1} \hat{\U}\tran\s^{k+1}
\end{bmatrix}&=  \begin{bmatrix}
\hat{\mathbf{\Lambda}}_a \hat{\mathbf{\Lambda}}_c-\hat{\mathbf{\Lambda}}_b^2 & -\hat{\mathbf{\Lambda}}_b \\
   \hat{\mathbf{\Lambda}}_b  & ~~\I
\end{bmatrix}  \begin{bmatrix}
\hat{\U}\tran\x^{k} \\
\hat{\mathbf{\Lambda}}_b^{-1} \hat{\U}\tran\s^{k}
\end{bmatrix} \nonumber \\
& \quad  - \alpha  \begin{bmatrix}
\hat{\mathbf{\Lambda}}_a \hat{\U}\tran \big(\grad \mathbf{f}(\mathbf{x}^{k}) - \grad \mathbf{f}(\bar{\x}^{k})+\w^k\big) \\
\hat{\mathbf{\Lambda}}_b^{-1} \hat{\mathbf{\Lambda}}_a \hat{\U}\tran \big(\grad \mathbf{f}(\bar{\x}^{k}) -\grad \mathbf{f}(\bar{\x}^{k+1})\big)
\end{bmatrix}. \label{trans_uda_non_diag_y}
\end{align}
\end{subequations}
The convergence of \eqref{trans_uda_non_diag} is be governed by the matrix:
\begin{align} \label{G_matrix}
\G \define \begin{bmatrix}
\hat{\mathbf{\Lambda}}_a \hat{\mathbf{\Lambda}}_c-\hat{\mathbf{\Lambda}}_b^2 & -\hat{\mathbf{\Lambda}}_b \\
   \hat{\mathbf{\Lambda}}_b  & ~~\I
\end{bmatrix} \in \real^{2d(n-1)\times 2d(n-1)}.
\end{align}
We next introduce a fundamental factorization of $\G$ that will be used to transform \eqref{trans_uda_non_diag} into our final key recursion. The next result is proven in Appendix \ref{app:lemma_diag_proof}. 
\begin{lemma}[\bfseries \small Fundamental factorization] \label{lemma:diagonalization}
  Suppose that  the eigenvalues of $\G$ are strictly less than one in magnitude. Then, there exists an invertible matrix $\hat{\V}$ such that the matrix $\G$ admits the similarity transformation
\begin{align} \label{G_diagonalized}
 \G  = \hat{\V} \mathbf{\Gamma} \hat{\V}^{-1},
\end{align}
where $\mathbf{\Gamma}$ satisfies $ \|\mathbf{\Gamma}\|<1$.  \qd
\end{lemma}
\noindent For the convergence of \eqref{trans_uda_non_diag}, it is necessary that $\G$ is a stable matrix (has eigenvalues strictly less than one in magnitude). To see this suppose that $\grad \f(\x)=\w^k=\zero$, then the convergence is dictated by \eqref{trans_uda_non_diag_y}, which diverges for unstable $\G$ if $\hat{\U}\tran\x^{0} \neq \zero$ or $\hat{\U}\tran\s^{0} \neq \zero$.   Thus, we implicitly assume that $\G$ is a stable matrix. Explicit expressions for $\mathbf{\Gamma}$ and $\hat{\V}$ for ED/D$^2$, EXTRA, and  GT methods  are derived in Appendix \ref{app:bounds_special_cases}, where we also find exact expressions for the eigenvalues of $\G$ for these methods.

 Finally, multiplying both sides of  \eqref{trans_uda_non_diag_y} by $\frac{1}{\upsilon}\hat{\V}^{-1}$ for any $\upsilon >0$ and using the structure \eqref{G_diagonalized}, we arrive at the following key result.
\begin{lemma}[\bfseries \small Final transformed recursion] Under Assumption \ref{assump:network}, there exists an invertible matrix $\hat{\V}$  such that recursion \eqref{trans_uda_non_diag} can transformed into
\begin{subequations} \label{error_diag_transformed}
\begin{align}
\bar{x}^{k+1} &=\bar{x}^{k} - \alpha  \overline{\grad \f}(\x^k)-\alpha \bar{\w}^k, \label{error_average_diag} \\
\hat{\e}^{k+1}&=\mathbf{\Gamma} \hat{\e}^{k} - \alpha \hat{\V}^{-1}  \begin{bmatrix}
\frac{1}{\upsilon} \hat{\mathbf{\Lambda}}_a \hat{\U}\tran \big(\grad \mathbf{f}(\mathbf{x}^{k}) - \grad \mathbf{f}(\bar{\x}^{k})+\w^k\big) \\
\frac{1}{\upsilon}  \hat{\mathbf{\Lambda}}_b^{-1} \hat{\mathbf{\Lambda}}_a \hat{\U}\tran \big(\grad \mathbf{f}(\bar{\x}^{k}) -\grad \mathbf{f}(\bar{\x}^{k+1})\big)
\end{bmatrix}, \label{error_check_diag}
\end{align}
\end{subequations}
where $\mathbf{\Gamma}$ was introduced in \eqref{G_diagonalized},  $\upsilon$ is an arbitrary strictly positive constant, and
\begin{align} \label{e_hat_def}
\hat{\e}^{k} \define \frac{1}{\upsilon}\hat{\V}^{-1} \begin{bmatrix}
\hat{\U}\tran\x^{k} \\
\hat{\mathbf{\Lambda}}_b^{-1} \hat{\U}\tran\s^{k}
\end{bmatrix}.
\end{align}
\qd
\end{lemma}
\begin{remark}[\sc Deviation from average] \rm
Recall  that $\bar{\x}^k= \one  \otimes \bar{x}^k$ where $\bar{x}^k=\tfrac{1}{n} \sum_{i=1}^n x^k_i$. Since $\hat{\U}\tran \hat{\U}=\I$, it holds that
\begin{align*}
\|\hat{\U}\tran \x\|^2 =  \x\tran \hat{\U} \hat{\U}\tran \hat{\U} \hat{\U}\tran \x =\|\hat{\U} \hat{\U}\tran \x^k\|^2 \overset{\eqref{UUtran}}{=}\|\x^k-\bar{\x}^k\|^2.
\end{align*}
Therefore, the quantity $\|\hat{\U}\tran \x^k\|^2$ measures the deviation of $\x^k$ from the average $\bar{\x}^k$.  Similarly, $\|\hat{\U}\tran \s^k\|^2$ measures the deviation of $\s^k$ with the average $(\tfrac{1}{n}\one \one\tran \otimes I_n) \s^k=\one \otimes \tfrac{\alpha}{n}\sum_{i=1}^n \grad f_i(\bar{x}^k)$ (see \eqref{s_definition}). Now, using \eqref{e_hat_def}, we have
\begin{align} \label{hat_relation_avg_Dev}
	\|\upsilon \hat{\V} \hat{\e}^{k} \|^2=\|\hat{\U}\tran\x^{k}\|^2 + \|\hat{\mathbf{\Lambda}}_b^{-1} \hat{\U}\tran\s^{k}\|^2.
\end{align}
Thus, the vector $\hat{\e}^{k}$ can be interpreted as a measure of a weighted deviation of $\x^k$ and  $\s^k$  from  $\bar{\x}^k$ and $\one \otimes \tfrac{\alpha}{n}\sum_{i=1}^n \grad f_i(\bar{x}^k)$, respectively.   \qd
\end{remark}
	\begin{remark} \label{remark:gt_difference}\rm
		This work handles the deviation of the individual vectors from the average $\mathbf{x}^k -\bar{\mathbf{x}}^k$, and the deviation of the ``gradient-tracking'' variable $\s^k$ from the averaged-gradients $\one \otimes \tfrac{\alpha}{n}\sum_{i=1}^n \grad f_i(\bar{x}^k)$ as one augmented quantity. This leads to two coupled error terms given by \eqref{trans_uda_non_diag}. This is one main reason that allows us to obtain tighter network-dependent rates compared to previous works. The rigorous factorization of the matrix $\mathbf{G}$ leads us to arrive at the final transformed recursion \eqref{error_diag_transformed} with contractive matrix $\mathbf{\Gamma}$, which is key for our result.

		This is in contrast to other works. For example, in previous GT works (see \eg, \cite{qu2017harnessing,xin2021improved,pu2021distributed}),  the deviation of the individual vectors from the average $\mathbf{x}^k -\bar{\mathbf{x}}^k$, and the deviation of the gradient-tracking variable from the averaged-gradients are handled independently, and thus, they do not exploit the coupling matrix between these two network quantities.  \qd
	\end{remark}	
	
	\section{Convergence Results}
In this section, we state and discuss the convergence results.	  
\subsection{Convergence Under General Non-Convex Costs}
The following result establishes the convergence under general non-convex smooth costs.
 \begin{theorem} [\bfseries \small Convergence of SUDA] \label{thm_nonconvex}
	Suppose that Assumptions \ref{assump:network}--\ref{assump:noise} hold and the step size satisfies $
	    \alpha \leq \min \left\{\frac{1}{2L},~ \frac{1-\gamma}{2L v_1v_2 \lambda_a}
	    ,~ \frac{\sqrt{\underline{\lambda_b}(1-\gamma)}}{2 L \sqrt{v_1 v_2 \lambda_a  }}
	    \right\}$. 	Then, the iterates $\{\x^k\}$ of \suda~with $\x^0=\one \otimes x^0$ ($x^0 \in \real^d$) satisfy
	\begin{equation} \label{eq:thm:nonconvex}
		\begin{aligned}
	\frac{1}{K}	 \sum_{k=0}^{K-1}  \Ex  \| \grad f(\bar{x}^{k})  \|^2 +   \Ex  \| \overline{\grad\f}(\x^k)\|^2 		&\leq 
	  	  \frac{8   (f(x^{0})-f^\star)}{\alpha K}	  
	 	   + 	\frac{ 4\alpha L \sigma^2 }{ n} 
	 	   +   \frac{12  \alpha^2  L^2   v_1^2 v_2^2   \zeta_0^2}{ \underline{\lambda_b}^2  (1-\gamma)K}
	 	   \\
	 	  & \quad     
	 +  \frac{16\alpha^2   L^2 v_1^2 v_2^2   \lambda_a^2  \sigma^2}{1-\gamma}+\frac{16\alpha^4 L^4 v_1^2 v_2^2  \lambda_a^2  \sigma^2}{ \underline{\lambda_b}^2(1-\gamma)^2n},
	\end{aligned}
	\end{equation}
	where $v_1 \define \|\hat{\V}\|$, $v_2 \define \|\hat{\V}^{-1}\|$ and
	\begin{align*}
	 \gamma &\define \|\mathbf{\Gamma}\|<1, \quad \underline{\lambda_b}\define\frac{1}{\|\mathbf{\Lambda}_b^{-1}\|}, \quad \lambda_a \define \|\mathbf{\Lambda}_a\|, \\
\zeta_0 &\define \tfrac{1}n \|(\A-\tfrac{1}{n}\one\tran \one \otimes I_d )  \big(\grad \mathbf{f}(\x^0) - \one \otimes \grad f(x^0) \big)\|.
	\end{align*}
	Consequently, if  we set $\alpha=\tfrac{1}{2L \beta +\sigma \sqrt{K/n}}$ where $\beta \define 1
+\frac{ v_1v_2 \lambda_a}{1-\gamma}
+\frac{\sqrt{v_1 v_2 \lambda_a  }}{\sqrt{\underline{\lambda_b}(1-\gamma)}}$. Then, we obtain the convergence rate
\begin{equation} \label{eq:thm:nonconvex_rate}
	\begin{aligned} 
 	\frac{1}{K}	 \sum_{k=0}^{K-1}    \Ex  \big\| \overline{\grad\f}(\x^k) \big\|^2 		\leq 
	  	O \left(    
	 	   	\frac{   \sigma }{ \sqrt{nK}} 
	 	   	+ \frac{1}{K}
	 	  	 + \frac{  n \sigma^2}{n+\sigma^2 K}
	 	  	  +    \frac{  n \zeta_0^2}{ n K+\sigma^2 K^2} \right).
	\end{aligned}
\end{equation}
\qd
\end{theorem}
\noindent The proof of Theorem \ref{thm_nonconvex} is established in Appendix \ref{app:convergence_analysis}.  The convergence rate \eqref{eq:thm:nonconvex_rate} shows that \suda~enjoys linear speedup since the dominating term is $O(1/\sqrt{nK})$ for sufficiently large $K$  \cite{lian2017can}.   Note that the rate given in \eqref{eq:thm:nonconvex_rate} treats the network quantities $\{\gamma,\lambda_a,\underline{\lambda_b},v_1,v_2\}$  as constants. However, these quantities can have a significant influence on the convergence rate as explained in the introduction. While the above result holds for EXTRA, ED/D$^2$ and GT methods, it is still unclear whether these methods can achieve enhanced transient time compared to \textsc{Dsgd}.  To show the network affect and see the implication of Theorem \ref{thm_nonconvex} in terms of network quantities, we will specialize Theorem \ref{thm_nonconvex} to ED/D$^2$ \eqref{exact_diff} and ATC-GT \eqref{GT_atc_alg} and reflect the values of the $\{\gamma,\lambda_a,\underline{\lambda_b},v_1,v_2\}$  in the convergence rate.

\begin{corollary}[\bfseries \small ED/D$^2$ convergence] \label{corollary:ED}
Suppose that all the conditions given in Theorem \ref{thm_nonconvex} hold and assume further that $W$ is positive semi-definite.   If we set $\alpha=\sqrt{n/K}$ and $
K \geq \max \left\{ 4 nL^2, 
~\frac{32L^2  \lambda^2 n}{ (1-\sqrt{\lambda})^2 \underline{\lambda}}  ,~  \frac{\sqrt{32}   L^2  \lambda n}{  \sqrt{1-\lambda} (1-\sqrt{\lambda} ) \sqrt{\underline{\lambda}}}  \right\}$, then  ED/D$^2$ \eqref{exact_diff}  with $\x^0=\one \otimes x^0$ ($x^0 \in \real^d$) has convergence rate
\begin{equation} \label{ED_nonconvex_result}
\begin{aligned} 
	\frac{1}{K}	 \sum_{k=0}^{K-1} \Ex  \| \overline{\grad\f}(\x^k)\|^2 	&\leq 
	  	O \bigg(  \frac{ f(x^{0}) - f^\star }{ \sqrt{n K}}	 
	  	+\frac{   \sigma^2 }{ \sqrt{nK}}  \bigg)
	 	 \\
	 	 & \quad  
	 	 + O \bigg( \frac{n  \lambda^2     \sigma^2}{ (1-\lambda) \underline{\lambda} K}
	 	  	 + \frac{n  \lambda^2     \sigma^2}{ (1-\lambda)^{3} \underline{\lambda} K^2}
	 	  	  +    \frac{   n  L^2  \zeta_0^2}{  (1-\lambda)^2  \underline{\lambda} K^2} \bigg),
	\end{aligned}
	\end{equation}
	where $\underline{\lambda}$ is the minimum non-zero eigenvalue of $W$.	Moreover, we have 
	\begin{align*}
	\zeta_0^2 =\textstyle \frac{1}{n} \sum_{i=1}^n \big\|\sum\limits_{j \in \cN_i} w_{ij} \grad f_j(x^0)-\grad f(x^0) \big\|^2 \leq   \lambda^2  \varsigma_0^2,
	\end{align*}
	  where $\varsigma_0^2 \define \tfrac{1}{n} \sum_{i=1}^n \big\| \grad f_i(x^0)-\grad f(x^0) \big\|^2$.
\begin{proof}
 The proof  follows by substituting  the bounds \eqref{exact_diff_bounds} established in Appendix \ref{app:bounds_special_cases} into \eqref{eq:thm:nonconvex}.
\end{proof}
\end{corollary}
\begin{corollary}[\bfseries \small ATC-GT convergence] \label{corollary:GT}
Suppose that all the conditions given in Theorem \ref{thm_nonconvex} hold and assume further that $W$ is positive semi-definite.  If we let $\alpha=\sqrt{n/K}$ and $
K \geq \max \left\{ 4nL^2,~ \frac{432 L^2  \lambda^4 n}{(1-\lambda)^2}
,~ \frac{72 L^2  \lambda^2 n}{(1-\lambda)^2} 
\right\}$, then ATC-GT \eqref{GT_atc_alg} with  $\x^0=\one \otimes x^0$ ($x^0 \in \real^d$) has convergence rate
\begin{equation} \label{ATC_GT_nonconvex_result}
\begin{aligned}
	\frac{1}{K}	 \sum_{k=0}^{K-1}    \Ex  \| \overline{\grad\f}(\x^k)\|^2 		 
	&\leq 
	  	O \bigg(  \frac{ f(x^{0}) - f^\star }{ \sqrt{n K}}	  
	 	   + 	\frac{   \sigma^2 }{ \sqrt{nK}}  \bigg)
	 	   \\	 	   
	 	  	 & \quad + O\bigg( \frac{n  \lambda^4     \sigma^2}{ (1-\lambda)K} 
	 	  	 + \frac{n  \lambda^4     \sigma^2}{ (1-\lambda)^4 K^2} 
	 	  	  +    \frac{   n    \zeta_0^2}{  (1-\lambda)^3 K^2} \bigg).
	\end{aligned}
	\end{equation}	
	Moreover, we have 
	\begin{align*}
	\zeta_0^2 = \textstyle \frac{1}{n} \sum_{i=1}^n \big\|\sum\limits_{j \in \cN_i} [w_{ij}]^2 \grad f_j(x^0)-\grad f(x^0) \big\|^2 \leq   \lambda^4 \varsigma_0^2,
	\end{align*}
	  where $\varsigma_0^2 \define \tfrac{1}{n} \sum_{i=1}^n \big\| \grad f_i(x^0)-\grad f(x^0) \big\|^2$.
\begin{proof}
 The proof  follows by substituting the bounds \eqref{gt_diff_W_bound} established in Appendix \ref{app:bounds_special_cases} into Theorem \ref{thm_nonconvex}.
\end{proof}
\end{corollary}
\noindent 
Note that linear speedup  is achieved when the dominating term is $O(\tfrac{1}{\sqrt{nK}})$. This is case when $K$ is sufficiently large so that the higher order terms are less than or equal to $O(\tfrac{1}{\sqrt{nK}})$. For example, for ED/D$^2$ bound in \eqref{ED_nonconvex_result}, linear speedup is achieved when
\begin{align*}
	 & O \bigg( \hspace{-0.5mm}\frac{n  \lambda^2     }{ (1-\lambda)  K}
	 	  	 \hspace{-0.5mm}+\hspace{-0.5mm} \frac{n  \lambda^2     }{ (1-\lambda)^{3}  K^2}
	 	  	  \hspace{-0.5mm}+\hspace{-0.5mm}    \frac{   n  }{  (1-\lambda)^2 K^2} \bigg)
	 	  	  \leq O \bigg(  \frac{ 1 }{ \sqrt{n K}}	 
	    \bigg).
\end{align*}
The above holds when $K \geq O\left(n^3/(1-\lambda)^2 \right)$. 
Here, we treated $\underline{\lambda}$ for ED/D$^2$ as constant since, for example, if we set $W \leftarrow (1-\theta)W+\theta I$ for constant $\theta>0$, then it holds that $\underline{\lambda} \in (\theta, 1)$. Table \ref{table_non_convex} compares the our results with existing works. It is clear that our bounds are tighter in terms of the spectral gap $1-\lambda$. Moreover, ED/D$^2$ and GT methods have enhanced transient time compared to \textsc{Dsgd}.

\begin{remark}[\sc Step size selection] \label{remark:stepsize} \rm
We remark that in Corollaries  \ref{corollary:ED} and \ref{corollary:GT}, the step size is chosen as $\alpha=\sqrt{n/K}$ to simplify the expressions. This choice is not optimal and tighter rates can be obtained if we meticulously select the step size. For example, we can select the step size as in Theorem \ref{thm_nonconvex} and obtain a rate similar to \eqref{eq:thm:nonconvex_rate} where the dominating term is $O(\sigma/\sqrt{nK})$ instead of $O(\sigma^2/\sqrt{nK})$. We can even get a tighter rate by carefully selecting the step size similar to \cite{koloskova2020unified}.  However, such choices does not affect the transient time order in terms of network quantities, which is the main conclusion of our results. \qd

\end{remark} 
\begin{remark}[\sc EXTRA and other GT variations] \rm
We remark that the matrix $\G$ defined in \eqref{G_matrix} is identical for both EXTRA \eqref{EXTRA} and  ED/D$^2$. Hence, the convergence rate of ED/D$^2$ given in \eqref{ED_nonconvex_result} also holds for EXTRA with the exception of the value of $\lambda^2$ in the {\em numerators}, which should be replaced by  one for EXTRA ($\lambda^2 \to 1$  in numerators).  Likewise, for the non-ATC-GT \eqref{GT_nonatc_alg} and semi-ATC-GT (see Appendix \ref{app:relation_to_other_methods}), the matrix $\G$ is identical to ATC-GT and the convergence rate of ATC-GT given \eqref{ATC_GT_nonconvex_result} holds for these other variations except for the value of $\lambda^2$ in the {\em numerators}, which is one for non-ATC-GT ($\lambda^2 \to 1$  in numerators) and $\lambda$ for the semi-ATC-GT ($\lambda^2 \to \lambda$  in numerators). Please see Appendix \ref{app:bounds_special_cases} for more details.

 Finally, note that for static and undirected graphs, our technique can be used to improve  the network bounds for EXTRA, ED/D$^2$, and  GT modifications for other cost-function settings such variance-reduced settings \cite{xin21hybrid}. 
 \qd
\end{remark}

\subsection{Convergence Under PL Condition}
We next state the convergence of \suda~under the PL condition given in Assumption \ref{assump:PL}.
 \begin{theorem} [\bfseries \small  PL case] \label{thm_PL_linear_conv}
	Suppose that Assumptions \ref{assump:network}--\ref{assump:PL} hold and the step size satisfies:
	\begin{align}  \label{pl_thm_ineq_constant}
\alpha \leq \min\left\{ 
\frac{1-\gamma}{3L},~
\frac{ \underline{\lambda_b}}{2  L },~
\frac{1-\gamma}{\sqrt{6} L v_1 v_2 \lambda_a}
,~
\left(\frac{\mu\underline{\lambda_b}^{2}(1-\gamma)}{8 L^{4} v_1^{2} v_2^{2} } \right)^{1/3}
\right\}.
\end{align}
Then, the iterates $\{\x^k\}$ of \suda~ with $\x^0=\one \otimes x^0$ ($x^0 \in \real^d$) satisfy
\begin{equation}  \label{eq:PL_thm_1}
\begin{aligned}
\frac{1}{n}\sum_{i=1}^n \Ex [f(x_i^{k}) - f^\star]
& \leq  
\left(1-\frac{\alpha \mu}{2} \right)^k r_0
\\
& \quad
+O\left(\frac{\alpha L \sigma^2}{\mu n}
+\frac{\alpha^2 L^2 v_1^2 v_2^2 \lambda_a^2  \sigma^2}{\mu (1-\gamma)}
+\frac{ \alpha^4 L^4 v_1^2 v_2^2 \lambda_a^2  \sigma^2}{  \mu \underline{\lambda_b}^2(1-\gamma)^2 n } \right),
\end{aligned}
\end{equation}
where $r_0 =2\Ex [f(x^{0})-f^\star] +
\alpha^2 L v_1^2 v_2^2 \zeta_0^2/ \underline{\lambda_b}^2$ and the quantities $v_1$, $v_2$, $\gamma$, $\lambda_a$, $\underline{\lambda_b}$ and $\zeta_0$ are defined as in Theorem \ref{thm_nonconvex}. 
\qd
\end{theorem}
\noindent The proof of Theorem \ref{thm_PL_linear_conv} is given in Appendix \ref{app:thm_pl_lin_conv_proof}. To discuss the implication of the above result, we specialize it to  ED/D$^2$ and ATC-GT.

\begin{corollary}[\bfseries \small  ED/D$^2$  convergence under PL condition] \label{corollary:pl_ED}
\sloppy Let the same conditions as in Theorem \ref{thm_PL_linear_conv} hold and suppose further that $W$ is positive semi-definite. If the step size satisfies $\alpha \leq \min\left\{ 
\frac{1-\sqrt{\lambda}}{3L},~
\frac{ \sqrt{1-\lambda}}{2  L },~
\frac{(1-\sqrt{\lambda}) \underline{\lambda} }{2\sqrt{12}  L \lambda}
,~
\left(\frac{\mu (1-\lambda)(1-\sqrt{\lambda}) \underline{\lambda}}{18 L^{4}  } \right)^{1/3} 
\right\}$, then ED/D$^2$ \eqref{exact_diff}  with $\x^0=\one \otimes x^0$ ($x^0 \in \real^d$) has convergence rate 
\begin{equation} \label{ED_PL_result}
\begin{aligned}
 \frac{1}{n}\sum_{i=1}^n \Ex [f(x_i^{k})-f^\star] &\leq 
 (1-\tfrac{\alpha \mu}{2})^k r_0 +
O\left(\frac{\alpha L \sigma^2}{\mu n}
+\frac{\alpha^2 L^2  \lambda^2  \sigma^2}{\mu (1-\lambda)\underline{\lambda} }
+\frac{ \alpha^4 L^4  \lambda^2  \sigma^2}{  \mu (1-\lambda)^3 \underline{\lambda} n } \right),
\end{aligned}
\end{equation}
where $r_0 =2\Ex [f(x^{0})-f^\star] +
8 \alpha^2 L  \zeta_0^2/ (1-\lambda) \underline{\lambda}$,  $\zeta_0^2 = \textstyle \frac{1}{n} \sum_{i=1}^n \big\|\sum\limits_{j \in \cN_i} w_{ij} \grad f_j(x^0)-\grad f(x^0) \big\|^2$, and $\underline{\lambda}$ is the minimum non-zero eigenvalue of $W$. Hence, selecting $\alpha=2\ln( K^2)/\mu K$, we obtain
     \begin{align} \label{rate_ED_pl}
\frac{1}{n}\sum_{i=1}^n \Ex [f(x_i^{K}) - f^\star] 
& \leq  
\frac{2\Ex [f(x^{0})-f^\star]}{ K^2}  
 \nonumber \\
 & \quad + \tilde{O}\left( 
\frac{  \sigma^2}{K  n}
 +\frac{   \lambda^2  \sigma^2}{K^2  (1-\lambda) \underline{\lambda}}
 + \frac{  \zeta_0^2}{K^4  (1-\lambda) \underline{\lambda} } 
 +
\dfrac{    \lambda^2  \sigma^2}{   K^4 (1-\lambda)^3 \underline{\lambda} n } \right),
\end{align}
where   $\tilde{O}(\cdot)$ hides logarithmic factors. 
 \begin{proof}
 Equation \eqref{ED_PL_result} follows from Theorem \ref{thm_PL_linear_conv} and the bounds \eqref{exact_diff_bounds}  derived in Appendix \ref{app:bounds_special_cases}. Now, if we set $\alpha=2\ln( K^2)/\mu K$, then $
    1-\tfrac{\alpha \mu}{2} \leq \exp(- \alpha \mu K/2) =\frac{1}{K^2}$ where $\exp(\cdot)$ denote the exponential function. Hence, for $\alpha=2\ln(K^2)/\mu K$ and large enough $K$, inequality \eqref{ED_PL_result} can be upper bounded by \eqref{rate_ED_pl}.
\end{proof}
\end{corollary}
\begin{corollary}[\bfseries \small  ATC-GT convergence under PL condition] \label{corollary:pl_GT}
\sloppy Let the same conditions as in Theorem \ref{thm_PL_linear_conv} hold and assume further that $W$ is positive semi-definite. Then, ATC-GT \eqref{GT_atc_alg} with  $\x^0=\one \otimes x^0$ ($x^0 \in \real^d$) and  $\alpha \leq \min\left\{ 
\frac{1-\lambda}{6L}
,~
\frac{1-\lambda}{ 6\sqrt{18} L \lambda^2 }
,~
\left(\frac{\mu (1-\lambda)^3}{432 L^{4}  } \right)^{1/3}
\right\}$  has convergence rate:
\begin{equation} \label{ATC_GT_PL_result}
\begin{aligned}
 \frac{1}{n}\sum_{i=1}^n \Ex [f(x_i^{k})-f^\star] &\leq 
 (1-\tfrac{\alpha \mu}{2})^k r_0
 +
O\left(\frac{\alpha L \sigma^2}{\mu n}
+\frac{\alpha^2 L^2  \lambda^4  \sigma^2}{\mu (1-\lambda)}
+\frac{ \alpha^4 L^4  \lambda^4  \sigma^2}{  \mu (1-\lambda)^4 n } \right).
\end{aligned}
\end{equation}
where $r_0 =2\Ex [f(x^{0})-f^\star] +
27 \alpha^2 L  \zeta_0^2/ (1-\lambda)^2$ and $\zeta_0^2 = \textstyle \frac{1}{n} \sum_{i=1}^n \big\|\sum\limits_{j \in \cN_i} [w_{ij}]^2 \grad f_j(x^0)-\grad f(x^0) \big\|^2$.
Hence, if we set $\alpha=2\ln( K^2)/\mu K$, then it holds that
      \begin{align} \label{rate_GT_pl}
\frac{1}{n}\sum_{i=1}^n \Ex [f(x_i^{K}) - f^\star] 
 &\leq  
\frac{2\Ex [f(x^{0})-f^\star]}{K^2}  
 \nonumber \\
 & \quad + \tilde{O}\left( 
\frac{  \sigma^2}{K  n}
 +\frac{   \lambda^4  \sigma^2}{K^2  (1-\lambda)}
 + \frac{   \zeta_0^2}{K^4  (1-\lambda)^2  }
 +
\frac{   \lambda^4  \sigma^2}{   K^4 (1-\lambda)^4 n } \right),
\end{align}
where 	 $\tilde{O}(\cdot)$ hides logarithmic factors. 
 \begin{proof}
Equation \eqref{ATC_GT_PL_result} follows from Theorem \ref{thm_PL_linear_conv} and the bounds \eqref{gt_diff_W_bound}  derived in Appendix \ref{app:bounds_special_cases}. Inequality \eqref{rate_GT_pl} follows by subsisting  $\alpha=2\ln(K^2)/\mu K$, into \eqref{gt_diff_W_bound}  and using  $
    1-\tfrac{\alpha \mu}{2} \leq \exp(- \alpha \mu K/2) =\frac{1}{K^2}$.
\end{proof}
\end{corollary}
\noindent We note that as in Remark \ref{remark:stepsize}, the selected step sizes in Corollaries  \ref{corollary:pl_ED} and \ref{corollary:pl_GT}  are not optimized. The hidden log factors can be removed  if we adopt decaying step-sizes techniques (\eg,  \cite{pu2019sharp}). However, the main conclusion we want to emphasize is the network dependent bounds, which do not change if select better step size choices.    Under the PL condition,  linear speedup  is achieved when $K$ is large enough such that the dominating term is $O(\tfrac{1}{nK})$. Table \ref{table_pl} lists the transient times implied by the above result and compares them with existing results. It is clear that our results significantly improves upon existing GT results. Moreover, our bound for ED/D$^2$ matches the existing bound, which are  under the stronger assumption of strong-convexity. 
\begin{remark}[\sc Steady-state error] \rm 
For constant step size $\alpha$ independent of $k$, we can let $k$ goes to $\infty$ in \eqref{ED_PL_result} and \eqref{ATC_GT_PL_result} to arrive at the following steady state results.
\begin{itemize}
\item For ED/D$^2$ \eqref{exact_diff}, we have
\begin{equation} \label{ED_PL_ss_result}
\begin{aligned}
\limsup_{k \rightarrow \infty}  \frac{1}{n}\sum_{i=1}^n \Ex [f(x_i^{k})-f^\star] &\leq 
O\left(\frac{\alpha L \sigma^2}{\mu n}
+\frac{\alpha^2 L^2  \lambda^2  \sigma^2}{\mu (1-\lambda) }
+\frac{ \alpha^4 L^4  \lambda^2  \sigma^2}{  \mu (1-\lambda)^3  n } \right).
\end{aligned}
\end{equation}
	  
\item For ATC-GT \eqref{GT_atc_alg}, we have
\begin{equation} \label{ATC_GT_PL_ss_result}
\begin{aligned}
\limsup_{k \rightarrow \infty}  \frac{1}{n}\sum_{i=1}^n \Ex [f(x_i^{k})-f^\star] &\leq 
O\left(\frac{\alpha L \sigma^2}{\mu n}
+\frac{\alpha^2 L^2  \lambda^4  \sigma^2}{\mu (1-\lambda)}
+\frac{ \alpha^4 L^4  \lambda^4  \sigma^2}{  \mu (1-\lambda)^4 n } \right).
\end{aligned}
\end{equation}	
\end{itemize}
The bound \eqref{ED_PL_ss_result} for ED/D$^2$ has the same network dependent bounds as in the strongly-convex case \cite{yuan2020influence}. Moreover, the bound \eqref{ATC_GT_PL_ss_result} for ATC-GT improves upon existing bounds for both strongly-convex \cite{pu2021distributed} and PL settings \cite{xin2021improved}, which are on the order of $
O\big(\alpha^2  \lambda^2  \sigma^2/ (1-\lambda)^3\big)
$. 
\qd
\end{remark}

\section{Simulation Results}\label{sec:simu-nonconvex}
In this section, we validate the established theoretical results with numerical simulations. 
\subsection{Simulation for non-convex problems}
\noindent \textbf{The problem.} We consider the logistic regression problem with a non-convex regularization term \cite{antoniadis2011penalized,xin2021improved}.
The problem formulation is given by $\min_{x \in \real^d} \frac{1}{n}\sum_{i=1}^n f_i(x) + \rho\, r(x)$, where
\begin{equation} \label{non_convex_lr}
 \begin{aligned} 
f_i(x) = \frac{1}{L}\sum_{\ell=1}^L \ln\big(1 + \exp(-y_{i,\ell}h_{i,\ell}\tran x)\big) \quad \mbox{and} \quad r(x) = \sum_{j=1}^d \frac{x(j)^2}{ 1 + x(j)^2}.
\end{aligned}
\end{equation}
In the above problem, $x=\col\{x(j)\}_{j=1}^d \in \real^d$ is the unknown variable to be optimized, $\{h_{i,\ell}, y_{i,\ell}\}_{\ell=1}^L$ is the training dateset held by agent $i$ in which $h_{i,\ell}\in \mathbb{R}^d$ is a feature vector while $y_{i,\ell} \in \{-1,+1\}$ is the corresponding label. The regularization $r(x)$ is a smooth but non-convex function and the regularization constant $\rho > 0$ controls the influence of $r(x)$. 

\vspace{1mm}
\noindent \textbf{Experimental settings.} In  our experiments, we set $d=20$, $L=2000$ and $\rho = 0.001$. To control data heterogeneity across the agents, we first let each agent $i$ be associated with a local solution $x^\star_{i}$, and such $x^\star_i$ is generated by $x^\star_i = x^\star + v_i$ where $x^\star\sim \cN(0, I_d)$ is a randomly generated vector while $v_i \sim \cN(0, \sigma^2_h I_d)$ controls the similarity between each local solution. Generally speaking, a large $\sigma^2_h$ result in local solutions  $\{x_i^\star\}$ that are vastly different from each other. With $x_i^\star$ at hand, we can generate local data that follows distinct distributions. At agent $i$, we generate each feature vector $h_{i,\ell} \sim \cN(0, I_d)$. To produce
the corresponding label $y_{i,\ell}$, we generate a random variable $z_{i,\ell} \sim \cU(0,1)$. If $z_{i,\ell} \le 1 + \exp(-y_{i,\ell}h_{i,\ell}\tran x_i^\star)$, we set $y_{i,\ell} = 1$; otherwise $y_{i,\ell} = -1$. Clearly, solution $x_i^\star$ controls the distribution of the labels. In this way, we can easily control data heterogeneity by adjusting $\sigma^2_h$. Furthermore, to easily control the influence of gradient noise, we will achieve the stochastic gradient by imposing a Gaussian noise to the real gradient, \ie, $\widehat{\nabla f}_i(x) = {\nabla f}_i(x) + s_i$ in which $s_i\sim \cN(0, \sigma^2_{n} I_d)$. We can control the magnitude of the gradient noise by adjusting $\sigma^2_n$. The metric for all simulations is $\Ex \|\nabla f(\bar{x})\|^2$ where $\bar{x}=\frac{1}{n}\sum_{i=1}^n x_i$.

\vspace{1mm} 
\noindent \textbf{Performances of SUDA with and without data heterogeneity.} In this set of simulations, we will test the performance of ED/D$^2$ and ATC-GT (which are covered by the SUDA framework) with constant and decaying step size, and compare them with \textsc{Dsgd}. In the simulation, we organize $n=32$ agents into an undirected ring topology.

Fig.~\ref{fig:peformance-const-decay-homo} shows the performances of these algorithms with homogeneous data, \ie, $\sigma_h^2 = 0$. The gradient noise magnitude is set as $\sigma_n^2 =  0.001$. 
In the left plot, we set up a constant step size $\alpha = 0.01$. In the right plot, we set an initial step size as $0.01$, and then scale it by $0.5$ after every $100$ iterations. It is observed in Fig.~\ref{fig:peformance-const-decay-homo} that all stochastic algorithms perform similarly to each other with homogeneous data. 
Fig.~\ref{fig:peformance-const-decay-slight-heterogeneity} shows the performance under heterogeneous data settings with $\sigma_h^2 = 0.2$. The gradient noise and the step size values are the same as in Fig.~\ref{fig:peformance-const-decay-homo}. It is clear from Fig.~\ref{fig:peformance-const-decay-slight-heterogeneity} that ED/D$^2$ and ATC-GT are more robust to data heterogeneity compared to \textsc{Dsgd}. We see that  ED/D$^2$ can converge as well as \textsc{Psgd} while ATC-GT performs slightly worse than ED/D$^2$. 

\begin{figure}[t!]
	\centering
	\includegraphics[width=0.4\textwidth]{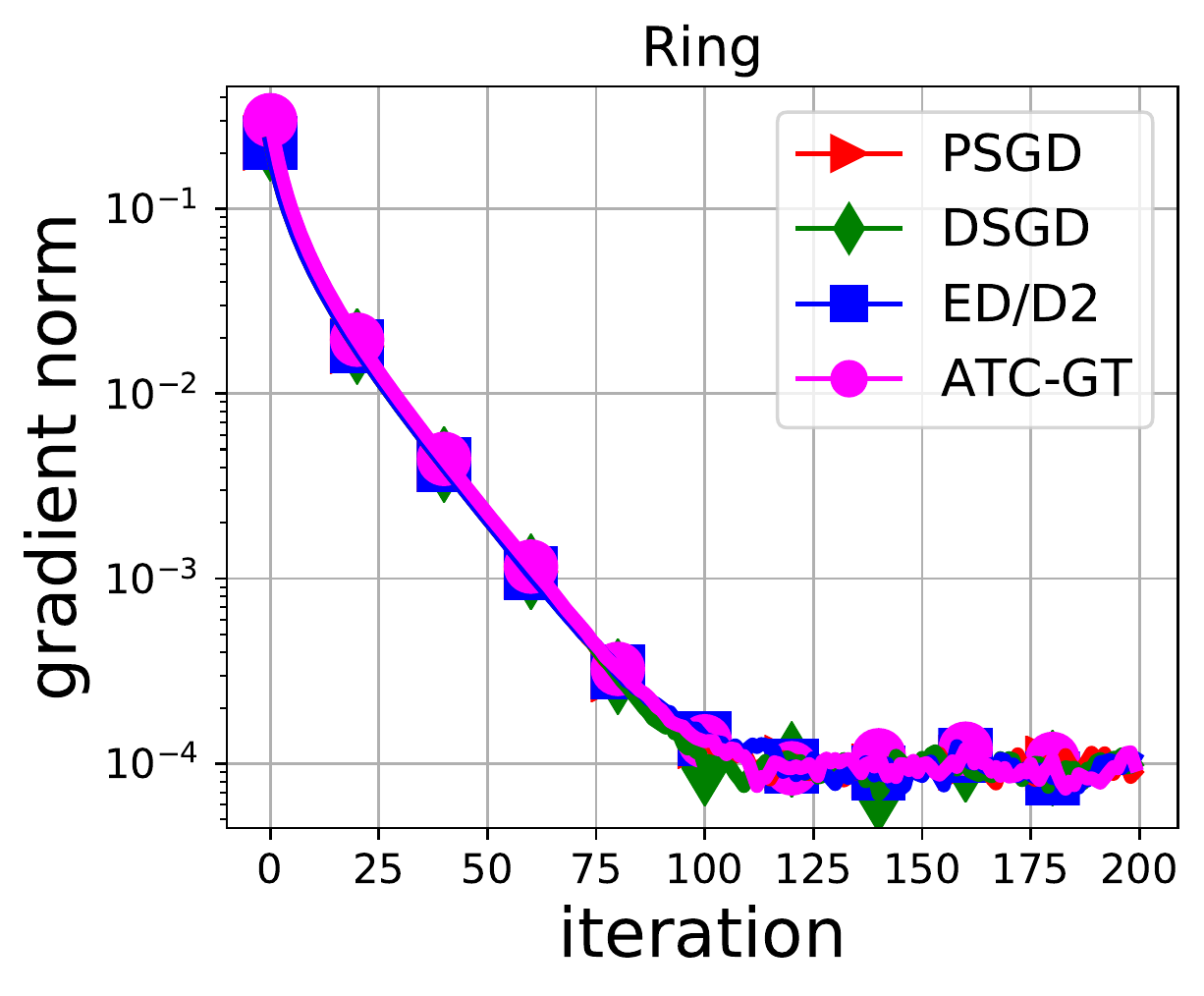}
	\hspace{1cm} 
	\includegraphics[width=0.4\textwidth]{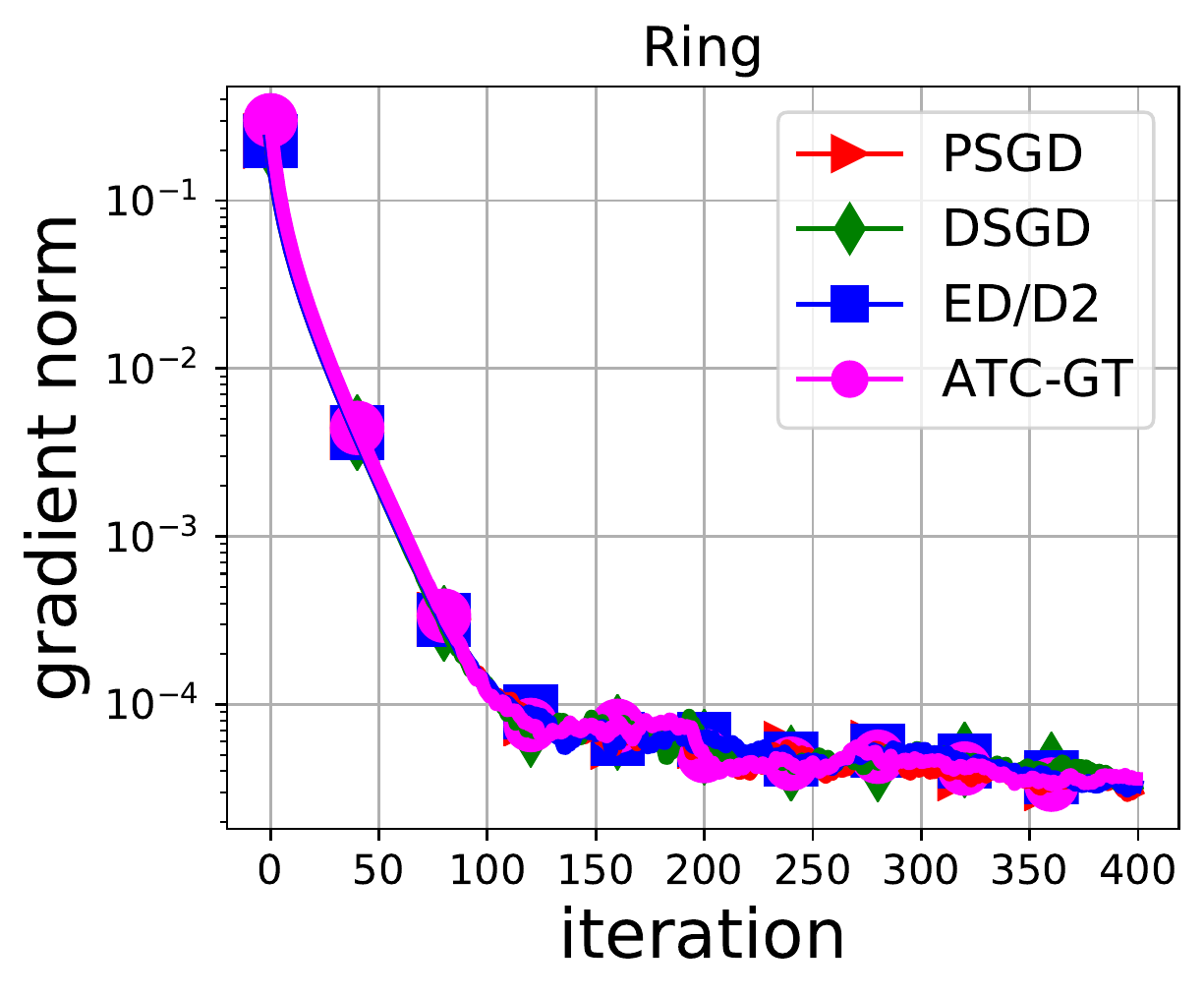}
	\caption{Performance of different stochastic algorithms to solve problem \eqref{non_convex_lr} with homogeneous data. Constant and learning  decaying rates are used in the left and right plots, respectively. All algorithms in both plots are over the ring topology with $\lambda = 0.99$.}
    \label{fig:peformance-const-decay-homo}
\end{figure}

\begin{figure}[t!]
	\centering
	\includegraphics[width=0.4\textwidth]{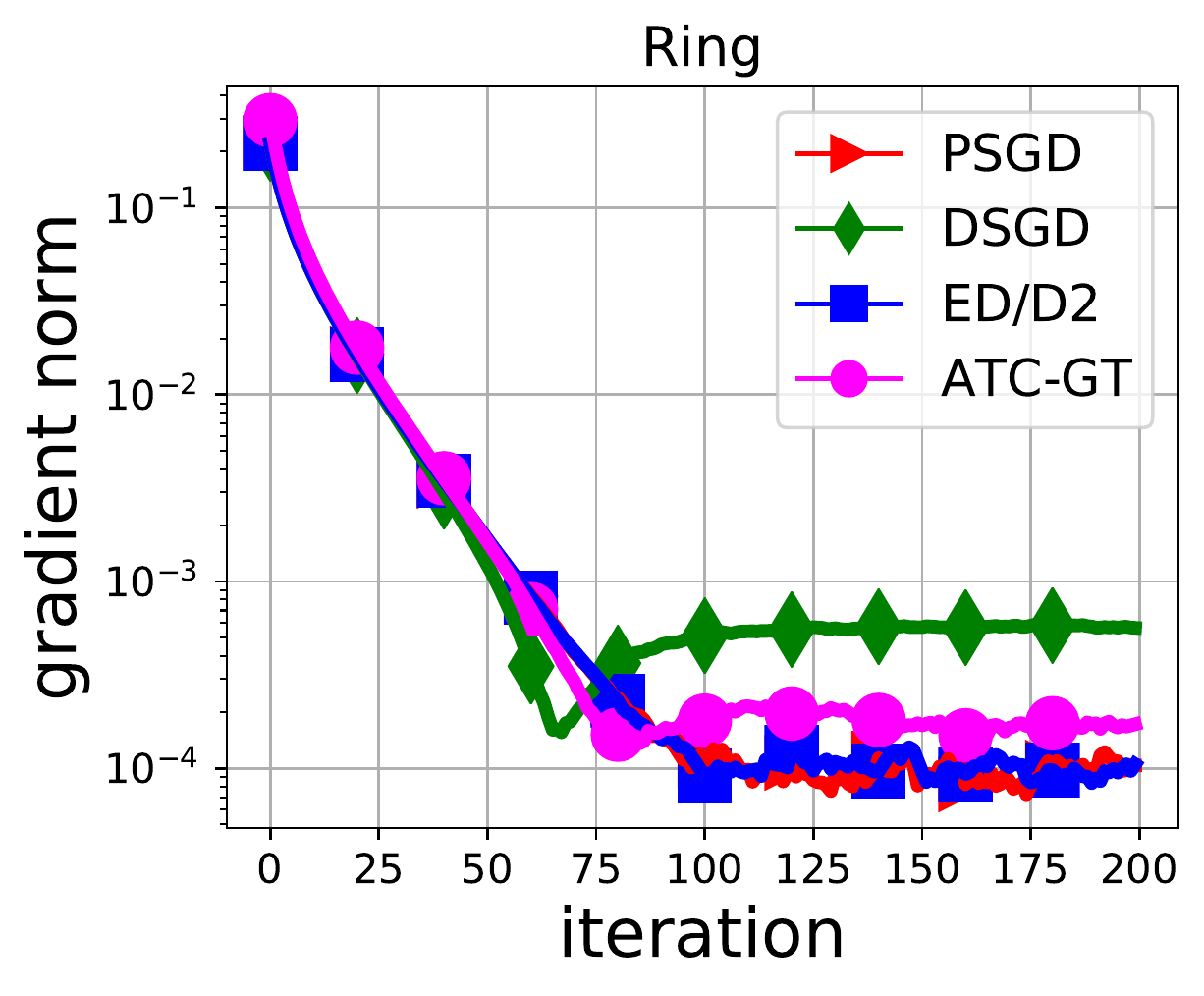}
	\hspace{1cm} 
	\includegraphics[width=0.4\textwidth]{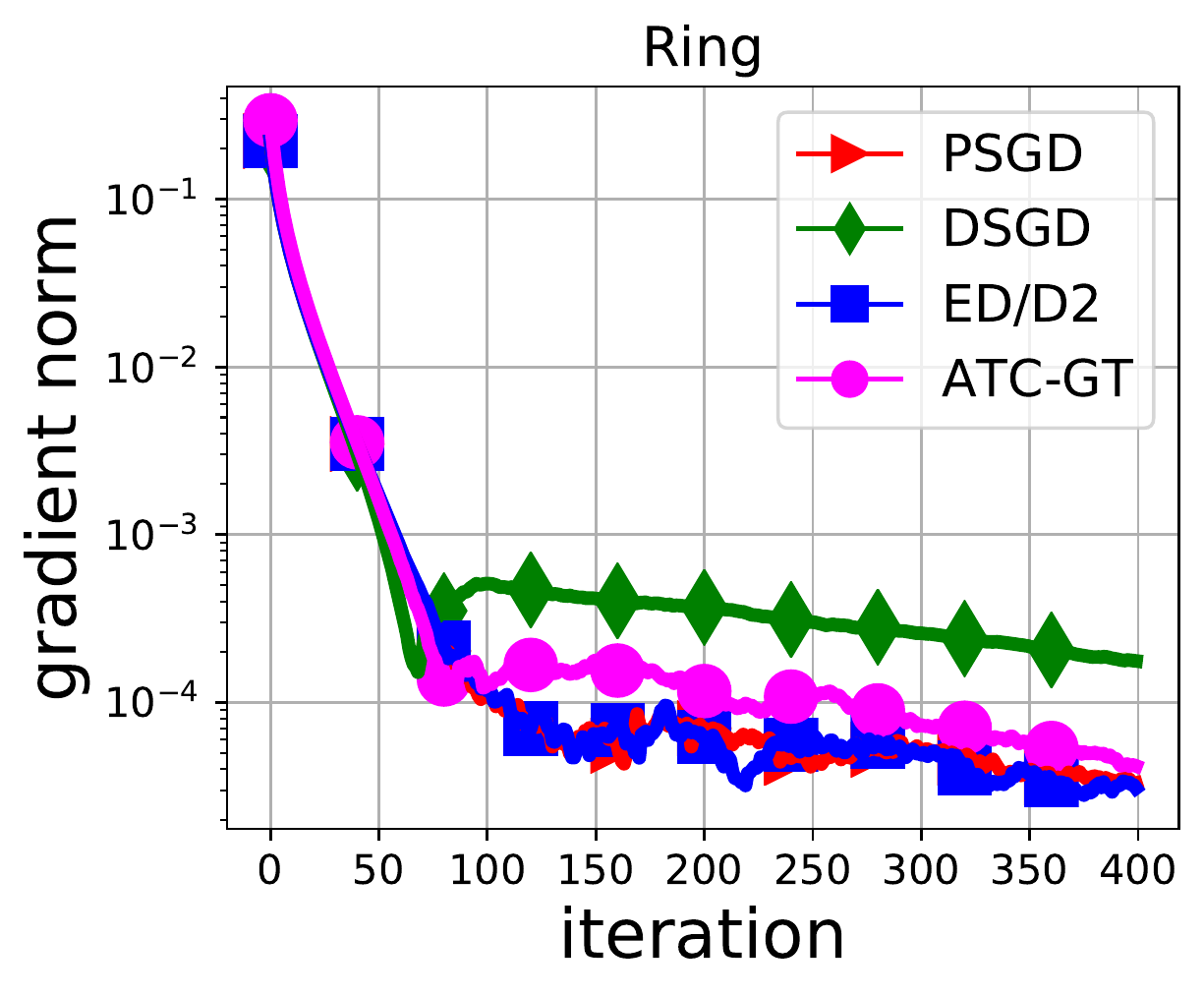}
	\caption{Performance of different stochastic algorithms to solve problem \eqref{non_convex_lr} with heterogeneous data. Constant and learning decaying rates are used in the left and right plots, respectively. All algorithms in both plots are over the ring topology with $\lambda = 0.99$. }
    \label{fig:peformance-const-decay-slight-heterogeneity}
\end{figure}

\vspace{1mm} 
\noindent \textbf{Influence of network topology.} In this set of simulations, we will test the influence of the spectral gap $1-\lambda$ on various decentralized stochastic algorithms. We generate four topologies: the Erdos-Renyi graph with probability $0.8$, the Ring topology, the Grid topology, and the scaled Ring topology with $\lambda = (9 + \lambda_{\rm Ring})/10$. The value of the mixing rate $\lambda$ for each topology is listed in the caption in Fig.~\ref{fig:influence_of_the_topology}. We utilize a constant step size $0.01$ for each plot. It is observed in Fig.~\ref{fig:influence_of_the_topology} that each decentralized algorithm will converge to a less accurate solution as $\lambda \to 1$ while \textsc{Psgd} is immune to the network topology. In addition, it is also observed that ED/D$^2$ is least sensitive to  the network topology while \textsc{Dsgd} is most sensitive (under heterogeneous setting), which is consistent with the our results listed in Table \ref{table_non_convex}.

\begin{figure}[t!]
	\centering
	\includegraphics[width=0.4\textwidth]{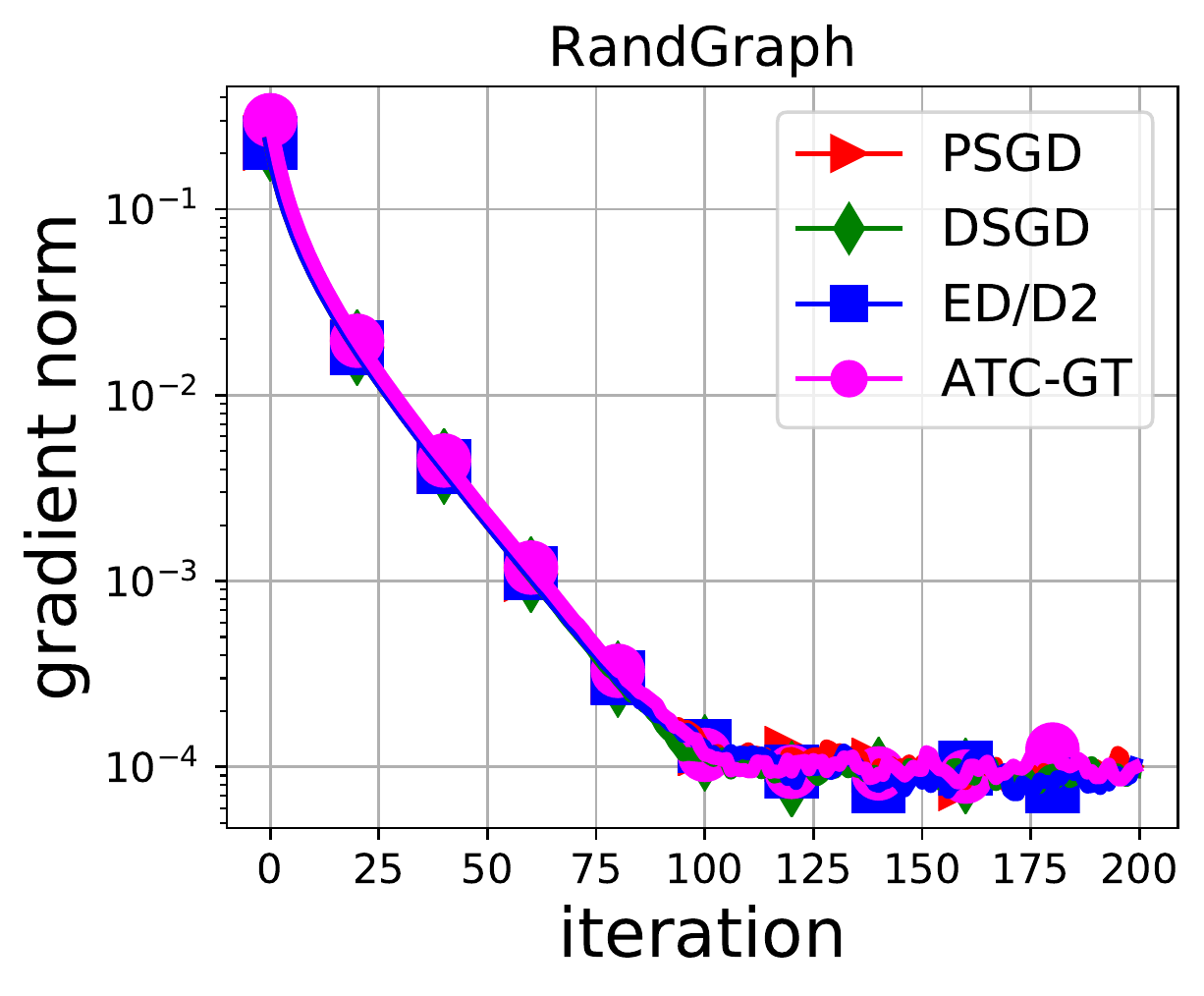} 
	\includegraphics[width=0.4\textwidth]{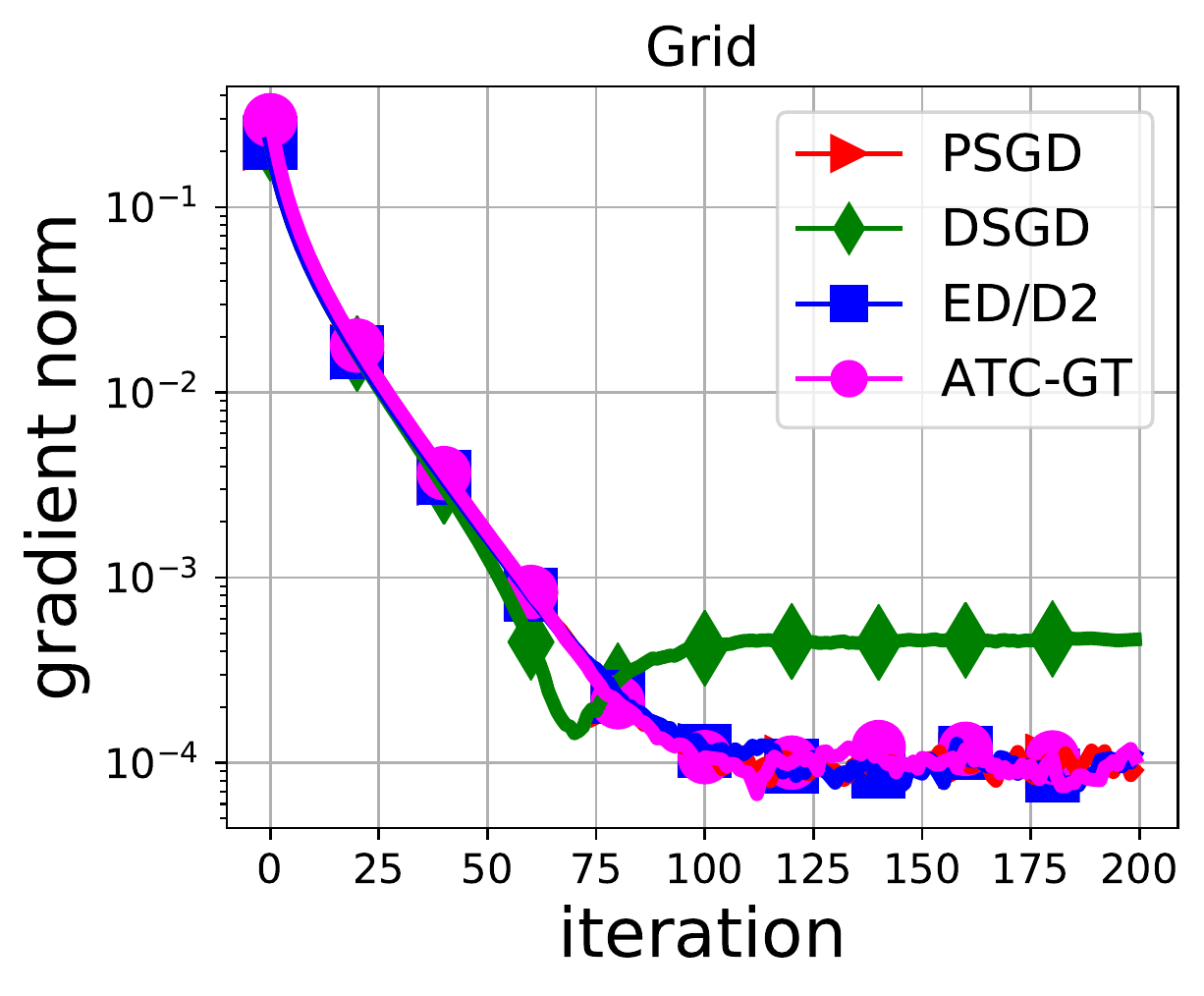}
	\includegraphics[width=0.4\textwidth]{figure/nonconvex_lr_ring_36nodes_const_10142021_changeLegend.pdf}
	\includegraphics[width=0.4\textwidth]{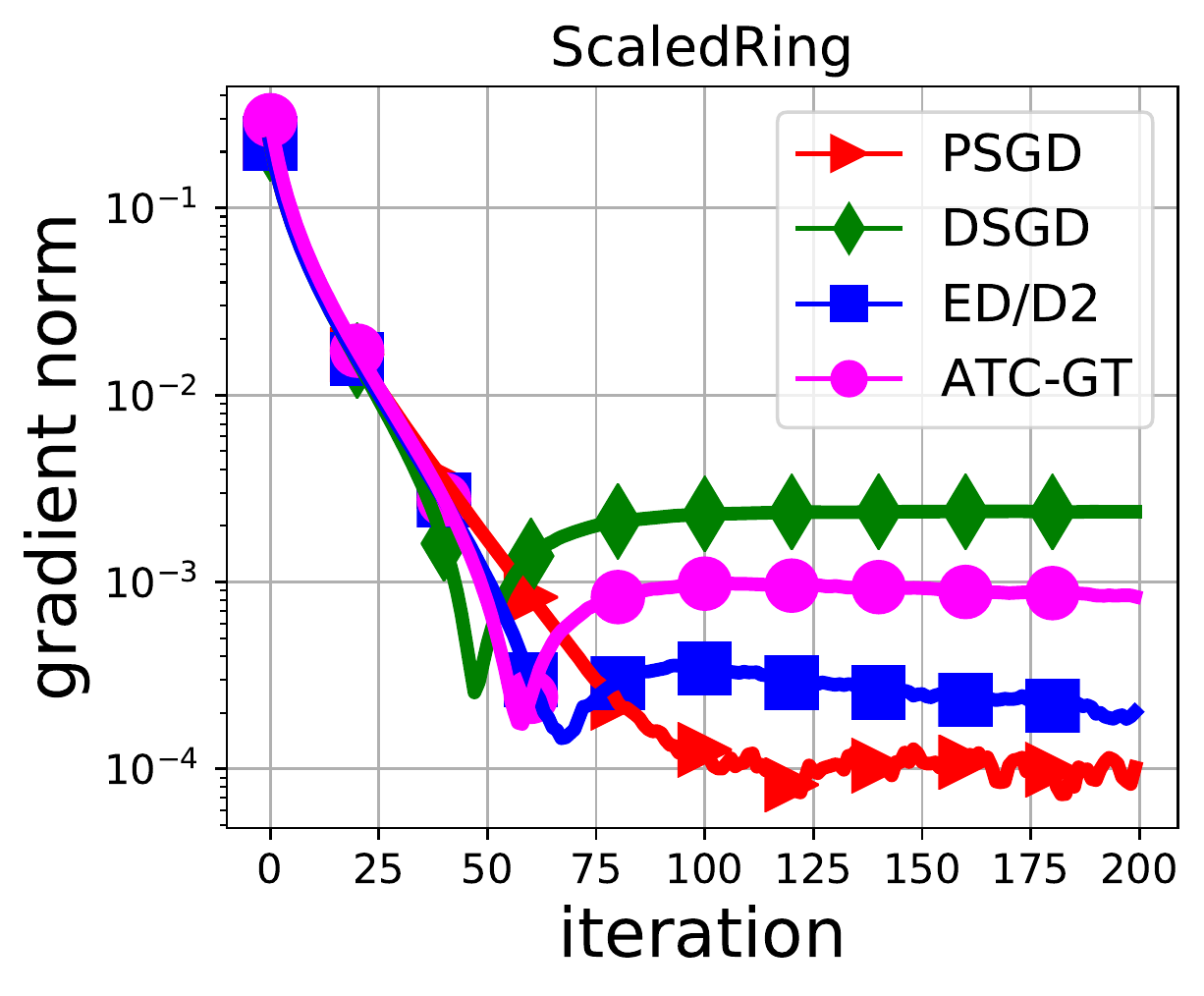}
	\caption{Performance of different stochastic algorithms to solve problem \eqref{non_convex_lr} with different topologies. Top-left: Erdos-Renyi random graph with $\lambda = 0.32$; Top-right: Grid  with $\lambda = 0.94$; Bottom-left: Ring  with $\lambda = 0.99$; Bottom-right: Scaled Ring  with $\lambda = 0.999$}
    \label{fig:influence_of_the_topology}
\end{figure}

\subsection{Simulation results under the PL condition}
\noindent \textbf{The problem.} In this section we examine the performance of ED/D$^2$ and GT algorithms for the non-convex problem under PL condition. We consider the same setup used in \cite{xin2021improved} where the problem formulation is given by $\min_{x \in \real^d} \frac{1}{n}\sum_{i=1}^n f_i(x)$ with $f_i(x) = x^2 + 3\sin^2(x) + a_i x \cos(x)$. By letting $\sum_{i=1}^n a_i = 0$, we have $f(x) = \frac{1}{n}\sum_{i=1}^n f_i(x) = x^2 + 3\sin^2(x)$ which is a non-convex cost function that satisfies the PL condition \cite{karimi2016linear}. 

\vspace{1mm}
\noindent \textbf{Experimental settings.} We set $n=32$ in all simulations. To generate $a_i$, we let $a_i = \sigma_h^2 \cdot i$ and $a_{n-i} = - a_i$ for $i \in \{1,\dots, n/2\}$ where $\sigma_h^2$ is used to control data heterogeneity. In this way, we can guarantee $\sum_{i=1}^n a_i = 0$. Similar to Sec.~\ref{sec:simu-nonconvex}, we will achieve the stochastic gradient by imposing a Gaussian noise to the real gradient, \ie, $\widehat{\nabla f}_i(x) = {\nabla f}_i(x) + s_i$ in which $s_i\sim \cN(0, \sigma^2_{n} I_d)$. The metric for all simulations is   $\Ex f(\bar{x})-f^\star$ where $\bar{x}=\frac{1}{n}\sum_{i=1}^n x_i$.

\vspace{1mm} 
\noindent \textbf{Influence of network topology.} We test the influence of the network topology on various stochastic decentralized methods. In simulations, we set data heterogeneity $\sigma_h^2 = 2$ and gradient noise $\sigma_n^2 = 0.1$. We generate two types of topologies: an Erdos-Renyi random graph with $\lambda = 0.28$, and an Erdos-Renyi random graph with $\lambda = 0.87$. We employ a constant step size $0.008$ for all tested algorithms. It is observed in Fig.~\ref{fig:PL-peformance} that the performance of all algorithms can be deteriorated by the badly-connected network topology when $\lambda \to 1$. We see that ED/D$^2$ is least sensitive to  the network topology while \textsc{Dsgd} is most sensitive (under heterogeneous setting), which is consistent with the our results listed in Table \ref{table_pl}.
\begin{figure}[t!]
	\centering
	\includegraphics[width=0.4\textwidth]{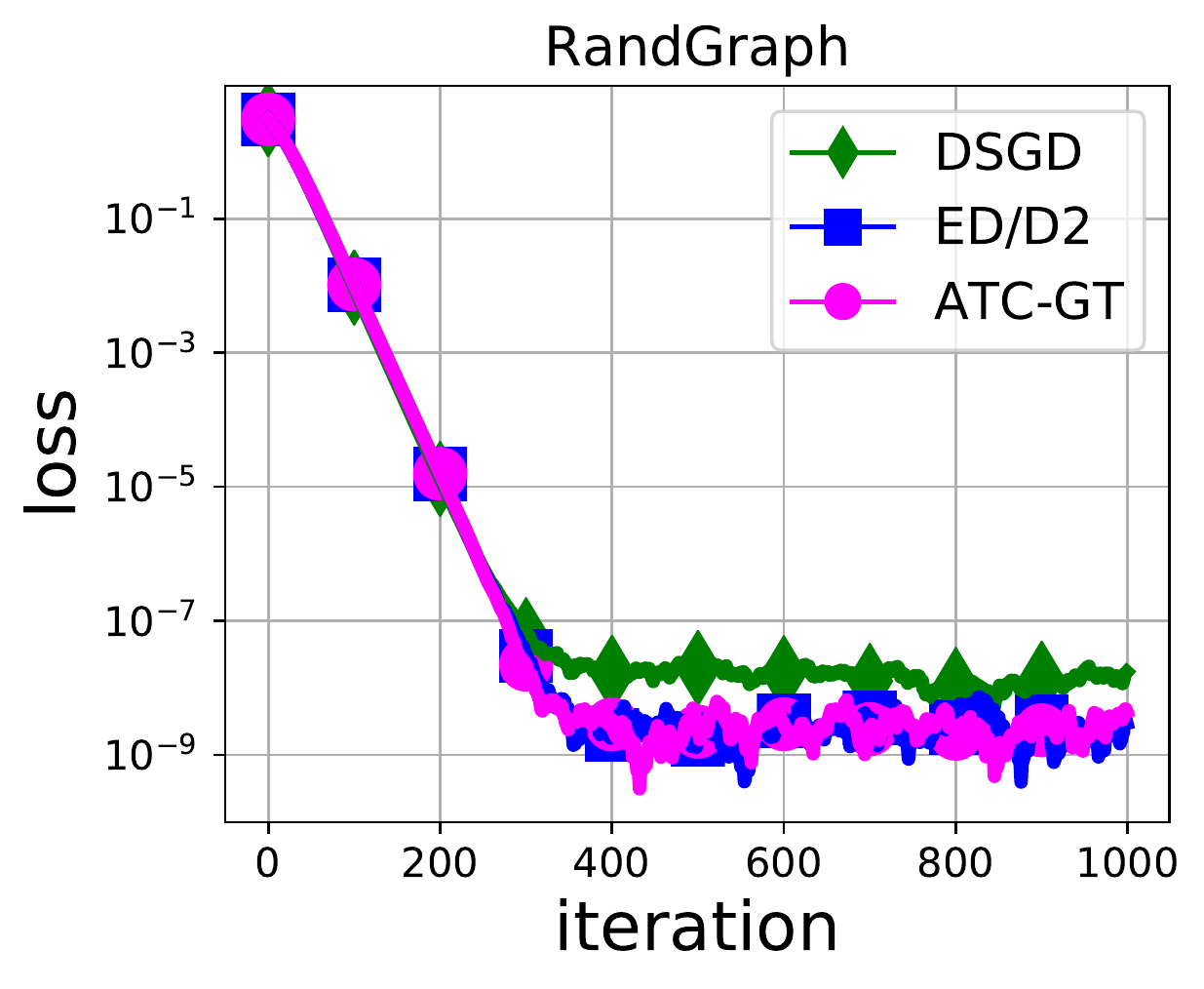}
	\hspace{1cm} 
	\includegraphics[width=0.4\textwidth]{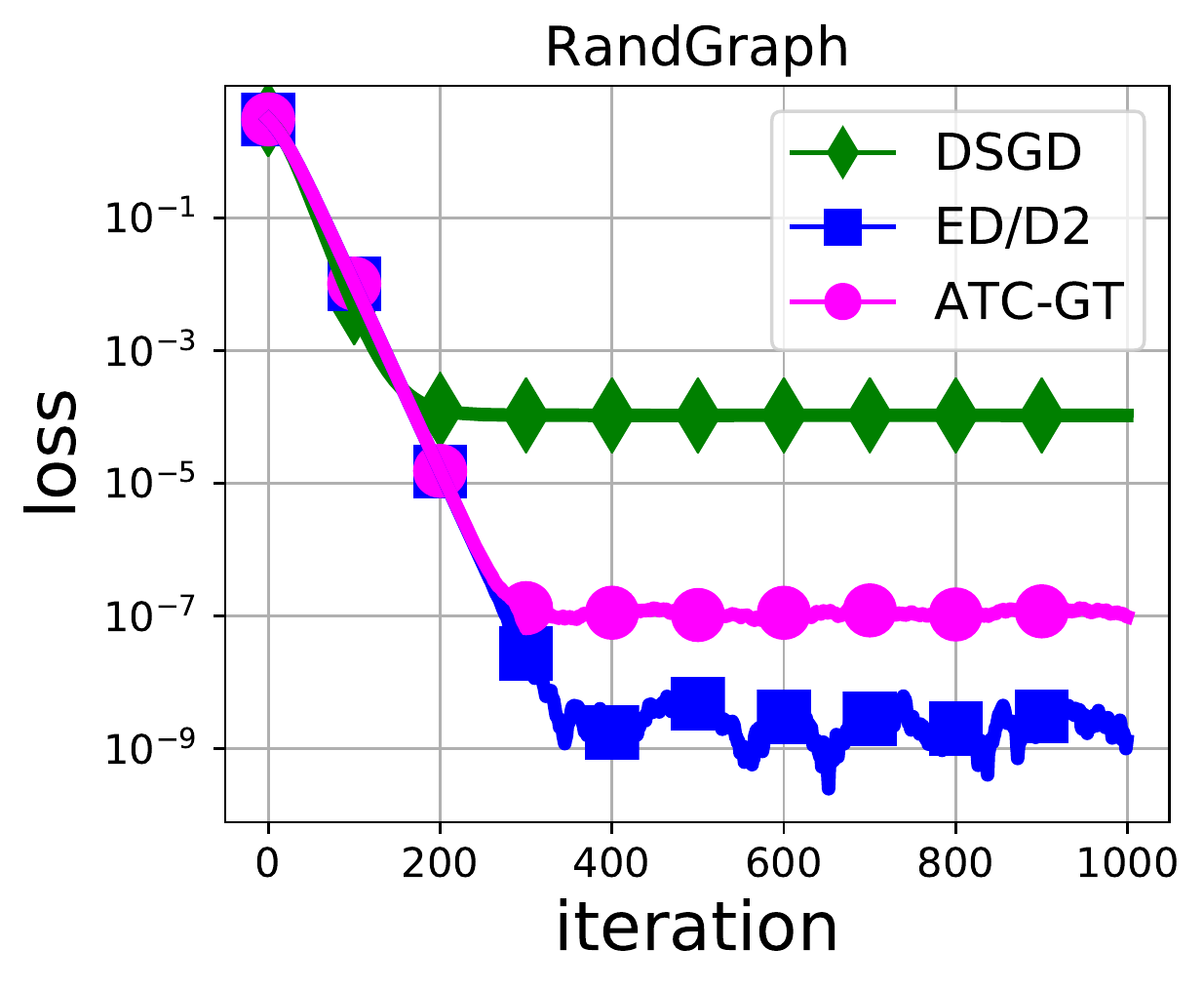}
	\caption{Performance of different stochastic algorithms to solve non-convex problem under PL condition with different topologies. Left: Erdos-Renyi random graph with $\lambda = 0.28$; Right: Erdos-Renyi random graph with $\lambda = 0.87$}
    \label{fig:PL-peformance}
\end{figure}

\section{Conclusion} 
In this work,  we analyzed the convergence properties of \suda~\eqref{SUDA_alg} for decentralized stochastic non-convex optimization problems. \suda~is a general algorithmic framework that includes several state of the art decentralized methods as special case such as EXTRA, Exact-Diffusion/D$^2$, and gradient-tracking methods. We established the convergence of \suda~under both general non-convex  and  PL condition settings. Explicit convergence rate bounds are provided in terms of the problem parameters and the network topology.  When specializing \suda~to the particular instances of ED/D$^2$, EXTRA, and GT-methods, we achieve improved network topology dependent rates compared to existing results under non-convex settings.  Moreover, our rate shows that ED/D$^2$, EXTRA, and GT-methods are less sensitive to network topology compared to  \textsc{Dsgd} under heterogeneous data setting.

	Finally, it should be noted that the lower bound from \cite{lu2021optimal} suggests that it could be possible to further improve these network-dependent rates. However, such improvement have only been established by utilizing multiple gossip rounds as discussed in the introduction. Therefore, one potential future direction is to investigate whether these rates can be further improved {\em without} utilizing  multiple gossip rounds.

{\small
\bibliography{myref} 
}

\newpage
\appendices 
\small
\section{Relation of \eqref{UDA_alg} to Existing Methods} \label{app:relation_to_other_methods}
In this section, we describe how different existing methods are related to  algorithm \eqref{UDA_alg}. First, note that we can equivalently describe the updates in \eqref{UDA_alg}  in terms of $\x^k$  by noting that $\x^1=\A (\C \x^0 - \alpha \grad \f(\x^0) )$ and
\begin{align*}
\mathbf{x}^{k+2} -\x^{k+1} &=  \A \C (\mathbf{x}^{k+1} -\x^{k})-\alpha \A \big( \grad \mathbf{f}(\mathbf{x}^{k+1})-\grad \mathbf{f}(\mathbf{x}^{k})\big)  -  \mathbf{B} (\mathbf{y}^{k+1} -\y^{k})  \\
& \overset{\eqref{dual_UDA}}{=}  \A \C (\mathbf{x}^{k+1} -\x^{k})-\alpha \A \big( \grad \mathbf{f}(\mathbf{x}^{k+1})-\grad \mathbf{f}(\mathbf{x}^{k})\big)  -   \B^2 \x^{k+1}.
\end{align*}
Hence, for $k \geq 0$:
\begin{align} \label{UDA_single_update}
\x^{k+2}=(\I-\B^2+\A\C)\x^{k+1} -\A \C \x^{k}-\alpha \A  \big( \grad \mathbf{f}(\mathbf{x}^{k+1})-\grad \mathbf{f}(\mathbf{x}^{k})\big) 
\end{align}
with $\x^1=\A (\C \x^0 - \alpha \grad \f(\x^0) )$.

\noindent \textbf{Specific instances.} We next show that by choosing specific $\A,\B,\C$ as a function of the combination matrix $\W$, we can recover several important state-of-the-art methods.   In the following, we do not assume that $\W$ is symmetric unless otherwise stated.
\begin{itemize}

\item \textsc{Exact-Diffusion/D$^2$.}  For Exact-Diffusion (ED) \cite{yuan2019exactdiffI} (a.k.a. D$^2$ \cite{tang2018d}), we assume $\W$ to be symmetric and positive-semidefinite. If  $\A=\W$, $\B=(\I-\W)^{1/2}$, and $\C=\I$, we get ED/D$^2$ \cite{yuan2019exactdiffI,tang2018d}:
\begin{align} \label{exact_diff_app}
\x^{k+2}=\W \Big(2\x^{k+1} - \x^{k}-\alpha  \big( \grad \mathbf{f}(\mathbf{x}^{k+1})-\grad \mathbf{f}(\mathbf{x}^{k})\big) \Big), 
\end{align}
with $\x^1=\W ( \x^0 - \alpha \grad \f(\x^0) )$. The above methods is also known as NIDS \cite{li2017nids}.

\item \textsc{EXTRA.} For EXTRA \cite{shi2015extra}, we also assume $\W$ to be symmetric and positive-semidefinite. If we choose  $\A=\I$, $\B=(\I-\W)^{1/2}$, and $\C=\W$, we get EXTRA \cite{shi2015extra}:
\begin{align} \label{EXTRA_app}
\x^{k+2}=\W \big(2\x^{k+1} - \x^{k} \big)-\alpha  \big( \grad \mathbf{f}(\mathbf{x}^{k+1})-\grad \mathbf{f}(\mathbf{x}^{k})\big),
\end{align}
with $\x^1=\W  \x^0 - \alpha \grad \f(\x^0) $.

\item \textsc{ATC-GT.} Consider  the adapt-then-combine GT method (ATC-GT):
\begin{subequations} \label{GT_atc_alg_app}
\begin{align} 
\x^{k+1}&=\W (\x^{k} - \alpha \g^{k}) \\
\g^{k+1} &= \bar{\W} \big(\g^{k} + \grad \mathbf{f}(\mathbf{x}^{k+1})-\grad \mathbf{f}(\mathbf{x}^{k}) \big), 
\end{align}
\end{subequations}
with $\g^0=\bar{\W} \grad \f(\x^0)$, $\x^0 = \bar{\W} \x^0$ (\ie, $\x^0$ is consensual), and $\W,\bar{\W}$ are  doubly-stochastic combination matrices such that $\bar{\W} \W=\W \bar{\W}$.  Subtracting $ \bar{\W} \x^{k}$ from both sides of the first equation, we have for $k \geq 0$
\begin{align*}
\x^{k+2}- \bar{\W}\x^{k+1}&=\W \x^{k+1}-\bar{\W} \W  \x^{k} - \alpha  (\W \g^{k+1}-\bar{\W} \W \g^{k}) \\
&=\W \x^{k+1}-\bar{\W} \W  \x^{k} - \alpha  \W( \g^{k+1}-\bar{\W}  \g^{k}).
\end{align*}
Using the second equation and rearranging, we get
\begin{align*} 
\x^{k+2}= (\W+  \bar{\W}) \x^{k+1}  - \bar{\W}  \W \x^{k}-\alpha  \bar{\W} \W \big( \grad \mathbf{f}(\mathbf{x}^{k+1})-\grad \mathbf{f}(\mathbf{x}^{k}) \big),
\end{align*}
with $\x^1=\bar{\W} \W   ( \x^0 - \alpha \grad \f(\x^0) )$.  The above is the same as \eqref{UDA_single_update}  when $
\A=\bar{\W} \W$, $\B^2=(\I-\bar{\W})(\I-\W)$, and $\C= \I$.
 Note that when $\bar{\W}=\W$, then \eqref{GT_atc_alg_app} becomes the GT method from \cite{xu2015augmented}.

\item \textsc{Non-ATC-GT.} Consider the gradient-tracking algorithm 
\begin{subequations} \label{GT_nonatc_alg_app}
\begin{align}
\x^{k+1}&=\W \x^{k} - \alpha \g^{k} \\
\g^{k+1} &= \bar{\W} \g^{k} + \grad \mathbf{f}(\mathbf{x}^{k+1})-\grad \mathbf{f}(\mathbf{x}^{k}),
\end{align}
\end{subequations}
where $\W$ and $\bar{\W}$ are  doubly-stochastic combination matrices. Eliminating $\g^k$, we get:
\begin{align*}
\x^{k+2}= (\W+  \bar{\W})  \x^{k+1} - \bar{\W} \W \x^{k} -\alpha  \big( \grad \mathbf{f}(\mathbf{x}^{k+1})-\grad \mathbf{f}(\mathbf{x}^{k})\big). 
\end{align*}
The above is exactly \eqref{UDA_single_update}  when $\A=\I$, $\B^2=(\I-\bar{\W})(\I-\W)$, and $\C= \bar{\W} \W$. If $\bar{\W}=\W$ then \eqref{GT_nonatc_alg} is the GT method studied in \cite{nedic2017achieving,qu2017harnessing}.

\item \textsc{Semi-ATC-GT.} Consider the following variation of  GT methods:
\begin{subequations} \label{GT_semiatc_alg1_app}
\begin{align} 
\x^{k+1}&=\W (\x^{k} - \alpha \g^{k}) \\
\g^{k+1} &= \bar{\W} \g^{k} + \grad \mathbf{f}(\mathbf{x}^{k+1})-\grad \mathbf{f}(\mathbf{x}^{k}),  
\end{align}
\end{subequations}
where the ATC structure is used for the update of $\x^k$ only. Following the same argument as before, we can show that \eqref{GT_semiatc_alg1_app} is equivalent to
\begin{align*} 
\x^{k+2}=  (\W+  \bar{\W}) \x^{k+1}  - \bar{\W}  \W \x^{k}-\alpha \W \big( \grad \mathbf{f}(\mathbf{x}^{k+1})-\grad \mathbf{f}(\mathbf{x}^{k}) \big).
\end{align*}
 The above is the same as \eqref{UDA_single_update}  when $\A= \W$, $\B^2=(\I-\bar{\W})(\I-\W)$, and  $\C= \bar{\W}$.

We can also consider the GT variant:
\begin{subequations} \label{GT_semiatc_alg2_app}
\begin{align} 
\x^{k+1}&=\W \x^{k} - \alpha \g^{k} \\
\g^{k+1} &= \bar{\W} \big(\g^{k} + \grad \mathbf{f}(\mathbf{x}^{k+1})-\grad \mathbf{f}(\mathbf{x}^{k}) \big) .
\end{align}
\end{subequations}
where the  adapt-then-combine structure is used for the gradient-tracking variable $\g^k$ only. Eliminating the gradient-tracking variable gives:
\begin{align*} 
\x^{k+2}=  (\W+  \bar{\W})  \x^{k+1}  - \bar{\W}  \W \x^{k}-\alpha \bar{\W} \big( \grad \mathbf{f}(\mathbf{x}^{k+1})-\grad \mathbf{f}(\mathbf{x}^{k}) \big),
\end{align*}
 which is exactly \eqref{UDA_single_update}  when $
\A= \bar{\W}$, $\B^2=(\I-\bar{\W})(\I-\W)$, and $\C= \W$.
Note that when $ \bar{\W}=\W$, then \eqref{GT_semiatc_alg1_app} and \eqref{GT_semiatc_alg2_app} become the GT variants from \cite{di2016next}.

\end{itemize}

	\noindent \textbf{Unified decentralized algorithm from \cite{alghunaim2019decentralized}.}  Consider the unified decentralized algorithm (UDA) proposed in \cite{alghunaim2019decentralized}:
	\begin{subequations} \label{alg_UDA0}
		\begin{align}
			\mathbf{z}^{k+1} &=   \C \mathbf{x}^{k}-\alpha \grad \mathbf{f}(\mathbf{x}^{k})   -  \mathbf{B} \mathbf{y}^{k} \label{z_uda0}   \\
			\mathbf{y}^{k+1} &= \mathbf{y}^{k}+ \mathbf{B}  \mathbf{z}^{k+1} \label{dual_uda0}  \\
			\mathbf{x}^{k+1} &=  \A \z^{k+1}. \label{x_uda0} 
		\end{align}
	\end{subequations}
	In the above description, the matrix $\C$ is equivalent to $\I-\C$ in \cite{alghunaim2019decentralized}.  We now  show that UDA \eqref{alg_UDA0}  is equivalent to \eqref{UDA_single_update} (hence, \eqref{UDA_alg}) if $\A$ and $\B^2$ commute (\ie, $\A \B^2=\B^2 \A$).  We can eliminate the variable $\y^{k}$ in \eqref{alg_UDA0} by noting that
	\begin{align*}
		\mathbf{z}^{k+2} -\z^{k+1} &=    \C(\mathbf{x}^{k+1} -\x^{k})-\alpha  \big( \grad \mathbf{f}(\mathbf{x}^{k+1})-\grad \mathbf{f}(\mathbf{x}^{k}) \big)  -  \mathbf{B} (\mathbf{y}^{k+1} -\y^{k})  \\
		&=    \C(\mathbf{x}^{k+1} -\x^{k})-\alpha  \big( \grad \mathbf{f}(\mathbf{x}^{k+1})-\grad \mathbf{f}(\mathbf{x}^{k})\big)  -   \B^2 \z^{k+1}.
	\end{align*}
	Rearranging, we have
	\begin{align*}
		\mathbf{z}^{k+2} &=      (\I-\B^2) \z^{k+1}+ \C(\mathbf{x}^{k+1} -\x^{k})-\alpha  \big( \grad \mathbf{f}(\mathbf{x}^{k+1})-\grad \mathbf{f}(\mathbf{x}^{k}) \big).
	\end{align*}
	We now multiply both sides by $\A$ and use $\A \B^2=\B^2 \A$, to get
	\begin{align*}
		\x^{k+2}=(\I-\B^2+\A \C)\x^{k+1} -\A \C  \x^{k}-\alpha \A   \big( \grad \mathbf{f}(\mathbf{x}^{k+1})-\grad \mathbf{f}(\mathbf{x}^{k})\big),
	\end{align*}
	which is exactly \eqref{UDA_single_update}.

\paragraph{Unified framework from \cite{xu2021distributed}.} 	The method from \cite{xu2021distributed} can described by (see \cite[Eq. (7)]{xu2021distributed}):
	\begin{align*}
	\x^{k+2}=(\I-\bar{\C}+ \bar{\A})\x^{k+1} - \bar{\A}  \x^{k}-\alpha \bar{\B}  \big( \grad \mathbf{f}(\mathbf{x}^{k+1})-\grad \mathbf{f}(\mathbf{x}^{k})\big),
\end{align*}
where the matrices $\bar{\A}$, $\bar{\B}$, and $\bar{\C}$ are required to satisfy certain conditions;   in \cite{xu2021distributed}  it is assumed that $\bar{\A}=\bar{\B} \bar{\D}$ (see \cite[Assumption 8]{xu2021distributed}).  Clearly, the above is equivalent to \eqref{UDA_single_update} when  $\A=\bar{\B}$,  $\B^2=\bar{\C}$, and $\ \C=\bar{\D}$.

\paragraph{Unified framework from \cite{sundararajan2018canonical}.}  The algorithm from \cite{sundararajan2018canonical} has the form (in our notation):
\begin{align*}
	\mathbf{x}^k&=\big((1-\zeta_3) \mathbf{I} +\zeta_3 \mathbf{W} \big)\mathbf{z}^k  \\ 
	\mathbf{z}^{k+1}&=\big((1-\zeta_1) \mathbf{I} +\zeta_1 \mathbf{W} \big)\mathbf{z}^k   -\alpha \nabla \mathbf{f} (\mathbf{x}^k)  + \big((\zeta_0+\zeta_2) \mathbf{I}- \zeta_2 \mathbf{W} \big) \mathbf{y}^k  \\
	\mathbf{y}^{k+1} &=\mathbf{y}^{k} -(\mathbf{I}-\mathbf{W}) \mathbf{z}^k .
\end{align*}
For notional simplicity, let 
\begin{align*}
	\mathbf{W}_{\zeta_3}&\define\big((1-\zeta_3) \mathbf{I} +\zeta_3 \mathbf{W} \big), \qquad  \mathbf{W}_{\zeta_1}\define\big((1-\zeta_1) \mathbf{I} +\zeta_1 \mathbf{W} \big) \\
	 \mathbf{B}_1&\define\big((\zeta_0+\zeta_2) \mathbf{I}- \zeta_2 \mathbf{W} \big), \qquad \mathbf{B}_2\define(\mathbf{I}-\mathbf{W}).
\end{align*}
Then, the prevision recursion is
\begin{align*}
	\mathbf{x}^k&=\mathbf{W}_{\zeta_3} \mathbf{z}^k  \\ 
	\mathbf{z}^{k+1}&=\mathbf{W}_{\zeta_1} \mathbf{z}^k   -\alpha \nabla \mathbf{f} (\mathbf{x}^k)  +\mathbf{B}_1  \mathbf{y}^k  \\
	\mathbf{y}^{k+1} &=\mathbf{y}^{k} -\mathbf{B}_2  \mathbf{z}^k .
\end{align*}
Eliminating the vector $\mathbf{y}^k$, we obtain the equivalent form:
\begin{align*}
	\mathbf{x}^k&=\mathbf{W}_{\zeta_3} \mathbf{z}^k  \\ 
	\mathbf{z}^{k+2} &= (\mathbf{I}+\mathbf{W}_{\zeta_1})  \mathbf{z}^{k+1} -(\mathbf{W}_{\zeta_1} +\mathbf{B}_1 \mathbf{B}_2 )\mathbf{z}^{k}  -\alpha \big(\nabla \mathbf{f} (\mathbf{x}^{k+1})-\nabla \mathbf{f} (\mathbf{x}^{k})\big).  
\end{align*}
Note that each pair from $\mathbf{W}_{\zeta_3},\mathbf{W}_{\zeta_1},\mathbf{B}_1,\mathbf{B}_2$ commute with each other since they are polynomial functions of $\mathbf{W}$. Thus, multiplying the second line in the previous algorithm by $\mathbf{W}_{\zeta_3}$, we obtain
\begin{align*}
	\mathbf{x}^{k+2} &= (\mathbf{I}+\mathbf{W}_{\zeta_1})  \mathbf{x}^{k+1} -(\mathbf{W}_{\zeta_1} +\mathbf{B}_1 \mathbf{B}_2 )\mathbf{x}^{k}  -\alpha \mathbf{\mathbf{W}_{\zeta_3}}  \big(\nabla \mathbf{f} (\mathbf{x}^{k+1})-\nabla \mathbf{f} (\mathbf{x}^{k})\big) .
\end{align*}
We see that this is equivalent to \eqref{UDA_single_update} when $
	\mathbf{A}=\mathbf{W}_{\zeta_3}$, $\mathbf{A C}=\mathbf{W}_{\zeta_1}- \mathbf{B}_1 \mathbf{B}_2$, and $\mathbf{I}-\mathbf{B}^2+\mathbf{A} \mathbf{C}= \mathbf{I}+\mathbf{W}_{\zeta_3}$.

\section{Fundamental Factorization}
\subsection{Proof of Lemma \ref{lemma:diagonalization}}
\label{app:lemma_diag_proof}
Recall that $\hat{\mathbf{\Lambda}}_a=\diag\{\lambda_{a,i}\}_{i=2}^n \otimes I_d$, $\hat{\mathbf{\Lambda}}_b=\diag\{\lambda_{b,i} \}_{i=2}^n \otimes I_d$, $\hat{\mathbf{\Lambda}}_c=\diag\{\lambda_{c,i}\}_{i=2}^n \otimes I_d$. Hence, the  matrix $\G$ defined in \eqref{G_matrix} can be rewritten as
\begin{align*} 
\G = \begin{bmatrix}
\diag\{\lambda_{a,i} \lambda_{c,i}-\lambda_{b,i}^2\}_{i=2}^n & -\diag\{\lambda_{b,i} \}_{i=2}^n \\
\diag\{\lambda_{b,i} \}_{i=2}^n & ~~I_{n-1}
\end{bmatrix} \otimes I_d.
\end{align*}
Utilizing the structure of $\G$, we can exchange the columns and rows of $\G$ through some permutation matrix $\P \in \real^{2d(n-1) \times 2d(n-1)}$ such that
\begin{align*}
\P \G \P\tran= \bdiag\{G_i\}_{i=2}^n \otimes I_d,
\end{align*}
where
\begin{align} \label{G_i_proof}
G_i \define \begin{bmatrix}
\lambda_{a,i} \lambda_{c,i}-\lambda_{b,i}^2 & -\lambda_{b,i} \\
\lambda_{b,i} & ~~1
\end{bmatrix} \in \real^{2 \times 2}.
\end{align}
Let us denote the eigenvalues of $G_i$ by $\gamma_{1,i}$ and $ \gamma_{2,i}$.  If the eigenvalues of $G_i$ are distinct ($\gamma_{1,i} \neq \gamma_{2,i}$), then there exists a $2 \times 2$ invertible $V_i$ such that \cite{horn2012matrix}:
\begin{align*}
G_i= V_i \Gamma_i V_i^{-1}, \quad \Gamma_i \define \diag\{\gamma_{1,i},\gamma_{2,i}\}.
\end{align*} 
 It follows that $\G$ is similar to a diagonal matrix $ \G  = \hat{\V} \mathbf{\Gamma} \hat{\V}^{-1}$,
where
\begin{align} \label{V_hat_eq_def}
\hat{\V} \define \P\tran \V, \quad \V\define \bdiag\{V_i\}_{i=2}^n \otimes I_d, \quad \mathbf{\Gamma} \define \bdiag\{\Gamma_i\} \otimes I_d,
\end{align}
with $\|\mathbf{\Gamma}\|=\max_{i \in \{2,\dots,n\}} \|\Gamma_i\|=\max_{i \in \{2,\dots,n\}} \{|\gamma_{1,i}|,|\gamma_{2,i}|\}< 1$.
Now, suppose that $G_i$ has repeated eigenvalues $\gamma_{1,i}=\gamma_{2,i}=\gamma_i$. Then, using  Jordan canonical form \cite{horn2012matrix}, there exists an invertible matrix $T_i$ such that
\begin{align*}
G_i= T_i J_i T_i^{-1}, \quad J_i \define \begin{bmatrix}
\gamma_{i} & 1 \\
0 & \gamma_{i}
\end{bmatrix}.
\end{align*} 
 If we let $E_i \define \diag\{1,\epsilon_i\}$ where $\epsilon_i>0$ is an arbitrary constant, then we can rewrite $G_i$ as
\begin{align*} 
G_i= V_i \Gamma_i V_i^{-1}, \quad V_i \define T_i E_i, \quad \Gamma_i \define E_i^{-1} J_i E_i = \begin{bmatrix}
\gamma_{i} & \epsilon \\
0 & \gamma_{i}
\end{bmatrix}.
\end{align*} 
 It follows that  $\G  = \hat{\V} \mathbf{\Gamma} \hat{\V}^{-1}$ where $\hat{\V}= \P\tran \V$, $\V =\bdiag\{T_iE_i\}_{i=2}^n \otimes I_d$, and $\mathbf{\Gamma} = \bdiag\{\Gamma_i\} \otimes I_d$. Since the spectral radius of any matrix is upper bounded by any norm \cite{horn2012matrix}, it holds that
 \begin{align} \label{epsilon_Gamma_bound}
 \|\Gamma_i\|^2 = \rho(\Gamma_i \Gamma_i^*) \leq \|\Gamma_i \Gamma_i^*\|_1 =|\gamma_i|^2+\epsilon_i |\gamma_i|+\epsilon_i^2 \leq \big(|\gamma_i|+\epsilon_i \big)^2.
 \end{align}
  Here, $\|.\|_1$ denote the maximum-absolute-column-sum matrix norm. Therefore, $\|\mathbf{\Gamma}\|< 1$ for any $\epsilon_i < 1-|\gamma_i|$.

\subsection{Special Cases of Lemma \ref{lemma:diagonalization}} \label{app:bounds_special_cases}
In this section, we specify the results of Section \ref{app:lemma_diag_proof} to the instances discussed in Appendix \ref{app:relation_to_other_methods}. First, we note that for any matrix
\begin{align*}
G = \begin{bmatrix}
a & b \\
c & d
\end{bmatrix} \in \real^{2 \times 2},
\end{align*}
the eigenvalues are: 
\begin{align} \label{eigenvalues_formula}
\gamma_{1,2} =\frac{(a+d) \pm \sqrt{(a+d)^2-4 (ad -bc)}}{2}.
\end{align}
Moreover, if $\gamma_1 \neq \gamma_2 $, then
\begin{align} \label{eigenvectors_diag_formula}
G=V \diag\{\gamma_1,\gamma_2\} V^{-1}, \quad V=\begin{bmatrix}
v_1 &  v_2
\end{bmatrix} ,
\end{align}
 where $v_{1,2}=\frac{1}{r} \col\{b,\gamma_{1,2}-a\}$ (or $v_{1,2}=\frac{1}{r} \col\{\gamma_{1,2}-d,  
c \}$) if  $b \neq 0$ (or $c \neq 0$) for any $r \neq 0$. Here, $v_{1}$ and $v_2$ are eigenvectors $G$ corresponding to eigenvalues $\gamma_{1}$ and $\gamma_2$, respectively.
\subsubsection{Exact-Diffusion/D$^2$ and EXTRA}
\textbf{ED/D$^2$.} We will consider the case for ED/D$^2$ \eqref{exact_diff_app} first. Recall from Appendix \ref{app:relation_to_other_methods} that for ED/D$^2$, we have $\A=\W$, $\B=(\I-\W)^{1/2}$, and $\C=\I$. For this case, we have $\lambda_{a,i}=\lambda_i$,  $\lambda_{c,i}=1$, and $\lambda_{b,i}=\sqrt{1-\lambda_i}$ where $\{\lambda_i\}$ denote the eigenvalues of $W$. Thus, the matrix \eqref{G_i_proof} becomes:
\begin{align*} 
G_i = \begin{bmatrix}
2\lambda_{i} -1 & -\sqrt{1-\lambda_i} \\
\sqrt{1-\lambda_{i}} & ~~1
\end{bmatrix} \in \real^{2 \times 2}.
\end{align*}
Using \eqref{eigenvalues_formula}, the eigenvalues of $G_i$ ($i=2,\ldots,n$) are:
\begin{align*} 
\gamma_{(1,2),i} = \lambda_i \pm \sqrt{ \lambda_i^2 - \lambda_i}.
\end{align*}
Note that $|\gamma_{(1,2),i}|<1$ if $-\frac{1}{3}<\lambda_i<1$. This implies that for the convergence of ED/D$^2$, we require $W > -\frac{1}3 I$. It is sufficient for our purposes to assume that $W \geq 0$. Thus, we have two cases for the decomposition of $G_i$ derived next.
\begin{itemize}
\item If $0<\lambda_i <1$, then the eigenvalues are complex and distinct:
\begin{align*} 
\gamma_{(1,2),i} = \lambda_i \pm j\sqrt{\lambda_i-\lambda_i^2}, \qquad |\gamma_{(1,2),i} |=\sqrt{\lambda_i}<1,
\end{align*}
where $j^2=-1$. Using \eqref{eigenvectors_diag_formula} with $r=\sqrt{1-\lambda_i}$, we can factor $G_i$ as $G_i=V_i \diag\{\gamma_{1,i},\gamma_{2,i}\} V_i^{-1}$
where
\begin{align*}
V_i=\begin{bmatrix}
-1 & -1 \\
\sqrt{1-\lambda_i}+j\sqrt{  \lambda_i} & \sqrt{1-\lambda_i}-j\sqrt{  \lambda_i} 
\end{bmatrix}.
\end{align*}
The inverse of $V_i$ is
\begin{align*}
V_i^{-1}=\frac{1}{2j \sqrt{\lambda_i}} \begin{bmatrix}
\sqrt{1-\lambda_i}-j\sqrt{  \lambda_i} & ~~1 \\
-\sqrt{1-\lambda_i}-j\sqrt{  \lambda_i} & -1 
\end{bmatrix} =\frac{1}{2 \sqrt{\lambda_i}} \begin{bmatrix}
-\sqrt{\lambda_i}-j\sqrt{1-\lambda_i} & -j \\
-\sqrt{\lambda_i}+j\sqrt{1-\lambda_i} & ~~j 
\end{bmatrix}.
\end{align*}
Note that
\begin{align*}
V_iV_i^*=\begin{bmatrix}
2 & -2 \sqrt{1-\lambda_i} \\ 
-2 \sqrt{1-\lambda_i} & 2
\end{bmatrix}.
\end{align*}
Since the spectral radius of matrix is upper bounded by any of its norm, it holds that   $\|V_i\|^2 =\rho(V_iV_i^*) \leq \|V_iV_i^*\|_1 \leq 4$. Following a similar argument for $V_i^{-1}$, we have
\begin{align*}
(V_i^{-1})(V_i^{-1})^*&=\frac{1}{4 \lambda_i} \begin{bmatrix}
-\sqrt{  \lambda_i}-j\sqrt{1-\lambda_i} & -j \\
-\sqrt{  \lambda_i}+j\sqrt{1-\lambda_i} & ~~j 
\end{bmatrix} \begin{bmatrix}
-\sqrt{  \lambda_i}+j\sqrt{1-\lambda_i} & -\sqrt{  \lambda_i}-j\sqrt{1-\lambda_i} \\
j & -j 
\end{bmatrix} \\
&=\frac{1}{4 \lambda_i} \begin{bmatrix}
2 & 2 \lambda_i-2+2j\sqrt{\lambda_i-\lambda_i^2} \\
2 \lambda_i-2+2j\sqrt{\lambda_i-\lambda_i^2} & 2
\end{bmatrix}.
\end{align*}
Hence, $\|V_i^{-1}\|^2 =\rho(V_i^{-1}(V_i^{-1})^*) \leq \|V_i^{-1}(V_i^{-1})^*\|_1 \leq \frac{1}{\lambda_i}$.

\item  If $\lambda_i =0$, then the eigenvalues of $G_i$ are both zero $
\gamma_{1,2} = 0$. Using Jordan canonical form, it holds that
\begin{align*}
G_i =\begin{bmatrix}
-1 & -1 \\
~~1 & ~~1
\end{bmatrix}, \qquad T_i^{-1} G_i T_i =\begin{bmatrix}
0 & 1 \\
0 & 0
\end{bmatrix}, \quad T_i =\begin{bmatrix}
~~1 & -\tfrac{1}{2} \vspace{0.5mm} \\ 
-1 &  -\tfrac{1}{2}
\end{bmatrix}.
\end{align*}
 If we let $E=\diag\{1,\epsilon\}$ for $0<\epsilon<1$  then
\begin{align*}
G_i= V_i \Gamma_i V_i^{-1}, \quad V_i= T_i E, \quad  \Gamma_i=\begin{bmatrix}
0 & \epsilon \\ 
0 & 0
\end{bmatrix}.
\end{align*}
Simple calculations show that $\|\Gamma_i\|=\epsilon<1$ and
\begin{align*}
 \|V_i\|^2 \leq \|T_i\|^2 \|E\|^2 = 2, \qquad \|V_i^{-1}\|^2 \leq \|T_i^{-1}\|^2 \|E^{-1}\|^2 =\frac{2}{\epsilon^2}.
\end{align*}
We can choose $\epsilon^2=\lambda=\max_{i \in \{2,\dots,n\}} \lambda_i$. 
\end{itemize}
Putting things together, we find
\begin{subequations} \label{exact_diff_bounds}
\begin{align} 
\|\hat{\V}\|^2 &= \|\P\tran \V\|^2 \leq 4, \quad \|\hat{\V}^{-1}\|^2 = \|\V^{-1} \P\|^2 \leq \frac{2}{ \underline{\lambda}}, \quad \gamma=\|\mathbf{\Gamma}\| = \sqrt{\lambda}, \\
 \lambda_a&= \|\hat{\mathbf{\Lambda}}_a\|=\lambda, \quad \underline{\lambda_b}= \frac{1}{\|\hat{\mathbf{\Lambda}}_b^{-1}\|}=\sqrt{1-\lambda},
\end{align}
\end{subequations}
where $\lambda=\max_{i \in \{2,\dots,n\}} \lambda_i$ and  $\underline{\lambda}$ is the minimum non-zero eigenvalue of $W$.

\noindent \textbf{EXTRA.} Observe that for  EXTRA \eqref{EXTRA_app}, the matrix $G_i$ is identical to ED/D$^2$. Hence, the same bounds \eqref{exact_diff_bounds} hold for EXTRA except for $\lambda_a$, which is equal to one, \ie,  $\lambda_a=1$,  for EXTRA.  
 
\subsubsection{GT methods}
\textbf{ATC-GT.}  Consider the ATC-GT method \eqref{GT_atc_alg} (or \eqref{GT_atc_alg_app} with $\bar{\W}=\W$). In this case, we have $\lambda_{a,i}=\lambda_i^2$, $\lambda_{b,i}= 1-\lambda_i$, and $\lambda_{c,i}=1$. Thus, the matrix $G_i$ ($i=2,\dots,n$) given in \eqref{G_i_proof} is
\begin{align*}
G_i = \begin{bmatrix}
2\lambda_i-1 & -(1-\lambda_i) \\
 1-\lambda_i  & 1
\end{bmatrix} \in \real^{2 \times 2}.
\end{align*}
Using \eqref{eigenvalues_formula}, we find that the eigenvalues of $G_i$ are identical $
\gamma_{1,i}=\gamma_{2,i}  =\lambda_i$.
Clearly,  the eigenvalues are strictly less than one by assumption since $0 \leq \lambda_i <1$ ($i=2,\dots,n$). If we let
\begin{align*}
T_i=\begin{bmatrix}
-1 & 0 \\
~~1 & \frac{1}{1-\lambda_i} \end{bmatrix}, \qquad T_i^{-1}=\begin{bmatrix}
-1 & 0 \\
~~1-\lambda_i & 1-\lambda_i
\end{bmatrix}, 
\end{align*}
then, it holds that
\begin{align*}
 G_i= T_i \begin{bmatrix}
\lambda_i & 1 \\
0 & \lambda_i
\end{bmatrix} T_i^{-1}. 
\end{align*}
 If we further let $E_i=\diag\{1,\epsilon_i\}$ where $\epsilon_i>0$ is an arbitrary constant, and define
\begin{align*}
V_i \define T_i E_i=\begin{bmatrix}
-1 & 0 \\
~~1 & \frac{\epsilon}{1-\lambda_i}
\end{bmatrix}, \quad V_i^{-1}=E_i^{-1} T_i^{-1} =\begin{bmatrix}
-1 & 0 \\
\frac{1-\lambda_i}{\epsilon} & \frac{1-\lambda_i}{\epsilon} 
\end{bmatrix}.
\end{align*} 
Then, we have $G_i= V_i \Gamma_i V_i^{-1}$ where
\begin{align*}
 \Gamma_i \define E_i^{-1} \begin{bmatrix}
\lambda_i & 1 \\
0 & \lambda_i
\end{bmatrix} E_i = \begin{bmatrix}
\lambda_{i} & \epsilon \\
0 & \lambda_{i}
\end{bmatrix}.
\end{align*} 
Choosing $\epsilon=(1-|\lambda_i|)/2$, it can be verified that $\|\mathbf{\Gamma}\|<1$.

The above derivations are sufficient  for our convergence analysis  to cover GT methods with $I < W \leq I$.  However, we shall assume in the following that $W \geq 0$ in order to get refined   network dependent bounds for GT methods.    Under this additional condition, we have $\lambda_i \geq 0$ and if we set $\epsilon=(1-\lambda_i)/2$, it holds that
\begin{align*}
\Gamma_i  = \begin{bmatrix}
\lambda_{i} & \frac{1-\lambda_i}{2} \\
0 & \lambda_{i}
\end{bmatrix}, \quad V_i =\begin{bmatrix}
-1 & 0 \\
~~1 & \frac{1}{2}
\end{bmatrix}, \quad V_i^{-1}=\begin{bmatrix}
-1 & 0 \\
~~2 & 2 
\end{bmatrix}.
\end{align*} 
Hence, for ATC-GT with $W \geq 0$, we have 
\begin{subequations} \label{gt_diff_W_bound}
\begin{align} 
\|\hat{\V}\|^2 &= \|\P\tran \V\|^2 \leq 3, \quad
\|\hat{\V}^{-1}\|^2 = \|\V^{-1} \P\|^2 \leq 9, \quad  \gamma=\|\mathbf{\Gamma}\| \leq  \tfrac{1+\lambda}{2}, \\
\lambda_a&=\|\hat{\mathbf{\Lambda}}_a\|=\lambda^2, \quad \underline{\lambda_b}=\frac{1}{\|\hat{\mathbf{\Lambda}}_b^{-1}\|}  = 1-\lambda,
\end{align}
\end{subequations}
where $\lambda=\max_{i \in \{2,\dots,n\}} \lambda_i$.

\noindent \textbf{Other GT variants.} Note that for the other variations of GT methods considered in Appendix \ref{app:relation_to_other_methods} with $\bar{\W}=\W$, the matrix $G_i$ is identical to the previous case. Hence, the same bounds \eqref{gt_diff_W_bound} holds for these other variations except for the value of $\lambda_a$, which is $\lambda_a=1$ for the non-ATC-GT \eqref{GT_nonatc_alg_app} and $\lambda_a=\lambda$ for the semi-ATC-GT variants \eqref{GT_semiatc_alg1_app} and \eqref{GT_semiatc_alg2_app}.

\section{Convergence Proof} \label{app:convergence_analysis}
In this section, we prove  Theorems \ref{thm_nonconvex} and \ref{thm_PL_linear_conv}. We will first list some useful  inequalities and facts that are used in our analysis.
\subsubsection*{Useful inequalities}
\begin{itemize}
\item Under the $L$-smoothness  condition \eqref{smooth_f_eq} given in Assumption \ref{assump:smoothness}, the aggregate cost $f$ is also $L$-smooth. It follows that \cite{nesterov2013introductory}:
	 \begin{align}
	      f(y) &\leq f(z)+ \langle \grad f(z),~ y-z \rangle+\tfrac{L}{2} \|y-z\|^2, \quad \forall~z,y \in \real^n. \label{bound:L_smooth_function}
	 \end{align}

\item Recall from \eqref{noise_gradient} that $\w^k = \grad \F (\x^k,\bxi^k) -\grad \mathbf{f}(\x^{k})$ and $\bar{\w}^k =\tfrac{1}{n} \sum_{i=1}^n \grad F_i (x_i^k,\xi_i^k) -\grad f_i(x_i^{k})$. Hence, under Assumption \ref{assump:noise}, we have 
\begin{equation}\label{noise_bound_impli}
\begin{aligned}
\Ex [\w^k | \bm{\cF}^{k}] &=\zero, \quad  &\Ex[\|\w^k\|^2 | \bm{\cF}^{k}] &\leq n\sigma^2, \\
 \Ex [\bar{\w}^k | \bm{\cF}^{k}] &=0, \quad  &\Ex[\|\bar{\w}^k\|^2 | \bm{\cF}^{k}] &\leq \frac{\sigma^2}{n}.
\end{aligned}
\end{equation}

\item 	 Using \eqref{hat_relation_avg_Dev}, it holds that
\begin{align} \label{avg_e_hat_bound}
\|\x^k-\bar{\x}^k\|^2=\|\hat{\U}\tran\x^{k}\|^2=\|\upsilon \hat{\V} \hat{\e}^{k} \| -\|\hat{\mathbf{\Lambda}}_b^{-1} \hat{\U}\tran\s^{k}\|^2  \leq   \upsilon^2 \|\hat{\V}\|^2 \| \hat{\e}^{k} \|^2.
\end{align} 
Moreover, we have
\begin{align} \label{lipschitz_bound_average}
\| \overline{\grad\f}(\x^k)-\grad f(\bar{x}^{k}) \|^2 &=\Big\| \tfrac{1}{n} \sum_{i=1}^n \big(\grad f_i(x_i^k)-\grad f_i(\bar{x}^k)\big) \Big\|^2 \nonumber \\
&\leq\frac{1}{n} \sum_{i=1}^n \big\| \grad f_i(x_i^{k})-\grad f_i(\bar{x}^{k}) \big\|^2 \nonumber \\
& \overset{\eqref{smooth_f_eq}}{\leq}  \frac{L^2}{n} \|\x^{k} -  \bar{\x}^{k}\|^2 
	\overset{\eqref{avg_e_hat_bound}}{\leq}  \frac{L^2 \upsilon^2 \|\hat{\V}\|^2 }{n} \|\hat{\e}^k\|^2.
\end{align}

\item  Let $\upsilon=\sqrt{n} v_2$ where $v_2=\|\hat{\V}^{-1}\|$. Using the definition of $\hat{\e}^{0}$ in \eqref{e_hat_def}, we have 
		\begin{align*}
\| \hat{\e}^{0} \|^2 \leq  \tfrac{v_2^2}{\upsilon^2} \left(\|\hat{\U}\tran\x^{0}\|^2 + \|\hat{\mathbf{\Lambda}}_b^{-1}  \hat{\U}\tran\s^{0}\|^2 \right) = \tfrac{1}{n}\|\hat{\U}\tran\x^{0}\|^2 +\tfrac{1}{n} \|\hat{\mathbf{\Lambda}}_b^{-1}  \hat{\U}\tran\s^{0}\|^2.
	\end{align*}
If the initialization is identical $\x^0=\one \otimes x^0$ (for some $x^0 \in \real^d$), then $\|\hat{\U}\tran\x^{0}\|=0$ and $\x^0=\bar{\x}^0$; moreover, $\z^0=\zero$ since $\y^0=\zero$. Hence,  we have
\begin{align} \label{e_0_bound}
\| \hat{\e}^{0} \|^2 \leq 	\frac{1}{ \underline{\lambda_b}^2 n} \|\hat{\U}\tran\s^{0}\|^2 &\overset{\eqref{s_definition}}{=} \frac{\alpha^2}{\underline{\lambda_b}^2 n} \|\hat{\U}\tran \A \grad \mathbf{f}(\bar{\x}^{0})\|^2 \nonumber \\
&= \frac{\alpha^2}{ \underline{\lambda_b}^2 n} \|(\A-\tfrac{1}{n}\one\tran \one \otimes I_d )  \grad \mathbf{f}(\x^0)\|^2 =  \frac{\alpha^2\zeta_0^2}{ \underline{\lambda_b}^2}.
\end{align}	
	where  $\zeta_0 \define \frac{1}n\|(\A-\tfrac{1}{n}\one\tran \one \otimes I_d )  \big(\grad \mathbf{f}(\x^0) - \one \otimes \grad f(x^0) \big)\|$.
	
\item For any $0 \leq \eta <1$, it holds that
\begin{align} \label{sum_bound}
 \sum_{\ell=0}^{k-1} \eta^{k-1-\ell} \leq \sum_{\ell=0}^{\infty} \eta^{\ell} =\frac{1}{1-\eta}.
\end{align}
 Moreover, for a non-negative sequence $\{w_\ell\}$ it holds that:
\begin{align} \label{sum_bound_2}
  \sum_{k=1}^K \sum_{\ell=0}^{k-1} \eta^{k-1-\ell} w_\ell  \leq  \left( \sum_{\ell=0}^\infty \eta^{\ell}  \right)\sum_{k=0}^{K-1} w_k    = \frac{1}{1-\eta}  \sum_{k=0}^{K-1} w_k.
\end{align}

\end{itemize}

\subsection{Descent and Consensus Inequalities}
In this section, we establish two inequalities given in Lemmas \ref{lemma_descent_inq} and \ref{lemma:cons_inequ}  that are essential to prove Theorems \ref{thm_nonconvex} and \ref{thm_PL_linear_conv}. We start with the following lemma regarding the iterates average.
\begin{lemma} [\bfseries \small Descent inequality] \label{lemma_descent_inq}
	Suppose that Assumptions \ref{assump:smoothness} and \ref{assump:noise} holds. If $\alpha \leq \tfrac{1}{2 L}$, then for any $k \geq 0$, we have
	\begin{equation} \label{descent_inquality}	
	\begin{aligned}
		  \Ex   f(\bar{x}^{k+1}) &\leq  \Ex  f(\bar{x}^{k}) - \frac{\alpha}{2}  \Ex  \| \grad f(\bar{x}^{k})  \|^2 - \frac{\alpha}{4}  \Ex  \| \overline{\grad\f}(\x^k)\|^2 
		   +\frac{\alpha L^2 \upsilon^2 v_1^2  }{2n}  \Ex  \|\hat{\e}^k\|^2+	\frac{ \alpha^2 L \sigma^2 }{2 n},
	\end{aligned}	
	\end{equation}
where $v_1 = \|\hat{\V}\|$ and $\upsilon$ is an arbitrary strictly positive constant.
\end{lemma}
\begin{proof}
Recall that $\bar{x}^{k+1}=\bar{x}^{k}-\alpha \overline{\grad\f}(\x^k) - \alpha 	\bar{\w}^k$ where $\overline{\grad\f}(\x^k) =\tfrac{1}{n} \sum_{i=1}^n \grad f_i(x_i^{k})$. Setting $y=\bar{x}^{k+1}$ and $z=\bar{x}^{k}$ in  inequality \eqref{bound:L_smooth_function}, taking conditional expectation, and using \eqref{noise_bound_impli}, it holds that
	\begin{align}
	  \Ex [  f(\bar{x}^{k+1}) ~|~ \bm{\cF}^{k}]
	&\leq f(\bar{x}^{k}) - \alpha \big\langle \grad f(\bar{x}^{k}), \overline{\grad\f}(\x^k) \big\rangle  +\tfrac{\alpha^2 L}{2} \Ex \big[\|\overline{\grad\f}(\x^k)  +  	\bar{\w}^k(\x^{k})\|^2 ~|~ \bm{\cF}^{k}\big] \nonumber \\
	  &\leq f(\bar{x}^{k}) - \alpha   \big\langle \grad f(\bar{x}^{k}), \overline{\grad\f}(\x^k)   \big\rangle +\tfrac{\alpha^2 L }{2}  \|\overline{\grad\f}(\x^k) \|^2 +  	\tfrac{\alpha^2  L \sigma^2}{2n}.
	\end{align}
	Since $2 \langle a, b \rangle = \|a\|^2 + \|b\|^2 - \|a - b\|^2$, we have
	\begin{align}
     -      \big \langle \grad f(\bar{x}^{k}), \overline{\grad\f}(\x^k)   \big\rangle
   &=  - \tfrac{1}{2}  \|\grad f(\bar{x}^{k})\|^2 - \tfrac{1}{2}  \|\overline{\grad\f}(\x^k) \|^2 + \tfrac{1}{2}  \|\overline{\grad\f}(\x^k)-\grad f(\bar{x}^{k}) \|^2.
\end{align}
	Combining the last two equations, we get
	\begin{align}
	  \Ex [  f(\bar{x}^{k+1}) | \bm{\cF}^{k}] &\leq f(\bar{x}^{k}) - \tfrac{\alpha}{2} \| \grad f(\bar{x}^{k})  \|^2 - \tfrac{\alpha}{2} (1-\alpha L) \| \overline{\grad\f}(\x^k)\|^2 \nonumber \\
	  & \quad +\tfrac{\alpha}{2}    \| \overline{\grad\f}(\x^k)-\grad f(\bar{x}^{k}) \|^2+	\tfrac{ \alpha^2 L \sigma^2 }{2 n}. \label{2bndsbd}
	\end{align}
 Substituting the bound \eqref{lipschitz_bound_average} into inequality \eqref{2bndsbd}, letting $\alpha \le \frac{1}{2L}$, and taking the expectation yields \eqref{descent_inquality}.
\end{proof}

We next establish an inequality regarding the consensus iterates $\{\hat{\e}^{k}\}$.
\begin{lemma}[\bfseries \small Consensus inequality] \label{lemma:cons_inequ}
Supose that Assumptions \ref{assump:network}--\ref{assump:noise} hold. Then, for all $k \geq 0$, we have
\begin{equation} \label{cons_ineq_lemma}
\begin{aligned}
			\Ex \|\hat{\e}^{k+1}\|^2 &\leq  
	  \left(\gamma + \frac{2\alpha^2 L^2 v_1^2 v_2^2 \lambda_a^2 }{1-\gamma}  \right) \Ex \|\hat{\e}^{k}\|^2  + \frac{2\alpha^4 L^2 v_2^2 \lambda_a^2 n }{\underline{\lambda_b}^2(1-\gamma)\upsilon^2}  \Ex \|\overline{\grad \f}(\x^k) \|^2  \\
	& \quad +\frac{2\alpha^2  v_2^2 \lambda_a^2 n   \sigma^2}{\upsilon^2}+\frac{2\alpha^4 L^2 v_2^2 \lambda_a^2  \sigma^2}{\underline{\lambda_b}^2(1-\gamma)\upsilon^2},
\end{aligned}
\end{equation}
where  $v_1 \define \|\hat{\V}\|$, $v_2 \define \|\hat{\V}^{-1}\|$, $\lambda_a \define \|\mathbf{\Lambda}_a\|$, $\gamma \define \|\mathbf{\Gamma}\|<1$, and $\underline{\lambda_b}\define\frac{1}{\|\mathbf{\Lambda}_b^{-1}\|}$.
\end{lemma}
	\begin{proof}
	From \eqref{error_check_diag}, we have
	\begin{align*}
	\|\hat{\e}^{k+1}\|^2&= \left\|\mathbf{\Gamma} \hat{\e}^{k} - \tfrac{\alpha}{\upsilon} \hat{\V}^{-1}  \begin{bmatrix}
  \hat{\mathbf{\Lambda}}_a \hat{\U}\tran \big(\grad \mathbf{f}(\mathbf{x}^{k}) - \grad \mathbf{f}(\bar{\x}^{k})\big) \\
\mathbf{\Lambda}_b^{-1} \hat{\mathbf{\Lambda}}_a \hat{\U}\tran \big(\grad \mathbf{f}(\bar{\x}^{k}) -\grad \mathbf{f}(\bar{\x}^{k+1})\big)
\end{bmatrix} - \tfrac{\alpha}{\upsilon} \hat{\V}^{-1} \begin{bmatrix}
   \hat{\mathbf{\Lambda}}_a \hat{\U}\tran \w^k  \\
\zero
\end{bmatrix} \right\|^2 \\
  &= \left\|\mathbf{\Gamma} \hat{\e}^{k} - \tfrac{\alpha}{\upsilon} \hat{\V}^{-1}  \begin{bmatrix}
  \hat{\mathbf{\Lambda}}_a \hat{\U}\tran \big(\grad \mathbf{f}(\mathbf{x}^{k}) - \grad \mathbf{f}(\bar{\x}^{k})\big) \\
 \hat{\mathbf{\Lambda}}_b^{-1} \hat{\mathbf{\Lambda}}_a \hat{\U}\tran \big(\grad \mathbf{f}(\bar{\x}^{k}) -\grad \mathbf{f}(\bar{\x}^{k+1})\big)
\end{bmatrix} \right\|^2 
\\
& \quad + \tfrac{\alpha^2}{\upsilon^2}  \| \hat{\V}_l^{-1} \hat{\mathbf{\Lambda}}_a \hat{\U}\tran \w^k   \|^2 
  -    \tfrac{2\alpha}{\upsilon}  \big\langle \hat{\V}_l^{-1}\hat{\mathbf{\Lambda}}_a \hat{\U}\tran \w^k , ~   \mathbf{\Gamma} \hat{\e}^{k} \big\rangle  
 \\
 & \quad  + \tfrac{2\alpha^2}{\upsilon^2}  \left \langle \hat{\V}_l^{-1}
   \hat{\mathbf{\Lambda}}_a \hat{\U}\tran \w^k ,~    \hat{\V}^{-1}  \begin{bmatrix}
\hat{\mathbf{\Lambda}}_a \hat{\U}\tran \big(\grad \mathbf{f}(\mathbf{x}^{k}) - \grad \mathbf{f}(\bar{\x}^{k})\big) \\
\hat{\mathbf{\Lambda}}_b^{-1} \hat{\mathbf{\Lambda}}_a \hat{\U}\tran \big(\grad \mathbf{f}(\bar{\x}^{k}) -\grad \mathbf{f}(\bar{\x}^{k+1})\big)
\end{bmatrix} \right \rangle,
	\end{align*}		
	where $\hat{\V}_l^{-1}$ is the left part of $\hat{\V}^{-1}=[\hat{\V}_l^{-1} ~ \hat{\V}_r^{-1}]$. 	Using Cauchy–Schwarz inequality and $ab \leq \tfrac{a^2}{2}+\tfrac{b^2}{2}$ for any non-negative scalars $a$ and $b$, the last term can be upper bounded by:
\begin{align*}
& \tfrac{2\alpha^2}{\upsilon^2}  \left \langle \hat{\V}_l^{-1}
   \hat{\mathbf{\Lambda}}_a \hat{\U}\tran \w^k ,~    \hat{\V}^{-1}  \begin{bmatrix}
\hat{\mathbf{\Lambda}}_a \hat{\U}\tran \big(\grad \mathbf{f}(\mathbf{x}^{k}) - \grad \mathbf{f}(\bar{\x}^{k})\big) \\
\hat{\mathbf{\Lambda}}_b^{-1} \hat{\mathbf{\Lambda}}_a \hat{\U}\tran \big(\grad \mathbf{f}(\bar{\x}^{k}) -\grad \mathbf{f}(\bar{\x}^{k+1})\big)
\end{bmatrix} \right \rangle \\
&\leq \tfrac{\alpha^2}{\upsilon^2} \| \hat{\V}_l^{-1}
   \hat{\mathbf{\Lambda}}_a \hat{\U}\tran \w^k \|^2 
   +\tfrac{\alpha^2 }{\upsilon^2}\left\|  \hat{\V}^{-1}  \begin{bmatrix}
 \hat{\mathbf{\Lambda}}_a \hat{\U}\tran \big(\grad \mathbf{f}(\mathbf{x}^{k}) - \grad \mathbf{f}(\bar{\x}^{k})\big) \\
 \hat{\mathbf{\Lambda}}_b^{-1} \hat{\mathbf{\Lambda}}_a \hat{\U}\tran \big(\grad \mathbf{f}(\bar{\x}^{k}) -\grad \mathbf{f}(\bar{\x}^{k+1})\big)
\end{bmatrix} \right\|^2.
\end{align*}	
Combining the last two equations and taking conditional expectation, we have
\begin{align*}
	\Ex [\|\hat{\e}^{k+1}\|^2 | \bm{\cF}^k] &\leq \Ex \left[ \left\|\mathbf{\Gamma} \hat{\e}^{k} - \tfrac{\alpha}{\upsilon} \hat{\V}^{-1}  \begin{bmatrix}
 \hat{\mathbf{\Lambda}}_a \hat{\U}\tran \big(\grad \mathbf{f}(\mathbf{x}^{k}) - \grad \mathbf{f}(\bar{\x}^{k})\big) \\
 \hat{\mathbf{\Lambda}}_b^{-1} \hat{\mathbf{\Lambda}}_a \hat{\U}\tran \big(\grad \mathbf{f}(\bar{\x}^{k}) -\grad \mathbf{f}(\bar{\x}^{k+1})\big)
\end{bmatrix} \right\|^2 \bigg |\bm{\cF}^k \right] \\
& \quad +   \tfrac{2 \alpha^2}{\upsilon^2} \Ex [\|
 \hat{\V}_l^{-1} \hat{\mathbf{\Lambda}}_a \hat{\U}\tran \w^k   \|^2 |  \bm{\cF}^k ] -  \tfrac{2\alpha}{\upsilon}  \Ex [      \big\langle \hat{\V}_l^{-1}\hat{\mathbf{\Lambda}}_a \hat{\U}\tran \w^k , ~   \mathbf{\Gamma} \hat{\e}^{k} \big\rangle   |  \bm{\cF}^k]   \\
 & \quad +  \tfrac{\alpha^2 v_2^2}{\upsilon^2} \Ex  \left[ \left\|   \begin{bmatrix}
 \hat{\mathbf{\Lambda}}_a \hat{\U}\tran \big(\grad \mathbf{f}(\mathbf{x}^{k}) - \grad \mathbf{f}(\bar{\x}^{k})\big) \\
 \hat{\mathbf{\Lambda}}_b^{-1} \hat{\mathbf{\Lambda}}_a \hat{\U}\tran \big(\grad \mathbf{f}(\bar{\x}^{k}) -\grad \mathbf{f}(\bar{\x}^{k+1})\big)
\end{bmatrix}  \right\| \bigg|  \bm{\cF}^k \right] \\
& \leq \Ex \left[ \left\|\mathbf{\Gamma} \hat{\e}^{k} - \tfrac{\alpha}{\upsilon} \hat{\V}^{-1}  \begin{bmatrix}
\hat{\mathbf{\Lambda}}_a \hat{\U}\tran \big(\grad \mathbf{f}(\mathbf{x}^{k}) - \grad \mathbf{f}(\bar{\x}^{k})\big) \\
 \hat{\mathbf{\Lambda}}_b^{-1} \hat{\mathbf{\Lambda}}_a \hat{\U}\tran \big(\grad \mathbf{f}(\bar{\x}^{k}) -\grad \mathbf{f}(\bar{\x}^{k+1})\big)
\end{bmatrix} \right\|^2 \bigg |\bm{\cF}^k \right] 
+ \tfrac{2\alpha^2\| \hat{\V}_l^{-1}\|^2 \lambda_a^2 n\sigma^2}{\upsilon^2}   \\
 & \quad + \tfrac{\alpha^2 v_2^2}{\upsilon^2} \Ex  \left[ \left\|  \begin{bmatrix}
\hat{\mathbf{\Lambda}}_a \hat{\U}\tran \big(\grad \mathbf{f}(\mathbf{x}^{k}) - \grad \mathbf{f}(\bar{\x}^{k})\big) \\
 \hat{\mathbf{\Lambda}}_b^{-1} \hat{\mathbf{\Lambda}}_a \hat{\U}\tran \big(\grad \mathbf{f}(\bar{\x}^{k}) -\grad \mathbf{f}(\bar{\x}^{k+1})\big)
\end{bmatrix}  \right\| \bigg|  \bm{\cF}^k \right],
	\end{align*}	
where $v_2 \define \|\hat{\V}^{-1}\|$ and the last inequality holds from \eqref{noise_bound_impli} and $\|\hat{\U}\tran \| \leq 1$. Taking the expectation and expanding the last term, we get 
	\begin{align}
	\Ex \|\hat{\e}^{k+1}\|^2  &\leq \Ex  \left\|\mathbf{\Gamma} \hat{\e}^{k} - \tfrac{\alpha}{\upsilon} \hat{\V}^{-1}  \begin{bmatrix}
 \hat{\mathbf{\Lambda}}_a \hat{\U}\tran \big(\grad \mathbf{f}(\mathbf{x}^{k}) - \grad \mathbf{f}(\bar{\x}^{k})\big) \\
 \hat{\mathbf{\Lambda}}_b^{-1} \hat{\mathbf{\Lambda}}_a \hat{\U}\tran \big(\grad \mathbf{f}(\bar{\x}^{k}) -\grad \mathbf{f}(\bar{\x}^{k+1})\big)
\end{bmatrix} \right\|^2  
+ \tfrac{2\alpha^2 \| \hat{\V}_l^{-1}\|^2 \lambda_a^2 n  \sigma^2}{\upsilon^2}  \nonumber \\
 & \quad +   \tfrac{\alpha^2 v_2^2 \lambda_a^2}{\upsilon^2}   \big( \Ex \| \grad \mathbf{f}(\mathbf{x}^{k}) - \grad \mathbf{f}(\bar{\x}^{k})\|^2+ (\underline{\lambda_b}^{-1})^2 \Ex \|  \grad \mathbf{f}(\bar{\x}^{k}) -\grad \mathbf{f}(\bar{\x}^{k+1})\|^2 \big).
	\end{align}				
where we used $\|\hat{\U}\|\leq 1$ and defined $\underline{\lambda_b}^{-1} \define \|\hat{\mathbf{\Lambda}}_b^{-1}\|$. We can bound the first term on the right hand side by using the inequality $\|\a+\b\|^2 \leq \frac{1}{t} \|\a\|^2+\frac{1}{1-t}\|\b\|^2$ for any $t \in (0,1)$. Doing so with $t=\gamma=\|\mathbf{\Gamma}\|<1$, we obtain
\begin{align}
	&\Ex \|\hat{\e}^{k+1}\|^2 \nonumber \\ 
	&\leq  \gamma \Ex \|\hat{\e}^{k}\|^2  + \tfrac{\alpha^2 v_2^2 \lambda_a^2 }{(1-\gamma)\upsilon^2}  \big( \Ex \|\grad \mathbf{f}(\mathbf{x}^{k}) - \grad \mathbf{f}(\bar{\x}^{k})\|^2
	+\tfrac{1}{\underline{\lambda_b}^2} \Ex\|  \grad \mathbf{f}(\bar{\x}^{k}) -\grad \mathbf{f}(\bar{\x}^{k+1})\|^2 \big) 
	\nonumber \\
	& ~ + \tfrac{2\alpha^2 \| \hat{\V}_l^{-1}\|^2 \lambda_a^2 n  \sigma^2}{\upsilon^2}  
	+ \tfrac{\alpha^2 v_2^2 \lambda_a^2 }{\upsilon^2}  \big( \Ex \|\grad \mathbf{f}(\mathbf{x}^{k}) - \grad \mathbf{f}(\bar{\x}^{k})\|^2
	+ \tfrac{1}{\underline{\lambda_b}^2} \Ex\| \grad \mathbf{f}(\bar{\x}^{k}) -\grad \mathbf{f}(\bar{\x}^{k+1})\|^2 \big)  
	 \nonumber \\
	&\leq  \gamma \Ex \|\hat{\e}^{k}\|^2 + \tfrac{2\alpha^2 L^2 v_2^2 \lambda_a^2 }{(1-\gamma)\upsilon^2}  \Ex \|\x^{k} - \bar{\x}^{k}\|^2+  \tfrac{2\alpha^2 L^2 v_2^2 \lambda_a^2 n }{\underline{\lambda_b}^2(1-\gamma)\upsilon^2}   \Ex \| \bar{x}^{k} -\bar{x}^{k+1}\|^2 +\tfrac{2\alpha^2  v_2^2 \lambda_a^2 n  \sigma^2}{\upsilon^2}  \nonumber \\
	& \overset{\eqref{avg_e_hat_bound}}{\leq }
	\left(\gamma+ \tfrac{2\alpha^2 L^2 v_1^2 v_2^2 \lambda_a^2 }{(1-\gamma)} \right) \Ex \|\hat{\e}^{k}\|^2  +  \tfrac{2\alpha^2 L^2  v_2^2  \lambda_a^2 n }{\underline{\lambda_b}^2(1-\gamma) \upsilon^2}   \Ex \| \bar{x}^{k} -\bar{x}^{k+1}\|^2 +\tfrac{2\alpha^2  v_2^2 \lambda_a^2 n  \sigma^2}{\upsilon^2}, 
	\label{ineq_conses_proof}
	\end{align}	
	where the second inequality holds due to smoothness condition \eqref{smooth_f_eq}, $\|\hat{\V}_l^{-1}\| \leq v_2$, and $1 \leq 1/(1-\gamma)$. The second term can be bounded by using
	\begin{align*}
	\Ex [\| \bar{x}^{k} -\bar{x}^{k+1}\|^2 | \bm{\cF}^k] &\overset{\eqref{error_average_diag}}{=} \Ex [\|\alpha   \overline{\grad \f}(\x^k)+\alpha \bar{\w}^k \|^2 | \bm{\cF}^k] 	 \overset{\eqref{noise_bound_impli}}{\leq} \alpha^2 \|\overline{\grad \f}(\x^k) \|^2  
	+\tfrac{\alpha^2 \sigma^2}{n}.
	\end{align*}
	Taking the expectation and substituting the resulting bound into \eqref{ineq_conses_proof}  yields inequality \eqref{cons_ineq_lemma}.
	\end{proof}

\subsection{Proof of Theorem \ref{thm_nonconvex} (Non-convex Case)}
\label{app:thm_nonconvex_proof}
  \begin{proof}[\bfseries Proof of Theorem \ref{thm_nonconvex}]
We start by deriving an ergodic bound on the consensus iterates $\{\hat{\e}^{k}\}$. Setting $\upsilon=\sqrt{n} v_2$ in \eqref{cons_ineq_lemma}, we get
\begin{align} 
		\Ex \|\hat{\e}^{k+1}\|^2 &\leq  
	  \left(\gamma + \tfrac{2\alpha^2 L^2 v_1^2 v_2^2 \lambda_a^2 }{1-\gamma}  \right) \Ex \|\hat{\e}^{k}\|^2  + \tfrac{2\alpha^4 L^2 \lambda_a^2 }{\underline{\lambda_b}^2(1-\gamma)}  \Ex \|\overline{\grad \f}(\x^k) \|^2  +2\alpha^2 \lambda_a^2   \sigma^2+\tfrac{2\alpha^4 L^2 \lambda_a^2  \sigma^2}{\underline{\lambda_b}^2(1-\gamma)n}.
\end{align}  
If the step size $\alpha$ satisfies
\begin{align} \label{step_size_nonconv_prof_1}
\gamma + \frac{2\alpha^2 L^2  v_1^2 v_2^2 \lambda_a^2}{1-\gamma}  \leq \frac{1+\gamma}{2} ~\Rightarrow~ \alpha \leq \frac{1-\gamma}{2L v_1v_2 \lambda_a},
\end{align}
then the previous bound can be upper bounded by
\begin{align} 
		\Ex \|\hat{\e}^{k+1}\|^2 &\leq  
	  \tfrac{1+\gamma}{2} \Ex \|\hat{\e}^{k}\|^2  + \tfrac{2\alpha^4 L^2 \lambda_a^2 }{\underline{\lambda_b}^2(1-\gamma)}  \left( \Ex \| \grad f(\bar{x}^{k})  \|^2+ \Ex \|\overline{\grad \f}(\x^k) \|^2 \right)  +2\alpha^2 \lambda_a^2   \sigma^2+\tfrac{2\alpha^4 L^2 \lambda_a^2  \sigma^2}{\underline{\lambda_b}^2(1-\gamma)n}.
\end{align}  
Recursively applying the previous inequality, it holds (for any $k=1,2,\dots$) that 
\begin{align} \label{cons_ineq_nonconvex_proof}
		\Ex \|\hat{\e}^{k}\|^2 &\leq  
	  \left(\tfrac{1+\gamma}{2}\right)^{k}  \|\hat{\e}^{0}\|^2   +\sum_{\ell=0}^{k-1} \left( \tfrac{1+\gamma}{2} \right)^{k-1-\ell}  \left( 2\alpha^2 \lambda_a^2   \sigma^2+\tfrac{2\alpha^4 L^2 \lambda_a^2  \sigma^2}{\underline{\lambda_b}^2(1-\gamma)n} \right) \nonumber \\
	  & \quad + \tfrac{2\alpha^4 L^2 \lambda_a^2 }{\underline{\lambda_b}^2(1-\gamma)} \sum_{\ell=0}^{k-1} \left( \tfrac{1+\gamma}{2} \right)^{k-1-\ell} \left( \Ex \| \grad f(\bar{x}^{\ell})  \|^2 + \Ex \|\overline{\grad \f}(\x^\ell) \|^2 \right) \nonumber \\
	  &\leq  
	  \left(\tfrac{1+\gamma}{2}\right)^{k}  \|\hat{\e}^{0}\|^2  
	 +\tfrac{2}{1-\gamma}  \left( 2\alpha^2 \lambda_a^2   \sigma^2+\tfrac{2\alpha^4 L^2 \lambda_a^2  \sigma^2}{\underline{\lambda_b}^2(1-\gamma)n} \right)
	   \nonumber \\
	  & \quad   + \tfrac{2\alpha^4 L^2 \lambda_a^2 }{\underline{\lambda_b}^2(1-\gamma)} \sum_{\ell=0}^{k-1} \left( \tfrac{1+\gamma}{2} \right)^{k-1-\ell} \left( \Ex \| \grad f(\bar{x}^{\ell})  \|^2 + \Ex \|\overline{\grad \f}(\x^\ell) \|^2 \right),
\end{align}  
where in the last inequality we used \eqref{sum_bound}.  Taking the average of both sides of \eqref{cons_ineq_nonconvex_proof} over $k=1,2\dots,K$, and using \eqref{sum_bound}, we have
  \begin{align} 
	\frac{1}{K} \sum_{k=1}^K	\Ex \|\hat{\e}^{k}\|^2 	  &\leq  
	 \frac{2 \|\hat{\e}^{0}\|^2}{(1-\gamma)K} 
	 +  \frac{4\alpha^2 \lambda_a^2   \sigma^2}{1-\gamma}+\frac{4\alpha^4 L^2 \lambda_a^2  \sigma^2}{\underline{\lambda_b}^2(1-\gamma)^2n}  \nonumber \\
	  & \quad  + \frac{2\alpha^4 L^2 \lambda_a^2 }{\underline{\lambda_b}^2(1-\gamma)K} \sum_{k=1}^K \sum_{\ell=0}^{k-1} \left( \tfrac{1+\gamma}{2} \right)^{k-1-\ell} \left( \Ex \| \grad f(\bar{x}^{\ell})  \|^2 + \Ex \|\overline{\grad \f}(\x^\ell) \|^2 \right) \nonumber \\
	  & \leq  
	 \frac{2 \|\hat{\e}^{0}\|^2}{(1-\gamma)K} 
	 +  \frac{4\alpha^2 \lambda_a^2   \sigma^2}{1-\gamma}+\frac{4\alpha^4 L^2 \lambda_a^2  \sigma^2}{\underline{\lambda_b}^2(1-\gamma)^2n}  \nonumber \\
	  & \quad  + \frac{4\alpha^4 L^2 \lambda_a^2 }{\underline{\lambda_b}^2(1-\gamma)^2 K} \sum_{k=0}^{K-1} \left( \Ex \| \grad f(\bar{x}^{k})  \|^2 + \Ex \|\overline{\grad \f}(\x^k) \|^2 \right),
\end{align}  
where in the last inequality we used the bound \eqref{sum_bound_2} on the last term. Adding $\frac{\|\hat{\e}^{0}\|^2}{(1-\gamma)K}$ to both sides of the previous inequality and using $\frac{\|\hat{\e}^{0}\|^2}{K} \leq \frac{\|\hat{\e}^{0}\|^2}{(1-\gamma)K}$, we get
  \begin{align}  \label{bound_nonconv_cons_final}
	\frac{1}{K} \sum_{k=0}^{K-1}	\Ex \|\hat{\e}^{k}\|^2  	  &\leq  
	 \frac{3 \|\hat{\e}^{0}\|^2}{(1-\gamma)K} 
	 +  \frac{4\alpha^2 \lambda_a^2   \sigma^2}{1-\gamma}+\frac{4\alpha^4 L^2 \lambda_a^2  \sigma^2}{\underline{\lambda_b}^2(1-\gamma)^2n}  \nonumber \\
	  & \quad  + \frac{4\alpha^4 L^2 \lambda_a^2 }{\underline{\lambda_b}^2(1-\gamma)^2 K}  \sum_{k=0}^{K-1} \left( \Ex \| \grad f(\bar{x}^{k})  \|^2 + \Ex \|\overline{\grad \f}(\x^k) \|^2 \right).
\end{align}  
Inequality \eqref{bound_nonconv_cons_final} will be used in the upcoming bound \eqref{sum_grad_ineq_nonconvex_proof}, which will be established next.
Subtracting $f^\star$ from both sides of \eqref{descent_inquality}, letting $\upsilon=\sqrt{n} v_2$, using $1/2 \geq 1/4$, and rearranging  gives 
	\begin{align}
	  \Ex  \| \grad f(\bar{x}^{k})  \|^2 +   \Ex  \| \overline{\grad\f}(\x^k)\|^2 		  &\leq \tfrac{4}{\alpha} \left( \Ex  \tilde{f}(\bar{x}^{k}) -  \Ex   \tilde{f}(\bar{x}^{k+1}) \right)
	  +2  L^2 v_1^2 v_2^2    \Ex  \|\hat{\e}^k\|^2+	\tfrac{ 2\alpha L \sigma^2 }{ n},
	\end{align}	
	where $\tilde{f}(\bar{x}^{k})\define f(\bar{x}^{k})-f^\star$. Summing over $k=0,\ldots,K-1$, dividing by $K \geq 1$, and using $-\tilde{f}(\bar{x}^{k})\leq 0$, it holds that
\begin{align} \label{sum_grad_ineq_nonconvex_proof}
	\frac{1}{K}	 \sum_{k=0}^{K-1}  \Ex  \| \grad f(\bar{x}^{k})  \|^2 +   \Ex  \| \overline{\grad\f}(\x^k)\|^2 		  &\leq \frac{4   \tilde{f}(\bar{x}^{0})}{\alpha K}	  
	  +	\frac{ 2\alpha L \sigma^2 }{ n}
	  + \frac{2  L^2 v_1^2 v_2^2  }{K}	 \sum_{k=0}^{K-1}   \Ex  \|\hat{\e}^k\|^2.
	\end{align}		
Substituting inequality \eqref{bound_nonconv_cons_final} into \eqref{sum_grad_ineq_nonconvex_proof} and rearranging, we obtain 
\begin{align} 
	&\left(1-\frac{8\alpha^4 L^4 v_1^2 v_2^2 \lambda_a^2   }{\underline{\lambda_b}^2(1-\gamma)^2}  \right)	\frac{1}{K}  \sum_{k=0}^{K-1}  \Ex  \| \grad f(\bar{x}^{k})  \|^2 +   \Ex  \| \overline{\grad\f}(\x^k)\|^2  \nonumber \\
	& \leq \frac{4   \tilde{f}(\bar{x}^{0})}{\alpha K}	  
	+	\frac{ 2\alpha L \sigma^2 }{ n}
	  + 	\frac{6 L^2 v_1^2 v_2^2   \|\hat{\e}^{0}\|^2}{(1-\gamma)K} 
	 +  \frac{8\alpha^2   L^2 v_1^2 v_2^2   \lambda_a^2  \sigma^2}{1-\gamma}+\frac{8\alpha^4 L^4 v_1^2 v_2^2  \lambda_a^2  \sigma^2}{ \underline{\lambda_b}^2(1-\gamma)^2n} .
	\end{align}	
	 If we set 
\begin{align} \label{step_size_nonconv_prof_2}
\frac{1}{2} \leq 1-\frac{8\alpha^4 L^4 v_1^2 v_2^2 \lambda_a^2  }{\underline{\lambda_b}^2(1-\gamma)^2}  \Longrightarrow \alpha \leq \frac{\sqrt{\underline{\lambda_b}(1-\gamma)}}{2 L \sqrt{v_1 v_2 \lambda_a }}, 
\end{align}	
	then we find
	\begin{align} \label{non_proof_last0}
	&\frac{1}{K}	 \sum_{k=0}^{K-1}  \Ex  \| \grad f(\bar{x}^{k})  \|^2 +   \Ex  \| \overline{\grad\f}(\x^k)\|^2 		 \nonumber \\
	 &\leq 
	  	  \frac{8   \tilde{f}(\bar{x}^{0})}{\alpha K}	  
	 	   + 	\frac{ 4\alpha L \sigma^2 }{ n} 
	 	   +   \frac{12  L^2 v_1^2 v_2^2    \|\hat{\e}^{0}\|^2}{(1-\gamma)K} 
	 +  \frac{16\alpha^2   L^2 v_1^2 v_2^2   \lambda_a^2  \sigma^2}{1-\gamma}+\frac{16\alpha^4 L^4 v_1^2 v_2^2  \lambda_a^2  \sigma^2}{ \underline{\lambda_b}^2(1-\gamma)^2n}.
	\end{align}
 Using \eqref{e_0_bound} in the above inequality, we arrive at \eqref{eq:thm:nonconvex}.  	The step size conditions in the theorem follows from $\alpha \leq \frac{1}{2L}$ (from Lemma \ref{lemma_descent_inq}) and the conditions \eqref{step_size_nonconv_prof_1}  and \eqref{step_size_nonconv_prof_2}.

  \end{proof}

\subsection{Proof of Theorem \ref{thm_PL_linear_conv} (PL Condition Case)}
\label{app:thm_pl_lin_conv_proof}
To prove Theorem \ref{thm_PL_linear_conv}, we need the following result.
\begin{lemma}
Under Assumption \ref{assump:smoothness}, the following holds:
\begin{align} \label{polyak_bound}
\|\grad f(x)\|^2 &\leq 2 L \big(f(x)-f^\star\big), \quad \forall~ x \in \real^d \\
\frac{1}{n}\sum_{i=1}^n (f(x_i^{k})-f^\star) &\leq 2 \big(f(\bar{x}^k)-f^\star\big)+\frac{L \upsilon^2 v_1^2}{n} \|\hat{\e}^k \|, \quad \forall ~k.  \label{f_PL_bound}
\end{align}
\end{lemma}
\begin{proof}
 Inequality  \eqref{polyak_bound} is proven in \cite{polyak1987introduction}. We include the proof here for convenience. Since $f^\star \leq f(z)$ for any $z \in \real^d$, we have
 \begin{align}
f^\star &\leq f\big(x-\tfrac{1}{L} \grad f(x)\big)
\leq f(x)-\tfrac{1}{2L} \|\grad f(x)\|^2,
	 \end{align}
 where the last inequality holds due to \eqref{bound:L_smooth_function} with $y=x-\tfrac{1}{L} \grad f(x)$ and $z=x$. Rearranging the above we arrive at \eqref{polyak_bound}. We now establish inequality \eqref{f_PL_bound}. The argument adjust the proof  \cite[Lemma 13]{xin2021improved} to our case. Substituting $y=x^k_i$ and $z=\bar{x}^k$ in \eqref{bound:L_smooth_function}, we get
 \begin{align*}
	      f(x^k_i) &\leq f(\bar{x}^k)+ \big \langle \grad f(\bar{x}^k), x^k_i-\bar{x}^k \big \rangle+\tfrac{L}{2} \|x^k_i-\bar{x}^k\|^2.
	 \end{align*}
	Using Cauchy-Schwarz and Young's inequalities, we have
\begin{align*}
\big \langle \grad f(\bar{x}^k), x^k_i-\bar{x}^k \big \rangle \leq \| \grad f(\bar{x}^k)\| \|x^k_i-\bar{x}^k \| &\leq \tfrac{1}{2L} \| \grad f(\bar{x}^k)\|^2 + \tfrac{L}{2}\|x^k_i-\bar{x}^k \|^2 \\
& \overset{\eqref{polyak_bound}}{\leq} f(\bar{x}^k)-f^\star + \tfrac{L}{2}\|x^k_i -\bar{x}^k\|^2.
\end{align*}
Combining the last two bounds, then subtracting $f^\star$ from both sides and averaging over $i$, we get	 
\begin{align*}
\frac{1}{n}\sum_{i=1}^n \big(f(x_i^{k})-f^\star\big) &\leq 2 \big(f(\bar{x}^k)-f^\star\big)+\tfrac{L}{n} \|\x^{k}- \bar{\x}^{k} \|. 
\end{align*}	 
	 Using \eqref{avg_e_hat_bound}, we arrive at \eqref{f_PL_bound}.
\end{proof}
\begin{proof}[\bfseries Proof of Theorem \ref{thm_PL_linear_conv}]
 Setting $\upsilon=\sqrt{n} v_2$ in \eqref{descent_inquality}, we get
		\begin{align} \label{des_ineq_pl_proof}	
			  \Ex   f(\bar{x}^{k+1}) &\leq  \Ex  f(\bar{x}^{k}) - \tfrac{\alpha}{2}  \Ex  \| \grad f(\bar{x}^{k})  \|^2 - \tfrac{\alpha}{4}  \Ex  \| \overline{\grad\f}(\x^k)\|^2 
	  +\tfrac{\alpha L^2  v_1^2 v_2^2 }{2}  \Ex  \|\hat{\e}^k\|^2+	\tfrac{ \alpha^2 L \sigma^2 }{2 n}.
	\end{align}	
 Subtracting $f^\star$ from both sides of \eqref{des_ineq_pl_proof}, using the PL inequality \eqref{PL_cond} and the bound $- \tfrac{\alpha}{4}  \Ex  \| \overline{\grad\f}(\x^k)\|^2  \leq 0$, we get
	\begin{equation} \label{pl_average_error_inq}	
	\begin{aligned}
	  \Ex   \tilde{f}(\bar{x}^{k+1})  &\leq  (1-\mu \alpha ) \Ex \tilde{f}(\bar{x}^{k}) + \tfrac{\alpha  L^2 v_1^2 v_2^2  }{2}  \Ex \| 	\hat{\e}^{k} \|^2+		\tfrac{ \alpha^2 L \sigma^2 }{2 n},
	\end{aligned}	
	\end{equation}
	where $\tilde{f}(\bar{x}^k)=f(\bar{x}^{k})-f^\star$. To establish our result, we will combine the above bound with a later bound given in \eqref{pl_consensus_error_inq}, which we establish next.	Letting $\upsilon=\sqrt{n} v_2$ in  \eqref{cons_ineq_lemma}, gives	
\begin{align} \label{cons_ineq_pl_proof}
		\Ex \|\hat{\e}^{k+1}\|^2 &\leq  
	  \left(\gamma + \tfrac{2\alpha^2 L^2 v_1^2 v_2^2 \lambda_a^2 }{1-\gamma}  \right) \Ex \|\hat{\e}^{k}\|^2  + \tfrac{2\alpha^4 L^2 \lambda_a^2 }{\underline{\lambda_b}^2(1-\gamma)}  \Ex \|\overline{\grad \f}(\x^k) \|^2  +2\alpha^2 \lambda_a^2   \sigma^2+\tfrac{2\alpha^4 L^2 \lambda_a^2  \sigma^2}{\underline{\lambda_b}^2(1-\gamma)n}.
\end{align}
	Note that
	\begin{align*}
 \|\overline{\grad \f}(\x^k) \|^2 &=\|\overline{\grad \f}(\x^k)-\grad f(\bar{x}^k)+\grad f(\bar{x}^k) \|^2
	 \leq 2 \|\overline{\grad \f}(\x^k)-\grad f(\bar{x}^k) \|^2+2  \|\grad f(\bar{x}^k) \|^2  
	\\
	& \overset{\eqref{polyak_bound}}{\leq} 2 \|\overline{\grad \f}(\x^k)-\grad f(\bar{x}^k) \|^2+4 L  \tilde{f}(\bar{x}^{k})\overset{\eqref{lipschitz_bound_average}}{\leq}   2 v_1^2 v_2^2 L^2 \|\hat{\e}^k \|^2+4 L  \tilde{f}(\bar{x}^{k}).
	\end{align*}
Substituting the above into \eqref{cons_ineq_pl_proof}, we get
		\begin{align} \label{cons_ineq_pl_proof2}
		\Ex \|\hat{\e}^{k+1}\|^2 &\leq  
	   \tfrac{8\alpha^4 L^3 \lambda_a^2 }{\underline{\lambda_b}^2(1-\gamma)}  \Ex \tilde{f}(\bar{x}^k) 
	   + \left(\gamma + \tfrac{2\alpha^2 L^2 v_1^2 v_2^2 \lambda_a^2   }{1-\gamma}+\tfrac{ 4 \alpha^4 L^4    v_1^2 v_2^2 \lambda_a^2 }{\underline{\lambda_b}^2(1-\gamma)}  \right) \Ex \|\hat{\e}^{k}\|^2 
	   \nonumber \\
	   & \quad  +2\alpha^2 \lambda_a^2   \sigma^2+\tfrac{2\alpha^4 L^2 \lambda_a^2  \sigma^2}{\underline{\lambda_b}^2(1-\gamma)n}.
\end{align}
 To simplify later expressions, we choose $\alpha$ such that
\begin{subequations}
\begin{align}
\frac{2\alpha^2 L^2 v_1^2 v_2^2 \lambda_a^2   }{1-\gamma}+\frac{ 4 \alpha^4 L^4 v_1^2 v_2^2 \lambda_a^2   }{\underline{\lambda_b}^2(1-\gamma)}  &\leq  \frac{3\alpha^2 L^2 v_1^2 v_2^2 \lambda_a^2 }{1-\gamma}, \\
	\gamma+\frac{3\alpha^2 L^2 v_1^2 v_2^2 \lambda_a^2 }{1-\gamma}
	 & \leq  \frac{1+\gamma}{2}, \\
	  \frac{8\alpha^4 L^3 \lambda_a^2 }{\underline{\lambda_b}^2(1-\gamma)}    &\leq \frac{\alpha \mu \lambda_a^2}{L v_1^2 v_2^2}.
\end{align}
\end{subequations}
The above inequalities are satisfied if 
\begin{align} \label{stepsize_pl_proof_1}
\alpha \leq \min\left\{\frac{ \underline{\lambda_b}}{2  L },~
\frac{1-\gamma}{\sqrt{6} L v_1 v_2 \lambda_a},~
\left(\frac{\mu \underline{\lambda_b}^{2}(1-\gamma)}{8 L^{4} v_1^{2} v_2^{2} } \right)^{1/3}
\right\}.
\end{align}
Under the previous conditions on $\alpha$,  inequality \eqref{cons_ineq_pl_proof2} is upper bounded by
	\begin{align} \label{pl_consensus_error_inq}
		\Ex \|\hat{\e}^{k+1}\|^2 &\leq  
	   \tfrac{\alpha \mu \lambda_a^2}{L v_1^2 v_2^2} \Ex \tilde{f}(\bar{x}^k) 
	   + \left( \tfrac{1+\gamma}{2}  \right) \Ex \|\hat{\e}^{k}\|^2 
	    	    + 2\alpha^2 \lambda_a^2   \sigma^2+\tfrac{2\alpha^4 L^2 \lambda_a^2  \sigma^2}{\underline{\lambda_b}^2(1-\gamma)n}.
\end{align}
Note that from \eqref{f_PL_bound} with $\upsilon=\sqrt{n} v_2$, we have
\begin{align} \label{tran_relation}
 \frac{1}{n}\sum_{i=1}^n \Ex [f(x_i^{k})-f^\star] &\leq  2 \Ex [f(\bar{x}^k)-f^\star]+L  v_1^2 v_2^2 \Ex \|\hat{\e}^k \|. 
\end{align}
Multiplying inequality \eqref{pl_average_error_inq} by $2$ and inequality \eqref{pl_consensus_error_inq} by $L v_1^2 v_2^2$ and rewriting them in matrix form, we have
\begin{align} \label{linear_dynamical_error2}
\begin{bmatrix} 2 \Ex \tilde{f}(\bar{x}^{k+1}) \\
L v_1^2 v_2^2 \Ex \|\hat{\e}^{k+1}\|^2 
\end{bmatrix}
 \leq 
\underbrace{\begin{bmatrix}
1-\mu \alpha  \vspace{0.5mm} 		 
&
 \alpha  L   \vspace{0.5mm} \\
\tfrac{\alpha \mu \lambda_a^2}{2} 
& 
 \tfrac{1+\gamma}{2}
\end{bmatrix}}_{\define H}
\begin{bmatrix} 2\Ex \tilde{f}(\bar{x}^{k})  \\
L v_1^2 v_2^2 \Ex \|\hat{\e}^{k}\|^2 
\end{bmatrix} 
+ \underbrace{\begin{bmatrix}
\frac{ \alpha^2 L \sigma^2 }{ n} \\
 2\alpha^2 L v_1^2 v_2^2 \lambda_a^2   \sigma^2
 + \tfrac{2\alpha^4 L^3 v_1^2 v_2^2 \lambda_a^2  \sigma^2}{\underline{\lambda_b}^2(1-\gamma)n}
\end{bmatrix}}_{\define h}.
\end{align}
We will establish the convergence of $\frac{1}{n}\sum_{i=1}^n \Ex [f(x_i^{k})-f^\star]$ through the convergence of \eqref{linear_dynamical_error2}. Note that
\begin{align} \label{rho_H}
\rho(H) \leq \|H\|_1= \max \left\{
1-\mu \alpha
+ \tfrac{\alpha \mu \lambda_a^2}{2}
, ~
\tfrac{1+\gamma}{2}+\alpha L
\right\} \leq 1-\tfrac{\mu \alpha}2 .
\end{align}
where the  last inequality holds under the  step size condition:
\begin{align} \label{stepsize_pl_proof_2}
\alpha \leq 
\frac{1-\gamma}{3L}  \leq \frac{1-\gamma}{2L+\mu}.
\end{align}
Hence, $\rho(H) <1$ and iterating inequality \eqref{linear_dynamical_error2}, we get
\begin{align} 
\begin{bmatrix} 2 \Ex \tilde{f}(\bar{x}^{k}) \\
L v_1^2 v_2^2 \Ex \|\hat{\e}^{k}\|^2 
\end{bmatrix}
 & \leq  
H^k 
\begin{bmatrix} 2\Ex \tilde{f}(\bar{x}^{0})   \\
L v_1^2 v_2^2 \Ex \|\hat{\e}^{0}\|^2 
\end{bmatrix} 
+ \sum_{\ell=0}^{k-1} H^\ell h \nonumber \\ 
 & \leq  
H^k 
\begin{bmatrix} 2\Ex \tilde{f}(\bar{x}^{0})  \\
L v_1^2 v_2^2 \Ex \|\hat{\e}^{0}\|^2 
\end{bmatrix} 
+(I- H)^{-1}h.
\end{align}
Taking the $1$-norm, using \eqref{tran_relation}, and using  the submultiplicative and triangle inequality  properties of the 1-norm, it holds that
\begin{align} \label{tran_PL_0}
\frac{1}{n}\sum_{i=1}^n \Ex [f(x_i^{k})-f^\star]
& \leq  
\|H^k\|_1  c_0
+ \left\|(I- H)^{-1} h \right\|_1 \leq  
\|H\|^k_1 c_0
+ \left\|(I- H)^{-1}h \right\|_1.
\end{align}
where $c_0 =2\Ex \tilde{f}(\bar{x}^{0}) +
L v_1^2 v_2^2 \Ex \|\hat{\e}^{0}\|^2 $.  We now bound the last term on the right.  To do that, we  compute  the inverse:
\begin{align*}
(I-H)^{-1} &= \begin{bmatrix}
\mu \alpha  \vspace{0.5mm} 		 
&
 -\alpha L \vspace{0.5mm} \\
-\frac{\alpha \mu \lambda_a^2}2 
& 
 \tfrac{1-\gamma}{2}
\end{bmatrix}^{-1} 
=
\frac{1}{\det(I-H)}
 \begin{bmatrix}
 \tfrac{1-\gamma}{2} \vspace{0.5mm} 		 
&
 \alpha L \vspace{0.5mm} \\
\frac{\alpha \mu \lambda_a^2}2
& 
\mu \alpha 
\end{bmatrix},
\end{align*}
where $\det(I-H)$ is the determinant of $I-H$:
\begin{align*}
\det(I-H)= \tfrac{1}{2} \big(\alpha \mu(1-\gamma) -  \alpha^2 \mu L \lambda_a^2 \big)=\tfrac{\alpha \mu}{2} \big( 1-\gamma -  \alpha L \lambda_a^2 \big).
\end{align*}
Hence,
\begin{align}
(I- H)^{-1} h 
&=
\frac{2}{\alpha \mu (1-\gamma -  \alpha L \lambda_a^2) }
 \begin{bmatrix}
 \tfrac{1-\gamma}{2} \vspace{0.5mm} 		 
&
 \alpha L \vspace{0.5mm} \\
\frac{\alpha^2 \mu \lambda_a^2}{2}
& 
\mu \alpha 
\end{bmatrix} 
\begin{bmatrix}
\frac{ \alpha^2 L \sigma^2 }{ n} \\
 2\alpha^2 L v_1^2 v_2^2 \lambda_a^2   \sigma^2
 + \tfrac{2\alpha^4 L^3 v_1^2 v_2^2 \lambda_a^2  \sigma^2}{\underline{\lambda_b}^2(1-\gamma)n}
\end{bmatrix}
\nonumber \\
&=  
\begin{bmatrix}
 \dfrac{  \alpha L \sigma^2 (1-\gamma)/n +
 4\alpha^2 L^2 v_1^2 v_2^2  \lambda_a^2   \sigma^2 +4 \alpha^4 L^4 v_1^2 v_2^2 \lambda_a^2  \sigma^2/(\underline{\lambda_b}^2(1-\gamma)n) }{  \mu (1-\gamma -  \alpha L \lambda_a^2) } \vspace{2mm}
 \\
\dfrac{ \alpha^3 \mu L  \lambda_a^2 \sigma^2/n  +
 4 \alpha^2 \mu L v_1^2 v_2^2   \lambda_a^2   \sigma^2
  +4 \alpha^4 \mu L^3 v_1^2 v_2^2 \lambda_a^2  \sigma^2/(\underline{\lambda_b}^2(1-\gamma)n) }{ \mu (1-\gamma -  \alpha L \lambda_a^2) }
 \end{bmatrix}.
\end{align}
Note that since $\lambda_a \leq 1$; hence, under step size condition \eqref{stepsize_pl_proof_2}, we have
\begin{align} 
 1-\gamma -  \alpha L \lambda_a^2 \geq  (1-\gamma)/2.
\end{align}
Combining the last two equations and using $L \geq \mu$, we get
\begin{align*}
  & \left\|(I- H)^{-1} h \right\|_1 \nonumber \\
& \leq     \dfrac{ \alpha L \sigma^2 (1-\gamma) /n +
 8\alpha^2 L^2 v_1^2 v_2^2  \lambda_a^2   \sigma^2 +8 \alpha^4 L^4 v_1^2 v_2^2 \lambda_a^2  \sigma^2/(\underline{\lambda_b}^2(1-\gamma)n) }{  \mu (1-\gamma)/2} \vspace{2mm}
 +
\dfrac{ 2 \alpha^3  L  \lambda_a^2 \sigma^2
}{ n(1-\gamma) } \\
 &=     \dfrac{ 2\alpha L \sigma^2}{\mu n} +\dfrac{
 16 \alpha^2 L^2 v_1^2 v_2^2  \lambda_a^2   \sigma^2}{\mu (1-\gamma)} +\dfrac{16 \alpha^4 L^4 v_1^2 v_2^2 \lambda_a^2  \sigma^2}{  \mu \underline{\lambda_b}^2(1-\gamma)^2 n } \vspace{2mm}
 +
\dfrac{ 2 \alpha^3  L  \lambda_a^2 \sigma^2
}{ n(1-\gamma) } \\
 & = O\left(\frac{\alpha L \sigma^2}{\mu n}
 +\frac{\alpha^2 L^2 v_1^2 v_2^2 \lambda_a^2  \sigma^2}{\mu (1-\gamma)} +
\dfrac{ \alpha^4 L^4 v_1^2 v_2^2 \lambda_a^2  \sigma^2}{  \mu \underline{\lambda_b}^2(1-\gamma)^2 n } \right) .
\end{align*}
Substituting the above into \eqref{tran_PL_0} and using \eqref{rho_H}, we obtain
\begin{align} \label{tran_PL}
\tfrac{1}{n}\sum_{i=1}^n \Ex \tilde{f}(x_i^{k})
& \leq  
(1-\tfrac{\alpha \mu}{2})^k c_0
+O\left(\frac{\alpha L \sigma^2}{\mu n}
+ \frac{\alpha^2 L^2 v_1^2 v_2^2 \lambda_a^2  \sigma^2}{\mu (1-\gamma)}
+ \frac{ \alpha^4 L^4 v_1^2 v_2^2 \lambda_a^2  \sigma^2}{  \mu \underline{\lambda_b}^2(1-\gamma)^2 n } \right).
\end{align}
 Note that if $\x^0=\one \otimes x^0$, then
\begin{align}
    c_0 =2\Ex \tilde{f}(\bar{x}^{0}) +
L v_1^2 v_2^2 \Ex \|\hat{\e}^{0}\|^2 \overset{\eqref{e_0_bound}}{\leq} 2\Ex \tilde{f}(\bar{x}^{0}) +
\frac{\alpha^2 L v_1^2 v_2^2 \zeta_0^2}{ \underline{\lambda_b}^2}.
\end{align}
Combining the last two inequalities we get \eqref{eq:PL_thm_1}; moreover, combining the step size conditions \eqref{stepsize_pl_proof_1} and \eqref{stepsize_pl_proof_2}, we get the step size condition in Theorem \ref{thm_PL_linear_conv}.

\end{proof}

\end{document}